\documentclass[11pt]{article}
\usepackage[margin=1in]{geometry}
\usepackage{mathtools}
\usepackage{enumitem}
\usepackage{amsmath}
\usepackage{amssymb}
\usepackage{amsthm}
\usepackage{etoolbox}
\usepackage{microtype}
\usepackage{titlesec}
\usepackage[longnamesfirst]{natbib}
\usepackage{xcolor}
\usepackage{proba}
\usepackage[hypertexnames=false]{hyperref}
\usepackage{algorithm}
\usepackage{algpseudocode}
\usepackage{multicol,multirow}
\usepackage{booktabs}
\usepackage{threeparttable}
\usepackage{subcaption}

\hypersetup{colorlinks=true,
  linkcolor=blue,urlcolor=blue,citecolor=blue}
\DeclareMathOperator*{\argmin}{arg\,min}

\DeclareMathOperator{\diam}{diam}

\renewcommand{\P}{\ensuremath{\mathbb{P}}}
\renewcommand{\T}{\ensuremath{\mathsf{T}}}

\newcommand{\rF}{\ensuremath{\mathrm{F}}}

\newcommand{\RSS}{\ensuremath{\mathsf{RSS}\hspace*{0.2mm}}}

\newcommand{\cA}{\ensuremath{\mathcal{A}}}
\newcommand{\cB}{\ensuremath{\mathcal{B}}}
\newcommand{\cI}{\ensuremath{\mathcal{I}}}

\newcommand{\bs}{\boldsymbol}

\newcommand{\cd}{\overset{d}{\longrightarrow}}
\newcommand{\cp}{\overset{\P}{\longrightarrow}}

\def\1{\ensuremath{\mathbf{1}}}

\newlist{inlineroman}{enumerate*}{1}
\setlist[inlineroman]{afterlabel=~,label=(\roman*)}

\DeclareFontFamily{U}{mathx}{\hyphenchar\font45}
\DeclareFontShape{U}{mathx}{m}{n}{
  <5> <6> <7> <8> <9> <10>
  <10.95> <12> <14.4> <17.28> <20.74> <24.88>
  mathx10
}{}
\DeclareSymbolFont{mathx}{U}{mathx}{m}{n}
\DeclareFontSubstitution{U}{mathx}{m}{n}
\DeclareMathAccent{\widecheck}{0}{mathx}{"71}
\DeclareMathAccent{\wideparen}{0}{mathx}{"75}

\newcommand{\vvvert}{{
    \vert\kern-0.25ex\vert\kern-0.25ex\vert}}
\newcommand{\bigvvvert}{{
    \big\vert\kern-0.35ex\big\vert\kern-0.35ex\big\vert}}
\newcommand{\Bigvvvert}{{
    \Big\vert\kern-0.3ex\Big\vert\kern-0.3ex\Big\vert}}
\newcommand{\biggvvvert}{{
    \bigg\vert\kern-0.25ex\bigg\vert\kern-0.25ex\bigg\vert}}

\titleformat
{\paragraph} 
[hang] 
{\bfseries\upshape} 
{} 
{0pt} 
{} 
[] 

\titlespacing*{\paragraph}{0pt}{6pt}{0pt}

\newcounter{proofparagraphcounter}

\newtheoremstyle{break}%
  {\topsep}{\topsep}%
  {\itshape}{}%
  {\bfseries}{}%
  {\newline}%
  {\thmname{#1}\thmnumber{ #2}%
  \thmnote{ \rm(#3)\addcontentsline{toc}{subsubsection}{{\it#1 #2}\/: #3}}}

\newtheoremstyle{breakplainproof}
  {\topsep}{\topsep}%
  {}{}%
  {\bfseries}{}%
  {\newline}%
  {\thmname{#1}\thmnumber{ #2}%
  \thmnote{ \rm(#3)\addcontentsline{toc}{subsubsection}{{\it#1}\/: #3}}}

\theoremstyle{break}
\newtheorem{theorem}{Theorem}
\newtheorem{lemma}{Lemma}
\newtheorem{assumption}{Assumption}

\newtheorem{example}{Example}

\theoremstyle{breakplainproof}

\newtheorem*{proof}{Proof}
\AtBeginEnvironment{proof}{\setcounter{proofparagraphcounter}{0}}%
\AtEndEnvironment{proof}{\null\hfill$\square$}%

\theoremstyle{remark}
\newtheorem*{remark}{Remark}

\begin{document}

\title{Balancing Flexibility and Interpretability: A Conditional Linear Model Estimation via Random Forest}

\author{
  Ricardo Masini \textsuperscript{1*} \and 
  Marcelo Medeiros\textsuperscript{2}
}

\maketitle

\footnotetext[1]{
  Department of Statistics,
  University of California, Davis.
}
\footnotetext[2]{
  Department of Economics,
  University of Illinois, Urbana-Champaign.
}
\let\thefootnote\relax
\footnotetext[1]{
  \textsuperscript{*}Corresponding author:
  \href{mailto:rmasini@ucdavis.edu}{\texttt{rmasini@ucdavis.edu}}
}
\newcommand{\thefootnote}{\arabic{footnote}}

\setcounter{page}{0}\thispagestyle{empty}

\vspace{-0.5cm}
\noindent
\begin{abstract}
Traditional parametric econometric models often rely on rigid functional forms, while nonparametric techniques, despite their flexibility, frequently lack interpretability. This paper proposes a parsimonious alternative by modeling the outcome $Y$ as a linear function of a vector of variables of interest $\bs{X}$, conditional on additional covariates $\bs{Z}$. Specifically, the conditional expectation is expressed as $\E[Y|\bs{X},\bs{Z}]=\bs{X}^{\T}\bs\beta(\bs{Z})$, where $\bs\beta(\cdot)$ is an unknown Lipschitz-continuous function. We introduce an adaptation of the Random Forest (RF) algorithm to estimate this model, balancing the flexibility of machine learning methods with the interpretability of traditional linear models. This approach addresses a key challenge in applied econometrics by accommodating heterogeneity in the relationship between covariates and outcomes. Furthermore, the heterogeneous partial effects of $\bs{X}$ on $Y$ are represented by $\bs\beta(\cdot)$ and can be directly estimated using our proposed method. Our framework effectively unifies established parametric and nonparametric models, including varying-coefficient, switching regression, and additive models. We provide theoretical guarantees, such as pointwise and $\mathcal{L}^p$-norm rates of convergence for the estimator, and establish a pointwise central limit theorem through subsampling, aiding inference on the function $\bs\beta(\cdot)$. We present Monte Carlo simulation results to assess the finite-sample performance of the method.

\vspace{10pt}
\noindent\textbf{Keywords}:
random forests; heterogeneous partial effects; machine learning.
\end{abstract}

\maketitle

\clearpage
\pagebreak

\tableofcontents
\pagebreak

\def\spacingset#1{\renewcommand{\baselinestretch}%
{#1}\small\normalsize} \spacingset{1}

\spacingset{1.5}

\section{Introduction}

In economics and related social sciences, epidemiology and medicine, psychology, and many other areas, estimating the partial effects (causal or not) of one or more factors on a target variable is extremely important. In general, partial effects estimation is conducted either on parametric models with strong functional-form assumptions or in overly simplified semi-parametric alternatives. These models provide interpretive clarity and computational efficiency but often impose restrictive assumptions that may inadequately capture the complexity of real-world data. In recent years, machine learning (ML) methods have broadened the statistician/econometrician’s toolkit by offering more flexible approaches to modeling relationships within potentially high-dimensional datasets without imposing stringent parametric assumptions; see, for example, \citet{sAgI2019} and \citet{rMeMmM2023} for recent review papers. Despite these advances, the challenge of balancing flexibility with interpretability remains a significant issue for applied researchers.

This paper presents a locally linear model that addresses the challenge of integrating the flexibility of machine learning (ML) methods with the interpretability of linear models. The proposed model aligns with several well-established parametric and semi-parametric specifications in the literature, including switching regression \citep{mgD1969,smGreQ1972}, varying-coefficient models, and additive models \citep{tHrT1993,rCrsT1993b}. It also incorporates recent advancements in machine learning, such as Random Forests (RF) introduced by \citet{breiman2001random}, Generalized Random Forests (GRF) presented by \citet{sAjTsW2019}, and local linear forests (LLF) discussed in \citet{rFjTsAsW2021}. Our central contribution is a robust framework that enables the nonparametric estimation of heterogeneous partial effects of one or more variables of interest on an outcome. Our approach is intuitive and computationally simple, achieved through an adaptation of the RF method. 

\subsection{Motivation}
Let $Y$ be a response random variable and define $h(\bs{x},\bs{z}):=\E(Y|\bs{X}=\bs{x},\bs{Z}=\bs{z})$, where $h(\cdot)$ is an unknown Lipschitz continuous function of two vectors of random covariates $\bs{X}$ and $\bs{Z}$. Suppose the goal is to estimate $h(\bs{x},\bs{z})$ for arbitrary data points $\bs{x}$ and $\bs{z}$, as well as the partial effects $\partial h(\bs{x},\bs{z})/\partial \bs{x}$  assuming $h$ is differentiable with respect to its first argument. Modern machine learning methods provide a range of nonparametric algorithms to estimate $h(\bs{x},\bs{z})$, enabling greater flexibility in modeling intricate relationships. However, these models often lack interpretability, and the computation of partial effects can be both intensive and non-trivial, particularly in the context of deep learning models.\footnote{Deep learning methods usually require the definition of too many hyperparameters, the network architecture, and also require large amounts of data for effective training.} Moreover, conducting inference on partial effects within the framework of general machine learning models remains an open problem, necessitating further research and development to enhance methodological rigor and applicability in statistical analysis.

This paper presents an alternative approach rooted in contemporary machine learning literature, specifically aimed at enhancing interpretability. We consider covariates $\bs{X}$ and $\bs{Z}$, which are not necessarily mutually independent, with a focus on the estimation  $\partial h(\bs{x},\bs{z})/\partial \bs{x}$. We propose a model in which the conditional expectation of $Y$, given $\bs{X}$ and 
$\bs{Z}$, is represented as follows:
\[
h(\bs{x},\bs{z}):=\E(Y|\bs{X}=\bs{x},\bs{Z}=\bs{z})=\bs\beta(\bs{z})^{\T}\bs{x}.
\]

This is a varying-coefficient model where the coefficients of the linear relationship between 
$\bs{X}$ and $Y$ vary as a function of $\bs{Z}$. Importantly, the partial effect of $\bs{X}$ on $Y$, $\partial h(\bs{x},\bs{z})/\partial \bs{x}$, is given directly by $\bs\beta(\bs{z})$, and thus exhibits heterogeneity across different values of $\bs{Z}$. Estimating the map $\bs{z}\mapsto \bs\beta(\bs{z})$ is equivalent to recovering the heterogeneous partial effects, providing a natural and intuitive model interpretation. Clearly, this is only the case when $\bs{Z}$ and $\bs{X}$ do not share the same covariates. We consider the estimation of $\bs{\beta}(\bs{z})$ by modification of the Random Forest (RF) method, where the trees have a linear model on each of their leaves. If $\bs{X}$ is a fixed scalar, the model equals the Random Forest (RF) method, where the trees have an intercept model on each of their leaves. 

This model presents two primary advantages. First, it facilitates flexible, data-driven estimation of complex relationships, all while ensuring the interpretability of partial effects, which is often a pivotal focus in applied research. Second, by integrating the flexibility of the Random Forest (RF) algorithm within a locally linear framework, we effectively capture heterogeneity in partial effects in a computationally feasible manner. This characteristic renders the approach particularly well-suited for large-scale empirical applications.

In many empirical applications, the dimensionality of $\bs{X}$ is expected to be small relative to the sample size (it is not uncommon to have an univariate $\bs{X}$), simplifying the application of our model. Partial effects of covariates can be estimated directly using modern machine learning and nonparametric techniques. We choose the RF algorithm for its robustness and flexibility, mainly because it requires fewer hyperparameter choices, mitigating the risks of overfitting and cherry-picking. Additionally, RF's efficient estimation algorithms make it computationally attractive. While alternative approaches like deep learning offer comparable flexibility, they are more sensitive to hyperparameter tuning and initial conditions, require larger datasets, and are computationally intensive. By contrast, RF offers a more stable and efficient solution.

\subsection{Main contributions and comparison to the literature}

The concept of introducing nonlinearity through varying coefficients in linear models is well-established. It originated in threshold regression models by \citet{mgD1969}, \citet{rQ1972}, \citet{smGreQ1972}, \citet{nmK1978}, \citet{hTksL1980}, and later work by \citet{rsT1989}. These studies focus on sharp parameter changes based on a univariate threshold variable, capturing structural shifts in the data. Smooth transitions (shifts) were introduced by \cite{ksChT1986a} and \citet{tT1994a}, but these models relied on a single transition variable. Extensions to multiple transition variables inspired by the neural network (NN) literature were explored by \citet{mcMaV2000}, \citet{mcMaV2005}, and \citet{sMcPmcM2004}, though these approaches were parametric and limited to low-dimensional settings. Nonparametric alternatives have been proposed by \citet{tHrT1993}, \cite{jFwZ1999}, and \citet{zCjFqY2000}. However, estimating these models becomes challenging with an increase in the number of covariates. In contrast, our semi-parametric approach based on random forests accommodates a large set of covariates, provided the dimensionality remains smaller than the sample size. 


Estimating regression trees with linear models in the terminal nodes is not new and is nested in the more general Generalized Random Forest model by \citet{sAjTsW2019}. \citet{rFjTsAsW2021} formalized the locally linear random forests and recommended a local Ridge regression to estimate the parameters. Nevertheless, this paper differs and complements the ones cited above in a few ways. First, we estimate the local linear models with the usual ordinary least-squares method, making our algorithm simpler than the ones in \citet{sAjTsW2019} and \citet{rFjTsAsW2021}. Second, as our primary focus is on the estimation of the partial effects $\bs{\beta}(\bs{z})$ and not $\E(Y|\bs{X}=\bs{x},\bs{Z}=\bs{z})$, we complement the previous papers by deriving a new consistency and asymptotic normality results for $\widehat{\bs\beta}(\bs{z})$. We also worked out the rates of convergence and derived the asymptotic covariance matrix for the estimator of $\widehat{\bs\beta}(\bs{Z})$. Third, we proposed a consistent estimator for the covariance matrix of $\widehat{\bs\beta}(\bs{Z})$. Fourth, unlike most papers in the literature, we show that our results are also valid when $\bs{Z}$ contains discrete random variables. Note that having discrete-valued elements of $\bs{Z}$ is not the same as introducing interaction effects between $\bs{X}$ and dummy variables constructed from the classes in $\bs{Z}$. Our approach allows the partial effect heterogeneity to be determined by unknown interactions among the elements of $\bs{Z}$. Finally, based on our new convergence results, we derived two tests to conduct inference on the partial effects and test whether the partial effects are homogeneous. 

\subsection{Outline of the Paper}
Following this introduction, the paper is structured as follows. In Section \ref{S:Setup}, we define the model, outline assumptions related to the data-generating mechanism, and present special cases of our proposal. In Section \ref{S:Main}, we present the theoretical results of this paper. More specifically, Section \ref{S:Estimation} details the main convergence results, while Sections \ref{S:GLRT} and \ref{S:LM} provide descriptions of the specification tests proposed in the paper. The case of discrete random variables is examined in Section \ref{S:Discrete}. Monte Carlo simulations are illustrated in Section \ref{S:Simulation}, and empirical samples are presented in Section \ref{S:Empirical}. Finally, Section \ref{S:Conclusions} wraps up the paper. All technical derivations are provided in the Supplemental Material.

\subsection{Notation}

Vectors are denoted by bold lowercase $\bs{x}$ and matrices by bold uppercase $\bs{M}$. $\bs{x}^{\T}$ denotes the transpose of the vector $\bs{x}$. Similarly $\bs{M}^{\T}$ denotes the transpose for matrix $\bs{M}$. We write $\|\bs{x}\|_p$ for $p\in[1,\infty]$ to denote the $\ell^p$-norm
if $\bs{x}$ is a
(possibly random) vector or the induced operator
$\ell^p$--$\ell^p$-norm if $\bs{M}$
is a matrix. For a matrix $\bs{M}$, we write
$\|\bs{M}\|_{\max}$ for the
maximum absolute entry and $\|\bs{M}\|_\rF$ for the Frobenius norm. We denote
positive semi-definiteness by $\bs{M} \succeq 0$ and write $\bs{I}_d$ for the $d \times
d$ identity matrix.

For scalar sequences $x_n$ and $y_n$, we write $x_n \lesssim y_n$ if there
exists a positive constant $C$ such that $|x_n| \leq C |y_n|$ for sufficiently
large $n$. We write $x_n \asymp y_n$ to indicate both $x_n \lesssim y_n$ and
$y_n \lesssim x_n$. Similarly, for random variables $X_n$ and $Y_n$, we write
$X_n \lesssim_\P Y_n$ if for every $\varepsilon > 0$ there exists a positive
constant $C$ such that $\P(|X_n| \geq C |Y_n|) \leq \varepsilon$, and write
$X_n \cp X$ and  $X_n \cd X$ for limits in probability and in distribution, respectively. For real numbers $a$ and $b$ we use
$a \land b = \max\{a,b\}$ and $a \lor b = \min\{a,b\}$.

\section{Setup}\label{S:Setup}

We consider the following model.

\begin{assumption}[Locally Linear Model]\label{ass:local-linear_spec}
Let $Y$ be an integrable random variable and $\bs{X}$,$\bs{Z}$ be random vectors taking value on $\R^{d_Z}$ and $\R^{d_X}$ respectively that do not share common variables. We assume that for some Lipschitz function $\bs\beta:\mathcal{Z}\subseteq\R^{d_Z}\to\R^{d_X}$ 
\begin{equation}\label{eq:local-linear_model}
h(\bs{x},\bs{z}):=\E[Y|\bs{X}=\bs{x},\bs{Z}=\bs{z}]=\bs{x}^{\T}\bs{\bs\beta}(\bs{z}).
\end{equation}
\end{assumption}

The local-linear model \eqref{eq:local-linear_model} has the advantage of having the marginal treatment effect built-in since $h(\bs{x},\bs{z})/\partial \bs{x}=\bs\beta(\bs{z})$. At the same time, the model is flexible enough in terms of the covariates $\bs{Z}$ to accommodate complex heterogeneous partial effects, which include higher-order interactions between $\bs{X}$ and $\bs{Z}$. Model \eqref{eq:local-linear_model} nests several notable cases of interest.

Since $\E[Y|\bs{X},\bs{Z}]$ is the minimizer of $\E[(Y-f(\bs{X},\bs{Z}))^2]$ over the class of functions $f$ such that $\E[f(\bs{X},\bs{Z})^2]<\infty$, we can explicitly characterize the function $\bs{z}\mapsto \beta(\bs{z})$ such that
$\bs\beta(\bs{z}) =  \bs{\Omega}(\bs{z})^{-1}\bs{\gamma}(\bs{z})$, where $\bs{\Omega}(\bs{z}):= \E[\bs{X}\bs{X}^{\T}|\bs{Z}=\bs{z}]$ and $\bs{\gamma}(\bs{z}):=\E[\bs{X}Y|\bs{Z}=\bs{z}]$,
and provided that $\bs{\Omega}(\bs{z})$ is almost sure positive definite.

\begin{example}[Heterogeneous treatment effects]
Consider a collection of potential outcomes $\{Y^*(t):t\in\mathcal{T}\}$ where $Y(t)$ is a random variable and $\mathcal{T}\subseteq\R$. In a binary treatment case $\mathcal{T}=\{0,1\}$. Here we allow continuous treatment such as $\mathcal{T}=[a,b]$ for $a<b$ or $\mathcal{T}=\R$. In the continuous treatment literature, the function $t\mapsto Y^*(t)$ is called the dose-response function, and $t\mapsto \E[Y^*(t)]$ is the average dose function. For each unit in the sample, we observe the realization of the treatment status $T$, a vector of covariates $\bs{Z}$ taking values on $\R^{d_Z}$, and the potential outcome corresponding to the treatment received $Y:=Y^*(t)$. 

The interest relies on estimating the (potentially) heterogeneous (marginal) treatment effect.
\[
\tau(\bs{z}) :=\begin{cases}
    \frac{d\E[Y^*(t)|\bs{Z}=\bs{z}]}{dt} & \text{if $\mathcal{T}$ is an interval}, \\
    \E[Y^*(1) -Y^*(0)|\bs{Z}=\bs{z}] &\text{if $\mathcal{T}=\{0,1\}$}.
\end{cases}
\]
By definition we have that $\E[Y^*(t)|T=t]=\E[Y|T=t]$ for $t\in\mathcal{T}$. However, in general, we have  $\E[Y^*(t)]\neq \E[Y|T=t]$ due to confounding effects. If conditional on the covariates $\bs{Z}$ the average dose-response function is mean independent of the treatment, we can consider the following model
\begin{equation}\label{ex:Treat0}
Y = \beta(\bs{Z})T +U,
\end{equation}
where $\E(U|\bs{Z},T)=0$. We can also include additional controls in the above model. For example, let $\bs{W}$ be a set of control variables and write
\begin{equation}\label{ex:Treat}
Y = \beta(\bs{Z})T + \bs\delta^{\T}\bs{W}+U,
\end{equation}
where $U$ is an error term such that $\E(U|\bs{Z},T,\bs{W})=0$. Equations \eqref{ex:Treat0} and\eqref{ex:Treat} are models where the treatment effects are not constant and may vary with subject characteristics defined by the vector $\bs{Z}$. Examples of such specifications can be found at \citet{beH2000}, \citet{sAgI2016}, or \citet{vCiFyL2018}, among many others.
\end{example}

\begin{example}[Smooth Transition Regression]
Suppose that the coefficients of a regression model change smoothly and monotonically between zero and one according to a scalar variable $Z_i$ such that
\begin{equation}\label{ex:STR}
Y = \bs\beta_0^{\T}\bs{X} + g(\bs{Z})\bs\beta_1^{\T}\bs{X} + U,
\end{equation}
where $\lim_{z\rightarrow\infty}g(\bs{z})=1$ and $\lim_{z\rightarrow-\infty}g(\bs{z})=0$ and  $\E(U|\bs{X},Z)=0$. In this case, $\bs\beta(\bs{z})=\bs\beta_0 + g(\bs{z})\bs\beta_1$. This is a general case of a two-regime smooth transition regression model. An interesting example of such specification is the nonlinear Phillips curve model discussed in \citet{wAmMmM2011}.
\end{example}

\begin{example}[Grouped patterns of heterogeneity]
Consider a model where the intercept varies according to a set of variables $\bs{Z}$. In this case, 
\begin{equation}\label{ex:Group}
Y = \bs\beta(\bs{Z}) + \bs\delta^{\T}\bs{X}+U,
\end{equation}
where $\E(U|\bs{X},Z)=0$. This type of specification is commonly employed to model clustering behavior and has gained traction in the empirical economic growth literature. See, for example, \citet{sBeM2015}. This specification is also related to the binscatter methods discussed in \citet{mdCrkCmhFyF2024}.
\end{example}

Given a random sample $\{(Y_i,\bs{X}_i,\bs{Z}_i):1\leq i\leq n\}$ of $(Y,\bs{X},\bs{Z})$ we propose to estimate \eqref{eq:local-linear_model} by a modification of the Random Forest procedure (refer to Algorithm \ref{alg:cap}). We start by defining for any nonempty set of indices $\mathcal{I}\subseteq [n]$, the residual sum of squares ($\RSS$) of a least-square estimation (we will impose conditions such that the estimator is well-defined). Hence, write
\begin{align*}
    \RSS(\cI)&:=\sum_{i\in\cI}\left[Y_i-\bs{X}_i^ {\T}\widehat{\bs\beta}(\cI)^2\right],\qquad \text{and}\qquad\widehat{\bs\beta}(\cI):=\left[\sum_{i\in\cI}\bs{X}_i\bs{X}_i^{\T}\right]^{-1}\sum_{i\in\cI}\bs{X}_iY_i.
\end{align*}
Also, for any element of $\bs{Z}$ indexed by $j\in[d_Z]$, and $\delta\in[0,1]$ define 
\begin{equation}\label{eq:MSE_continuos}
\Delta(\mathcal{I},j,\delta):= \RSS(\{i\in\cI:Z_{ij}\leq \delta \} ) + 
\RSS(\{i\in\cI:Z_{ij}> \delta \} ). 
\end{equation}
The optimum splitting point $\delta^*$ of the index set $\mathcal{I}$  along direction $j$ is given by
\begin{equation}\label{eq:split_criteria}
    \delta^*:=\delta^*(\mathcal{I},j) \in\argmin_{\delta\in[0,1]} \Delta(R,j,\delta).
\end{equation}
Note that we might take $\delta^*\in\{Z_{i,j}:i\in\mathcal{I}\}$. Finally, define the left and right child nodes of $\mathcal{I}$ by
\[
\cI^{-}:=\{i\in\cI:Z_{ij}\leq \delta^*\},\qquad\text{and}\qquad \cI^+:=\{i\in\cI:Z_{ij}> \delta^*\}.
\]

Starting for an initial index set $\mathcal{B}\subseteq [n]$ and repeating the steps above, we end up with a partition of $[0,1]^{d_Z}$ into disjoint rectangles with axis-aligned sides (leaves). Furthermore, since the initial index set (root node) and the slipt direction $j$ are randomly chosen independent of the data, we denote by $\omega$ this independent source of randomness in the algorithm. 

For $\bs{z}\in [0,1]^{d_Z}$, let $R(\bs{z},\omega)$ denotes the unique (random) leaf containing $\bs{z}$, i.e.,  $R(\bs{z},\omega)$ is a random rectangle that depends on the observation indexed by $\mathcal{B}$ and an external source of randomness. For an nonempty $\cA\subseteq [n]$ such that $\cA\cap\cB=\emptyset$, define $\cA(\bs{z},\omega)=\{i\in\mathcal{A}:\bs{Z}_i\in R(\bs{z},\omega)\}$, $\widehat{m}(\bs{x},\bs{z},\omega):=\bs{x}^{\T}\widehat{\bs\beta}(\bs{z},\omega)$, and $\widehat{\bs\beta}(\bs{z},\omega):= \widehat{\bs\beta}(\mathcal{A}(\bs{z},\omega))$. 

More explicitly, write
\begin{equation}\label{eq:beta_hat_definition}
    \widehat{\bs\beta}(\bs{z},\omega)= \widehat{\bs{\Omega}}(\bs{z},\omega)^ {-1}\widehat{\bs{\gamma}}(\bs{z},\omega),
\end{equation}
where $\widehat{\bs{\Omega}}(\bs{z},\omega):= \frac{1}{|\cA(\bs{z},\omega)|}\sum_{i\in\cA(\bs{z},\omega)}\bs{X}_i\bs{X}_i^{\T}$ and  $\widehat{\bs{\gamma}}(\bs{z},\omega):=\frac{1}{|\cA(\bs{z},\omega)|}\sum_{i\in\cA(\bs{z},\omega)}\bs{X}_iY_i$. Finally, for $B\geq 1$, we propose to estimate $m(\bs{x},\bs{z})$ using $\overline{m}(\bs{x},\bs{z}):= \frac{1}{B}\sum_{b=1}^B\widehat{m}(\bs{X},\bs{z},\omega_b) = \bs{X}^{\T}\overline{\bs\beta}(\bs{z})$ with
\begin{equation}\label{eq:RF_estimator_definition}
\overline{\bs\beta}(\bs{z}):=\frac{1}{B}\sum_{b=1}^B\widehat{\bs\beta}(\bs{z},\omega_b)
\end{equation}
where $\{\omega_b:b\in[B]\}$ is an independent and identically sequence independent of $\{(Y_i,\bs{X}_i,\bs{Z}_i):i\in[n]\}$. We also investigate the properties of the following estimator:
\[
\widecheck{\bs\beta}(\bs{z}):=\overline{\bs{\Omega}}(\bs{z})^{-1}\overline{\bs{\gamma}}(\bs{z}),
\]
where $\overline{\bs{\Omega}}(\bs{z}):= \frac{1}{B}\sum_{b=1}^B \widehat{\bs{\Omega}}(\bs{z},\omega_b)$ and $\overline{\bs{\gamma}}(\bs{z}):= \frac{1}{B}\sum_{b=1}^B \widehat{\bs{\gamma}}(\bs{z},\omega_b)$.

\begin{algorithm}[H]
\caption{Local-Linear Double-sample tree}\label{alg:cap}
\begin{algorithmic}
\Require $n\geq 2$ training samples of the form $(y_i,\bs{x}_i,\bs{z}_i)\in (\R\times\R^ {d_X} [0,1]^{d_Z})$, a subsample size $s<n$, minimum leaf size $k\leq s$, number of trees $B\in\N$, minimum fraction of observation in each child node $\alpha\in(0,0.5)$ and the probability of selecting a variable to split is at least $\pi/d_Z$ for $\pi\in(0,1]$.
\For{$b=1,\dots, B$}
    \State 1. Draw a random subsample of size $s$ for $\{1,\dots, n\}$ without replacement
    \State 2. Partition the subsample into $\mathcal{A}\cup\mathcal{B}$ such that $|\mathcal{A}|=\lfloor s/2\rfloor$ and $|\mathcal{B}|=\lceil s/2\rceil$
    \State 3. Grow a tree using only data from $\mathcal{B}$ via standard CART algorithm 
    \While{the number of observations in every leaf is not between $k$ and $2k-1$}
    \State Sample a variable $Z_j$ from $\{Z_1,\dots, Z_{d_Z}\}$ 
    \State Decide the cutoff in $Z_j$ based on criteria \eqref{eq:split_criteria} imposing the $\alpha$-regularity condition above
    \EndWhile

    \State 4. Estimate the leaves responses using only data from  $\mathcal{A}$ as in \eqref{eq:beta_hat_definition}
\EndFor
\State Compute the random forest estimator given in \eqref{eq:RF_estimator_definition}
\end{algorithmic}
\end{algorithm}

\section{Main Results}\label{S:Main}
In this section, we state the main results of the paper and their underlying assumptions.

\subsection{Estimation}\label{S:Estimation}

\begin{assumption}\label{ass:main}
We consider the following conditions:
\begin{enumerate}
    \item[(a)] (Sampling) The training sample $\{(Y_i,\bs{X}_i,\bs{Z}_i):i\in[n]\}$ is $n$ independent copies of $(Y,\bs{X},\bs{Z})$ taking values on $\R\times \R^{d_X} \times [0,1]^{d_Z}$ and $\bs{X}_i$ has finite moments up to order $4$.
    \item[(b)] (Density)  $\bs{Z}$ admits a density $f$ that is bounded away from zero and infinity
    \item[(c)] (Smoothness) $\bs{z}\mapsto \bs\beta(\bs{z})$, $\bs{z}\mapsto \bs{\Omega}(\bs{z})$, $\bs{z}\mapsto \E[\bs{X}(Y-\E[Y|\bs{X},\bs{Z}])|\bs{Z}=\bs{z}]$ and $\bs{z}\mapsto \V[\bs{X}(Y-\E[Y|\bs{X},\bs{Z}]|\bs{Z}=\bs{z}]$ are Lipschitz-continuous, and $\bs{z}\mapsto \V[X_jX_k|\bs{Z}=\bs{z}]$ and $\bs{z}\mapsto \V[X_j(Y-\E[Y|\bs{X},\bs{Z}]|\bs{Z}=\bs{z}]$ are bounded for $j,k\in[d_Z]$;
    \item[(d)] (Conditional Moment Lower bounds) The matrices $\E[\bs{X}\bs{X}^{\T}|\bs{Z}=\bs{z}]$ and $\V[\bs{X}(Y-\E[Y|\bs{X},\bs{Z}])|\bs{Z}=\bs{z}]$ are non-singular for each $\bs{z}\in[0,1]^d$;
    \item[(e)] The trees are generated by Algorithm \eqref{alg:cap} with  with $\alpha\in (0,0.5)$ and $\pi\in (0,1]$.
\end{enumerate}
\end{assumption}
In Assumption \ref{ass:main}(e), we require the minimum of observations per leaf to increase with the same size to ensure that each tree is consistent, as opposed to the consistency of the random forest. Specifically, we need the variance of $\widehat{\bs{\Omega}}(\bs{z})$ to vanish on each tree as the sample size increases. At the same time, $k$ must grow slower than the subsampling rate $s$ so that each cell is split enough times to make its diameter shrink toward zero, and the tree bias vanishes as the sample size increases.

Recall that, for the tree constructed with randomness $\omega$, $R(\bs{z},\omega)$ denotes the leaf containing $\bs{z}$. Define the map $\widetilde{\bs{\beta}}:[0,1]^d\to \R^{d_X}$ by
\[
\widetilde{\bs\beta}(\bs{z}) := \nabla_x\E[Y|\bs{X}=\bs{x},\bs{Z}\in R(\bs{z},\omega)] = \E[\bs\beta(\bs{Z})|\bs{Z}\in R(\bs{z},\omega)];\qquad \bs{z}\in[0,1]^d.
\]
Note that in general we expect $\widetilde{\bs{\beta}}(\cdot)\neq \bs{\beta}(\cdot)$.
However, due to the sample split we have that $\widehat{\bs\beta}(\cdot,\omega)$ is an unbiased estimator for $\widetilde{\bs\beta}(\cdot)$ even when $\widehat{\bs\beta}(\cdot,\omega)$ is  biased for $\bs{\beta}(\cdot)$ (and $\widetilde{\bs\beta}(\cdot)$).

\begin{theorem}[Unbiasedness]\label{thm:unbiasedness}
    Under Assumptions \ref{ass:local-linear_spec} and \ref{ass:main}, $\E[\widehat{\bs\beta}(\bs{z})] = \widetilde{\bs\beta}(\bs{z})$ for all $\bs{z}\in[0,1]^d$.
\end{theorem}

\begin{theorem}[Rate of Convergence]\label{thm:beta_converge_rate}
    Under Assumptions \ref{ass:local-linear_spec} and \ref{ass:main}, if $k\gtrsim s^\epsilon$ for some  $\epsilon\in(0,1)$  then for any $\delta \in (0,1)$
\[
|\widehat{\bs\beta}(\bs{z},\omega) - \bs\beta(\bs{z})|\lesssim_\P k^{-\frac{1}{2}}\left(\frac{s}{k}\right)^{\epsilon/2} + \left(\frac{k}{s}\right)^{(1-\delta)\frac{K(\alpha)\pi}{d_Z}};\qquad \bs{z}\in[0,1]^d,  
\] 
where $K(\alpha):= \frac{\log ((1-\alpha)^{-1})}{\log(1/\alpha)}\in(0,1)$.
If further
then , for $2\leq p\leq q$,
\begin{align*}
    \left[\int_{[0,1]^d}|\widehat{\bs\beta}(\bs{z},\omega) - \bs\beta(\bs{z})|^pf(\bs{z}) dz \right]^{1/p} &\lesssim_\P k^{-\frac{1}{2}}\left(\frac{s}{k}\right)^{\epsilon/2} + \left(\frac{k}{s}\right)^{(1-\delta^*)\frac{K(\alpha)\pi}{d_Z}},
\end{align*}
where $\delta^*\in (0,1)$ is defined in Lemma \ref{lem:upper_bound_diameter}. 
\end{theorem}

\begin{remark}
    The second term in the right-hand size is due to the tree's bias. Given the Lipschitz condition, asymptotic unbiasedness is obtained by shrinking the leaf diameter towards zero, which is a consequence of subsequent splits in the node, which, in turn, depend on $k/s$. The second one is due to the variance within each leaf and, not surprisingly, improves for larger $k$. The bias-variance trade-off could be optimized with select $k\asymp s^\eta$ for $\eta>0$ as a function of $d_Z$, $\alpha$, and $\pi$ such that those two terms are of the same order.
\end{remark}


Although a tree construct using Algorithm \ref{alg:cap} is indeed honest and symmetric, there is no guarantee to be a $k$-PNN predictor. To see that, fix a test point $\bs{z}\in[0,1]^d$ and let $R(\bs{z},\omega)$ denote the unique leaf 
containing $z$. Recall that $R(\bs{z},\omega)$ is independent of $\mathcal{A}$-sample due to honesty. Then the number of observations in the $R(\bs{z},\omega)$ wich we denoted by $|\mathcal{A}(\omega,z)|$ follows a Binomial distribution conditional on $R(\bs{z},\omega)$ with $s$ trial an probability of success $p(\bs{z},\omega) :=\P\big(\bs{Z}\in R(\bs{z},\omega)|R(\bs{z},\omega)\big)$.

Therefore, the expected number of $\mathcal{A}$-sample observations in $R(\bs{z},\omega)$ conditional on $R(\bs{z},\omega)$ is given by $s p(\bs{z},\omega)$ while the number of $\mathcal{B}$-sample in $R(\bs{z},\omega)$ is between $k$ and $2k-1$ by construction. The question becomes how $|\mathcal{A}(\bs{z},\omega)|$ relates to $ k$. As it is shown in Lemma \ref{lem:lower_bound}, $|\mathcal{A}(\omega,z)|$ can be upper and lower bound in probability as 
\[
\P\left[s\left(\frac{s}{k}\right)^{-1/K(\alpha)}\lesssim|\mathcal{A}(\bs{z},\omega)|\lesssim s\left(\frac{s}{k}\right)^{-K(\alpha)}\right]\gtrsim 1 - \frac{1}{s(s/k)^{1/K(\alpha)}}.
\]
Set $k\asymp s^\eta$ for $\eta\in [0,1)$. For $\eta\in (1-K(\alpha),1)$ we have that $|A(\bs{z},\omega)|$ diverges as $s\to\infty$ with high probability. Precisely.
\[
\P\left[s^{1-\frac{1-\eta}{K(\alpha)}}\lesssim|\mathcal{A}(\bs{z},\omega)|\lesssim s^{1-(1-\eta) K(\alpha)}\right]\to 1
\]
For $\eta\in[0,1-K(\alpha)]$ and, in particular, $\eta=0$ (fixed $k$), the above bound is vacuous. It is not clear whether, for any fixed $k>0$, it is possible to claim that 
$|\mathcal{A}(\bs{z},\omega)|\gtrsim k$ with high probability.

\begin{theorem}[Random Forest Asymptotic Normality]\label{thm:asym_normality_RF}
    Under Assumptions \ref{ass:local-linear_spec} and \ref{ass:main}, if $k\asymp s^{\eta}$ for $\eta\in(0,1)$ and $s\asymp  n^\omega$. where  $\omega\in \left(\left[(1-\eta)\left(\frac{\pi K(\alpha)}{d_Z}+\frac{1}{K(\alpha)}\right)\right]^{-1},1\right)$, then as $n\to\infty$
    \[
\bs\Sigma(\bs{z})^{-1/2}\left[\overline{\bs\beta}(\bs{z})-\bs\beta(\bs{z})\right] = \frac{s}{n}\sum_{i=1}^n h_i(\bs{z}) \cd \mathsf{N}(0,\bs I_{d_X}),\qquad \bs{z}\in[0,1]^{d_Z},   
    \]
     where $\bs\Sigma(\bs{z}):=\bs{\Omega}^{-1}(\bs{z}) \bs\Lambda(\bs{z})\bs{\Omega}^{-1}$ with 
     \begin{equation}\label{eq:asymp_variance_eX}
    \bs\Lambda(\bs{z}):=\frac{s^2}{n}\E\left[\epsilon^2\theta(\bs{z},\bs{Z})^2 \bs X \bs{X}^ {\T}\right],
\end{equation}
and $\theta(\bs{z},\bs{z}')=\E_\omega[S(\bs{z},\bs{z}',\omega)]$ with $S(\bs{z},\bs{z}',\omega) = |\{j\in\mathcal{A}:Z_j\in R(\bs{z},\omega)\}|^{-1}$ if $z'\in R(\bs{z},\omega)$ and zero otherwise, for $\bs{z},\bs{z}'\in[0,1]^{d_Z}.$
In particular
    \[
    \|\overline{\bs\beta}(\bs{z})-\bs\beta(\bs{z})\|\lesssim_\P n^{-\frac{1}{2}\big[\bs\beta(1-\eta)\big(\frac{\pi K(\alpha)}{d_Z}+  \frac{1}{K(\alpha)}\big) -1\big]}(\log n)^{{d_Z}/2}\to 0;\qquad \bs{z}\in[0,1]^{d_Z}.
    \]
\end{theorem}
\begin{remark}
    Converge of each tree is necessary for the random forest convergence. Hence, the requirement to set $k\asymp s^\eta$ according to Theorem \ref{thm:asym_normality_RF}. It is interesting to compare the rate of convergence of a (single) tree with $s=n$ with the rate of convergence of the forest with the restrictions to ensure asymptotic normality.
\end{remark}

We propose to estimate the covariance matrix $\bs\Sigma(\bs{z})$ appearing in Theorem \ref{thm:asym_normality_RF} using
\begin{equation}\label{eq:covariance_estimator}
    \widehat{\bs\Sigma}(\bs{z}):=[\widehat{\bs{\Omega}}(\bs{z})]^{-1}\widehat{\bs\Lambda}(\bs{z}) [\widehat{\bs{\Omega}}^{-1}(\bs{z})]^{-1},
\end{equation}
where $\widehat{\bs{\Omega}}(\bs{z})$ is given by $\eqref{eq:beta_hat_definition}$. As for an estimator for $\widehat{\bs\Lambda}(\bs{z})$, define the residual function of the RF  by $\bs{z}\mapsto \widehat{\epsilon}_i(\bs{z}):= Y_i - \bs{X}_i^{\T}\overline{\bs\beta}(\bs{z})$ for $i\in[n]$ and $\bs{z}\in[0,1]^{d_Z}$. Note  any pair $\bs{z},\bs{z}'\in[0,1]^{d_Z}$ we observe $B$ independent realizations of $S(\bs{z},\bs{z}',\omega_b)$. So $\theta(\bs{z},\bs{z}')$ can be unbiased estimated by $\widehat{\theta}(\bs{z},\bs{z}'):=\frac{1}{B}\sum_{b=1}^B S(\bs{z},\bs{z}',\omega_b)$. So we propose to estimate $\widehat{\bs\Lambda}(\bs{z})$ using

\begin{equation}\label{eq:jacknife_estimator}
\widehat{\bs\Lambda}(\bs{z}):=\frac{s^2}{n}\sum_{i=1}^n\widehat{\epsilon}_i(\bs{z})^2\widehat{\theta}(\bs{z},\bs{Z}_i)^2\bs{X}_i \bs{X}_i^{\T} = \frac{s^2}{n}\sum_{i=1}^n\widehat{\epsilon}_i(\bs{z})^2\left[\frac{1}{B}\sum_{b=1}^B S(\bs{z},\bs{Z}_i,\omega_b)\right]^2\bs{X}_i \bs{X}_i^{\T}.
\end{equation}

Even when an observation $\bs{Z}_i$ is not used to grow a tree $\omega_b$, we can compute $S(\bs{z},\bs{Z}_i,\omega)$ by checking whether $\bs{Z}_i$ lies in $R(\bs{z},\omega)$ and if so count how many observations are in $R(\bs{z},\omega)$. Since we have $\E[\epsilon\theta(\bs{z},\bs{Z})X]=0$, we might consider the centered version of the estimator below given by
\begin{equation}\label{eq:asymp_var_estimator_centered}
\widehat{\bs\Lambda}_c(\bs{z}):=\widehat{\bs\Lambda}(\bs{z}) - s^2.
\end{equation}

\subsection{Testing for Heterogeneity}\label{S:Testing}

\subsubsection{Generalized Likelihood Ratio Test}\label{S:GLRT}

Suppose we wish to test that the conditional expectation does not depend on $\bs{Z}$ (homogeneous partial effect). Specifically, we are interested in the parametric null.
\[
\mathcal{H}_0:\bs\beta(\bs{z}) = \bs\beta_0\qquad\text{for some $\bs\beta_0\in \R^{d_X}$ and all $\bs{z}\in[0,1]^d$}
\]
against the non-parametric alternative hypothesis $\mathcal{H}_1:\bs\beta(\bs{z})$ is Lipschitz. Following \cite{jFcZjZ2001}; we propose to test $\mathcal{H}_0$ using the Generalized LRT, which is given as (after taking logs)
\[
\Lambda(\mathcal{H}_0):=\frac{n}{2}\log\frac{\RSS_0}{\RSS},
\]
where $\RSS_0 := \sum_{i=1}^n (Y_i - \bs{X}_i^{\T}\widetilde{\bs\beta})^2 $, $\RSS := \sum_{i=1}^n \left[Y_i - \bs{X}_i^{\T}\overline{\bs\beta}(\bs{Z}_i)\right]^2 $ and $\widetilde{\bs\beta}$ is the OLS estimator.

\begin{theorem}[Generalized LRT]\label{thm:LR_Test} Under the same condition of Theorem \ref{thm:asym_normality_RF} and $\mathcal{H}_0$, if further $\V[Y|\bs{X},\bs{Z}]=\sigma^2$ for some positive constant $\sigma^2$, then $\Lambda\stackrel{a}{\sim} \chi^2_{\mu}$ where $\chi^2_\mu$ is a chi-squared distribution with $\mu$ degrees of freedom, in the sense that,
    \[
    \nu^ {-1}\left[\Lambda(\mathcal{H}_0) - \mu + o(\mu)\right]\cd \mathsf{N}(0,1)
    \]
where
\begin{align*}
    \mu &:= d_X \left[2s\E[\theta(\bs{Z},\bs{Z})]+\frac{s^2}{n}\E[\theta(\bs{Z},\bs{Z)}^2] + s^2\E[\theta(\bs{Z},\bs{Z'})^2]\right]\\
    \nu^2 &:= 4 d_X\E\Big[\Big(s\theta(\bs{Z},\bs{Z}') +s^2\E[\theta(\bs{Z},\bs{Z}')\theta(\bs{Z},\bs{Z}'')]\Big)^2 \Big].
\end{align*}
and $\theta(\bs{z},\bs{z}')$ is defined in Theorem \ref{thm:asym_normality_RF}; and $\bs{Z}'$ and $\bs{Z}''$ are independent copies of $\bs{Z}$.
\end{theorem}
%

\subsubsection{Langrage Multiplier Test} \label{S:LM}

The idea is to explore the fact that for a Lagrange Multiplier (LM) type test, we are only required to estimate the model under the null and obtain a consistent estimator under both the null and alternative. When the null completely characterizes $\bs\beta(\bs{z})$ or when $\bs\beta(\bs{z})$ is known up to finite unknowns (parametric), we propose a somewhat canonical test.

Under $\mathcal{H}_0:\bs\beta(\bs{z})$ is not a function of $z$, we have that $\E[\bs{Z}(Y-\bs{X}^{\T}\widetilde{\bs\beta})=0$ for some unknown $\widetilde{\bs\beta}\in\R^{d_X}$. Let $\widehat{\epsilon}_{i}^ {OLS}:=Y_i - \bs{X}_i^{\T}\widehat{\bs\beta}_{OLS}$ for $i\in[b]$ where $\widehat{\bs\beta}_{OLS}$ is the OLS estimator of $Y$ regressed onto $X$

So, we can use the sample moment condition below as a basis for constructing our test statistic.
\[
M=\sum_{i=1}^n \widehat{\epsilon}_{i}^ {OLS}\bs{Z}_i,
\]
because under $\mathcal{H}_0$
\[
M/\sqrt{n} = L\frac{1}{\sqrt{n}}\sum_{i=1}^n \epsilon_iW_i + o_\P(1);\qquad W_i := (\bs{Z}_i^{\T}, \bs{X}_i^ {\T})^{\T}
\]
where $L:=(I_d: -\E[Z\bs{X}^{T}]\big(\E[\bs{X}\bs{X}^{T}\big)^{-1})$ is an $(d_Z (d_Z+d_X))$ matrix. Since $\frac{1}{\sqrt{n}}\sum_{i=1}^n \epsilon_iW_i\cd \mathsf{N}(0,\E[\epsilon^2WW^ {\T}])$, we have that $M/\sqrt{n}\cd \mathsf{N}(0,L\E[\epsilon^2WW^ {\T}]L^{\T})$ therefore
\[
M^ {\T} (nV)^{-1} M\cd\chi^2_{d_Z};\qquad V:=L\E[\epsilon^2WW^ {\T}]L^{\T}
\]
Let $\widehat{\epsilon}_{i}^{RF}:=Y_i - \bs{X}_i^{\T}\overline{\bs\beta}$ be the residuals of the RF. Under the alternative (and the null) $\overline{\bs\beta}(\bs{z})\cp \bs\beta(\bs{z})$ then $\epsilon_{i}^{RF}\cp\epsilon_i$ for $i\in[n]$. Define the plug-in estimators
\[
\widehat{L}:=\left[\bs{I}_d: -\sum_{i=1}^nZ_i\bs{X}^{\T}_i\left(\sum_{i=1}^n\bs{X}_i\bs{X}^{\T}_i\right)^{-1}\right],\qquad \widehat{V}:=\widehat{L}\frac{1}{n}\sum_{i=1}^ n\left[(\widehat{\epsilon}_{i}^{RF})^2W_iW^{\T}_i\right]\widehat{L}^{\T}
\]

Hence, the test statistics become
\[
T:= M^ {\T} (n \widehat{V})^{-1} M\cd\chi^2_{d_Z}
\]
as $n\to\infty$ under $\mathcal{H}_0$.

Note that the same process works to test $\mathcal{H}_0:\bs\beta(\bs{z})=\bs\beta_0(\bs{z})$ for some known function $\bs{z}\mapsto\bs\beta_0(\bs{z})$.

\subsection{Extension to discrete controls }\label{S:Discrete}

Partition the control variables as $\bs{Z}= (\bs{Z}',\bs{Z}'')$ where $\bs{Z}''= (Z_{1}'',\dots Z_{d_{\bs{Z}''}}'')$ and $\bs{Z}''_j$ are discrete random variables with $m_j\geq 2$ categories  for $j\in[d_{\bs{Z}''}]$. Without loss of generality we may assume that $Z_{j}''$ is supported on $\mathcal{S}_j=\{0,1/(m_j-1),1/(m_j-2), \dots, 1\}$ and thus $\bs{Z}''$ has support on $\mathcal{S} = \bigtimes_{j=1}^{d_{\bs{Z}''}} \mathcal{S}_j$.

Let $R=\bigtimes_{j=1}^{d_Z}[a_j,b_j]$ denote a rectangle in $[0,1]^{d_Z}$. For convenience, we can define the length of a rectangle along a discrete variable $\bs{Z}''_j$ as the fraction of the categories in the rectangle. Specifically $0\leq \diam(R)_j := (|\mathcal{S}_j\cap [a_j,b_j]|-1)/(m_j-1)\leq 1$.

When a discrete variable $Z_j$ is chosen to split we replace $\eqref{eq:MSE_continuos}$ by the condition
\begin{equation}\label{eq:MSE_discrete}
\delta^* \in\argmin_{\delta\in \mathcal{S}_j}\big[ \RSS(\{i\in\cI:Z_{ij}=\delta) \} ) + \RSS(\{i\in\cI:Z_{ij}\neq \delta) \} ) \big],
\end{equation}
and identify the left and right child nodes by the index sets
\begin{equation}\label{eq:index_discrete}
      \cI^{-}:=\{i\in\cI:Z_{ij}= \delta^*\},\quad \cI^+:=\{i\in\cI:Z_{ij}\neq \delta^*\}.
\end{equation}

Fix a test point $\bs{z}=(\bs{z}',\bs{z}'' )\in[0,1]^{d_{\bs{Z}'}}\times \mathcal{S}$ and let $M_j(\bs{z})$ denote the number of splits along the $j$-th discrete variable to form the (unique) leaf containing $z$. Since $s/k\to \infty$ we have that some large $s$
\[
\P(M_j(\bs{x})< m_j-1)\leq \P\left( M_{j}(\bs{z})\leq  \tfrac{(1-\delta)\pi}{d_Z} \frac{\log(s/(2k-1))}{\log(1/(1-\alpha))}\right)\leq \left(\frac{s}{2k-1}\right)^{-\frac{\delta^2\pi^2}{2d_Z^2\log(1/(1-\alpha))}}.
\]
Let $\mathcal{E}_D(\bs{x}) = \bigcap_{j\in[d_Z'']}\{M_j(\bs{x}) = m_j-1\}$, then by the union bound we have
\[
\P(\mathcal{E}_D(\bs{x})) \gtrsim 1 -d_Z''\left(\frac{s}{k}\right)^{-\frac{\delta^2\pi^2}{2d_Z^2\log(1/(1-\alpha))}}
\]
i.e., all discrete variables have a single class in the leaf $R(\bs{z},\omega)$ with probability at least $1-\left(\frac{s}{k}\right)^\frac{-\delta^2\pi^2}{2d_Z^2\log(1/(1-\alpha))}$. Hence, the bias on the leaf $R(\bs{z},\omega)$ can be upper bounded on $\mathcal{E}_D(\bs{x})$ as
\[
\max_{(u',u'')\in R((z',z''),\omega)}\|\bs\beta(u',u'')-\bs\beta(z',z'')\|\leq \|\bs\beta(u',z'')-\bs\beta(z',z'')\|\leq C\diam(R(\bs{z},\omega))
\]

\begin{assumption}\label{ass:main_discrete}
We consider the following conditions:
\begin{enumerate}
    \item[(b)] (Conditional Density)  We can partition $\bs{Z}= (\bs{Z}', \bs{Z}'')$ such that $\bs{Z}'$ has a conditional density with respect to $\bs{Z}''$, which we denote by $f_{\bs{Z}'|\bs{Z}''}$, and  $1/C\leq f_{\bs{Z}'|\bs{Z}''}\leq C$ for some constant $C>0$. Finallly, $\P(\bs{Z}''=\bs{z}'')>0$ for all $z''\in\mathcal{S}$. 
    \item[(c)] (Smoothness) Let $z=(z',z'' )\in[0,1]^{d_{Z_C}} \mathcal{S}$. We assume that $z'\mapsto \bs\beta(z',z'')$ $z'\mapsto \bs{\Omega}(z',z'')$, $z'\mapsto \E[\bs{X}(Y-\E[Y|\bs{X},\bs{Z}])|Z_C=\bs{z}', Z_D=\bs{z}'']$ and $\bs{z}\mapsto \V[\bs{X}(Y-\E[Y|\bs{X},\bs{Z}]|Z_C=\bs{z}', Z_D=\bs{z}'']$ are Lipschitz-continuous for every $z''\in\mathcal{S}$, and $\bs{z}\mapsto \V[X_jX_k|\bs{Z}=\bs{z}]$ and $\bs{z}\mapsto \V[X_j(Y-\E[Y|\bs{X},\bs{Z}]|\bs{Z}=\bs{z}]$ are bounded for $j,k\in[d_Z]$;
\end{enumerate}
\end{assumption}

\begin{theorem}[Discrete Control Variables]\label{thm:discrete}
    Under Assumptions \ref{ass:local-linear_spec} and \ref{ass:main} with \ref{ass:main}(b) replaced by \ref{ass:main_discrete}(b) and \ref{ass:main}(c) replaced by \ref{ass:main_discrete}(c), Theorems \ref{thm:unbiasedness}-\ref{thm:asym_normality_RF} hold.
\end{theorem}

\section{Monte Carlo Simulation}\label{S:Simulation}

We conducted a simulation study to test the validity of our results in finite sample. We consider several specifications with different sample sizes. Specifically, we set the number of observations to be $n\in\{250,500,1000\}$. The dimension $d_Z$ of the variables determining model heterogeneity is $d_Z=\{1,2,3,5\}$. Each simulated model is estimated by the Random Forest as per the Algorithm \ref{alg:cap} with $B=3000$ subsampling replications. We consider a fraction of $s=0.8$ observations in each subsample, $k=s^{1/6}$ and $\alpha=0.005$. The number of Monte Carlo replications is 500.

\subsection{Model Estimation}
 
The simulated models are determined by the general equation
 \begin{equation}
 Y_i=\beta_0(\bs{Z}_i)+\beta_1(\bs{Z}_i)X_i +U_i,
 \end{equation}
where $U_i$ is an independent and normally distributed random variable with zero mean and standard deviation equal to 0.5. The scalar continuous treatment variable is $X_i\sim\textsf{Uniform}(0,1)$. $\bs{Z}_i$ is a random vector of $d_Z$ mutually independent random variables uniformly distributed between 0 and 1. 

We set $\beta_0(Z_i)=0$ for all $i=1,\ldots,n$, and consider the following specifications for the slope parameter $\beta_1(\bs{Z}_i)$:
\begin{itemize}
\item{Model I: $p=1$}
\begin{align*}
\beta_1(\bs{Z}_i) = 5\left[1 + \frac{1}{1+e^{-100(Z_{i}-0.3)}}-\frac{1}{1+e^{-100(Z_{i}-0.7)}}\right].
\end{align*}
Figure \ref{fig:Simul1}(a) illustrates the evolution of $\beta_1(Z_i)$ as a function of $Z_i$. As can be observed from the plot, the slope coefficient smoothly changes between zero and one according to the source of heterogeneity $Z_i$.

\item{Model II: $p=2$}
\begin{align*}
\beta_1(\bs{Z}_i)=2 f_{i,1} f_{i,2} - 2 f_{i,1}  (1-f_{i,2}),
\end{align*}
where 
\begin{align*}
f_{i,1}&=1 + \frac{1}{1+e^{-100(Z_{1,i}-0.3)}}-\frac{1}{1+e^{-100(Z_{1,i}-0.7)}}\quad\text{and}\\
f_{i,2}&=1 + \frac{1}{1+e^{-100(Z_{2,i}-0.3)}}-\frac{1}{1+e^{-100(Z_{2,i}-0.7)}}
\end{align*}
Figure \ref{fig:Simul1}(b) illustrates the evolution of $\beta_1(\bs{Z}_i)$ as a function of $\bs{Z}_i$. As in the case of the first simulated model, the slope coefficient smoothly changes between different regimes.

\item{Model III: $p=3$}
\begin{align*}
\beta_1(\bs{Z}_i)= 2 f_{i,1} f_{i,2}  f_{i,3} - f_{i,1} f_{i,2} (1-f_{i,3}) - 1.5 f_{i,1} (1-f_{i,2}),
\end{align*}
where $f_{i,j}=\frac{1}{1+e^{-100(Z_{j,i}-0.5)}}$, for $j=1,2,3$ and $i=1,\ldots,n$.

The resulting model is a product of logistic functions. Figure \ref{fig:Simul1}(c) illustrates the heterogeneity in the slope coefficient as a function of $Z_1$ and $Z_2$ holding $Z_3$ fixed.

\item{Model IV: $p=5$}
\begin{align*}
\beta_1(\bs{Z}_i)& =f_{i,1} f_{i,2} f_{i,3} - f_{i,1} f_{i,2}(1-f_{i,3}) - 1.5*f_{i,1}(1-f_{i,2}) + 1.5(1-f_{i,1}) f_{i,4} \\&- 0.8(1-f_{i,1})(1-f_{i,4}) f_{i,5} + 0.7(1-f_{i,1})(1-f_{i,4})(1-f_{i,5})
\end{align*}
where $f_{i,j}=\frac{1}{1+e^{-100(Z_{j,i}-0.5)}}$, for $j=1,\ldots,5$ and $i=1,\ldots,n$.

As in the previous models, Model IV also implies a heterogeneity pattern where the slope coefficient smoothly changes among six limiting regimes.

\end{itemize}

The results are reported in Figure \ref{F:simul2} and Tables \ref{T:results1}--\ref{T:results5}. Figure \ref{F:simul2} reports results for Model I where $p=1$. It illustrates the median slope estimation across the Monte Carlo simulations, as well as the 95\% confidence bands. It is evident the estimation is precise in this case.

Table \ref{T:results1} presents the averages derived from the Monte Carlo simulations for various descriptive statistics pertinent to goodness-of-fit assessments. The evaluated statistics include the mean, standard deviation, kurtosis, and skewness. For Model I, Panel (I.a) examines the estimated residuals from the model fit, Panel (I.b) addresses the estimation error associated with the varying intercept of the model, and Panel (I.c) pertains to the estimation error for the varying slope coefficient. The subsequent panels replicate these findings for Models II, III, and IV.

Several facts emerge from the table. First, the model approximation is satisfactory even in samples as small as 250 observations. Notably, the average of the residual standard deviation is close to 0.5, which is the true value. Furthermore, the average estimated kurtosis is approximately three, and the average estimated skewness is close to zero, indicating that the residuals are approximately normally distributed. Additionally, the results suggest that the estimation of the varying intercept is more precise than the slope estimation, as the average standard deviation of the errors is significantly smaller for the former than for the latter. Finally, as expected, the performance of the method slightly deteriorates as the dimension of $\bs{Z}$ increases. 

Tables \ref{T:results2} and \ref{T:results3} present results for several test values pertinent to the vector $\bs{Z}$. We analyze 11 points where $\bs{Z}$ constitutes a $5 1$ that spans from a vector of zeros to a vector of ones, with $0.1$ increments. For simplicity and without loss of generality, we assume that all elements of the test points are equal. Table \ref{T:results2} reports the average bias and the mean squared error (MSE) for the varying intercept coefficient evaluated at each test point for different sample sizes. It is evident that the intercept estimation is precise and improves with increasing sample size. Table \ref{T:results3} reports the average bias and the mean squared error (MSE) for the varying slope coefficient evaluated at each test point for different sample sizes. The results in the table corroborate our previous conclusion that as the sample size increases, the method's performance improves. Finally, there is evidence that the function approximation is better for points closer to the center of the distribution of $\bs{Z}$.

\subsection{Inference}
We present coverage results in Tables \ref{T:results4} and \ref{T:results5}. Specifically, we provide the average 90\% and 95\% coverage across Monte Carlo simulations for both the intercept and the slope parameters. Table \ref{T:results4} details the coverage for the intercept, while \ref{T:results5} presents the corresponding results for the slope parameter. It is evident from the tables that the coverage probabilities for the intercept are close to the expected values, even with small sample sizes and an increased number of variables. This finding supports our previous simulation results, indicating minimal bias in intercept estimation. Conversely, the coverage for the slope parameter is significantly underestimated, particularly for values that lie farther from the center of the covariate distribution. This observation aligns with the larger biases reported in Table \ref{T:results3}. It is important to mention that such narrow coverage probabilities have also been reported in the simulations in \citet{rFjTsAsW2021}.

\subsection{Linarity Testing}
Finally, to evaluate the finite sample performance of the Lagrange Multiplier homogeneity (linearity) test described in Section \ref{S:LM}. We simulated linear (homogeneous) models with $p\in\{1,2,3,5\}$ covariates. The results are reported in Figure \ref{F:size}, which illustrates the size discrepancy relative to the nominal size. As observed, the size distortions remain negligible. 

\section{Empirical Illustrations}\label{S:Empirical}
In this section, we illustrate our proposed methodology utilizing two distinct datasets. The first is a synthetic dataset generated by a widely used model specifically designed to evaluate nonparametric methods. The second is a real dataset concerning the economic convergence of Brazilian municipalities.

\subsection{Simulated data}

The first dataset consists of observations generated according to the following model:
\begin{align}
Y_i &= \beta(\bs{Z}_i)X_i + U_i,\quad i=1,\ldots,n\\
\beta(\bs{Z}_i)&=10\sin(\pi Z_{1,i} Z_{2i})+20(Z_{3,i}-0.5)^2+10 Z_{4,i}+5 Z_{5,i},
\end{align}
where $\bs{Z}_i$ is a vector of mutually independent uniform random variables taking values on $[0,1]^5$, and $U_i$ is a zero-mean normally distributed random variable with unit variance. The model for $\beta(\bs{Z}_i)$ has been widely employed in the literature to evaluate semi-parametric models. In this context, heterogeneity is jointly influenced by interactions, quadratic forms, and a robust linear signal. Figure \ref{F:emp1} illustrates the three sources of heterogeneity.
See, for example, \citet{rFjTsAsW2021} for a similar data-generating process.     

To evaluate the performance of the model presented in this paper, we consider various sample sizes (2000, 4000, 8000, 16000, 32000). Figure \ref{F:empdist} illustrates the empirical distribution of $\bs\beta(\bs Z_i)$, $i=1, \ldots, 32000$. We estimate the locally linear random forest model using the same hyperparameter settings as those utilized in the Monte Carlo simulation.

Figure \ref{F:scatter} shows results concerning the estimation parameters. Panels (a) and (b) illustrate the scatter plot of the fitted $\bs\beta(\bs{Z})$ against the true values for $n=2000$ and $n=36000$, respectively. A linear regression line is also included. Figure \ref{F:biasmse} reports the evolution of the bias and the MSE as a function of the sample size. As we can observe from both figures, the estimation improves as the sample size increases, which aligns with established statistical theory. However, as anticipated by our theoretical results, the convergence to the true values is notably slow.

\subsection{Economic growth and convergence among Brazilian municipalities}

We illustrate our methodology by testing heterogeneity in the convergence among Brazilian municipalities between 1970 and 2000. Our starting point is the simplified convergence equation presented in \citet{rjBxS1992}:
\begin{equation}\label{converg}
\log \left( \frac{Y_{i,t}}{Y_{i,t-1}}\right) =a_{i}+\gamma_i \log
\left(Y_{i,t-1}\right)+\phi_{i}\left( t-1\right)+U_{i,t},
\end{equation}
where $Y_{i,t}$ is the per capita income of region $i$ in period $t$, $a_{i}$ is associated with the steady-state level of per capita income and the rate of technological progress, $\phi_{i}$ is a parameter related to the time trend determined by the technological progress, and $U_{i,t}$ is the random term. Convergence corresponds to the parameter $\gamma_i$.

From a conceptual point of view, two alternative assumptions determine the most important distinction of convergence concepts. First, we can assume that $a_{i}=a$, $\gamma_i=\gamma$, and $\phi_{i}=\phi $, i.e., that the basic parameters of preference and technology are the same for all economies represented in the sample. This is when $\gamma<0$ represents \emph{unconditional convergence} - a situation where poorer regions tend unconditionally to grow more quickly than richer ones. Alternatively, we can state a weaker assumption, allowing for possible differences in the steady state across the economies considered and heterogeneity in the convergence rate. In terms of equation (\ref{converg}), $a_{i}$, $\gamma_i$ and $\phi _{i}$ are allowed to vary across different regions. 

We estimate (\ref{converg}) in a cross-section setup, where there is no identifiable time trend, and we are not able to distinguish between $\phi_{i}$ and $a_{i}$. Thus, we estimate the following equation:
\begin{equation}\label{converg2}
\log \left( \frac{Y_{i,2000}}{Y_{i,1970}}\right) =\alpha_{i}+\gamma_i\log\left(Y_{i,1970}\right)+U_{i,2000}.  
\end{equation}

The data originate from the Brazilian Demographic Censuses conducted in 1970 and 2000. The geographical units were adjusted to account for the reorganization of Brazilian territory throughout this time frame. In 1970, Brazil was composed of 3,951 municipalities. By 2000, the count had increased to 5,507 municipalities. Consequently, all data collected in 2000 were aggregated to align with the municipal structure as it existed in 1970. Our dependent variable is defined as the average growth in per capita income from 1970 to 2000 for each municipality. The independent variable utilized in this analysis is the logarithm of the per capita income level recorded in 1970. 

Figure \ref{fig:geo1} illustrates the differences across municipalities in terms of growth rate. Figure \ref{fig:geo1} shows that the variations in growth do not coincide with the administrative state frontiers. There are substantial variations within many of the Brazilian states.

We employ the semi-parametric approach outlined in the previous section to assess conditional convergence. Variations in preferences and technological parameters are endogenously incorporated into the analysis through geographical proximities. The underlying assumption suggests that cities situated in close proximity experience similar steady states. Within our modeling framework, we estimate (\ref{converg2}), recognizing that $\alpha _{i}$ represents a semi-parametric function of the latitude and longitude coordinates. The results are displayed in Figures \ref{fig:geo2} and \ref{fig:geo3}, which report the heterogeneity patterns in the intercept and the convergence parameter. The estimated geographic heterogeneity varies remarkably, with notable clusters of high-income municipalities in the central and southern parts of the country. Significant geographic differences across municipalities do not coincide with the geographic structure of the Brazilian states.

\section{Conclusions}\label{S:Conclusions}
This paper presents a semi-parametric framework for estimating heterogeneous partial effects, combining the flexibility of machine learning techniques with the interpretability of traditional parametric models. By utilizing a Random Forest-based methodology, we provide a robust and adaptable approach to capturing complex heterogeneity in the relationship between explanatory variables and outcomes.

Our theoretical contributions establish key consistency and asymptotic normality results, ensuring the reliability of our estimator. Importantly, our framework accommodates both continuous and discrete covariates while maintaining desirable statistical properties. The Monte Carlo simulations demonstrate the method’s accuracy, even in moderate sample sizes, and highlight the precision of our approach in recovering varying intercepts and slopes. Additionally, our empirical analysis of Brazilian municipal economic convergence underscores the practical relevance of our method, revealing substantial geographic heterogeneity in growth dynamics.

Future research can extend this methodology to settings with dependent data and high-dimensional covariates, broadening its applicability to dynamic panel models and network data. The current approach can accommodate high-dimensional controls ($\bs{Z}$) under additional conditions on the underlying data-generating process. Specifically, it would require that (i) only a small number of controls are relevant to the model (sparsity assumption) and (ii) the relevant controls are independent of the irrelevant controls and the outcome variable. The latter is a strong assumption in most applications, so we did not pursue this route here. 

Regarding the dependent data across observations, the subsampling and sample split steps must be adjusted to preserve both the data dependence structure and the independence of the two samples. A block-subsampling technique is a natural candidate; however, implementing this would significantly alter the proof techniques for the subsequent steps and might obscure the main idea of the paper. Furthermore, refinements in inference procedures for heterogeneous effects, including the construction of a uniformly valid confidence interval, would enhance the model’s applicability in applied research. Additionally, considering Lemma \ref{lem:upper_bound_diameter}(d), a significant modification to the algorithm would be necessary to achieve uniform convergence of the proposed estimators.
 
Overall, this paper offers a flexible, interpretable, and computationally efficient tool for studying heterogeneity in the partial effect of a variable of interest, bridging the gap between parametric and nonparametric estimation. 

\begin{table}[H]
\caption{Godness-of-Fit}
\label{T:results1}
\centering{}
\resizebox{\textwidth}{!}{
\begin{threeparttable}
This table reports the average across the Monte Carlo simulations for several descriptive statistics related to goodness-of-fit. The statistics are the mean, the standard deviation, the kurtosis, and skewness. Panel (I.a) considers the estimated residuals from the model fit. Panel (I.b) considers the estimation error for the varying intercept of the model. Panel (I.c) is related to the estimation error for the varying slope coefficient. The remaining panels show the same results for models II, III, and IV.
\begin{tabular}{@{}lccccccccccccccc@{}}
\toprule
& \multicolumn{3}{c}{Panel (I.a): $p=1$} && \multicolumn{3}{c}{Panel (II.a): $p=2$} && \multicolumn{3}{c}{Panel (III.a): $p=3$} && \multicolumn{3}{c}{Panel (IV.a): $p=5$} \\
& \multicolumn{3}{c}{\underline{\textbf{error term}}} && \multicolumn{3}{c}{\underline{\textbf{error term}}} && \multicolumn{3}{c}{\underline{\textbf{error term}}} && \multicolumn{3}{c}{\underline{\textbf{error term}}}\\
          & \multicolumn{3}{c}{Sample Size} && \multicolumn{3}{c}{Sample Size} && \multicolumn{3}{c}{Sample Size} && \multicolumn{3}{c}{Sample Size}\\
          \cmidrule(lr){2-4}
          \cmidrule(lr){6-8}
          \cmidrule(lr){10-12}
          \cmidrule(lr){14-16}
          & 250 & 500 & 1000 && 250 & 500 & 1000 && 250 & 500 & 1000 && 250 & 500 & 1000\\
\midrule
mean               &  -0.0019  & -0.0002 &   0.0001 &&  -0.0087 &  -0.0069 &  -0.0047 &&   0.0024 &   0.0014  &  0.0009  &&  0.0004  &  0.0002 &   0.0001 \\
standard deviation &   0.4506  &  0.4458 &   0.4568 &&   0.6207 &   0.5058 &   0.4807 &&   0.4750 &   0.4533  &  0.4571  &&  0.5638  &  0.5235 &   0.5154 \\
kurtosis           &   3.0154  &  3.0182 &   3.0346 &&   3.9624 &   3.3529 &   3.0946 &&   3.0793 &   3.1119  &  3.0315  &&  2.8920  &  2.9941 &   2.9761 \\
skewness           &   0.0128  &  0.0015 &   0.0027 &&   0.5457 &   0.2077 &   0.0683 &&   0.0710 &   0.0337  &  0.0182  &&  0.0102  &  0.0095 &   0.0126 \\
\midrule
& \multicolumn{3}{c}{Panel (I.b): $p=1$} && \multicolumn{3}{c}{Panel (II.b): $p=2$} && \multicolumn{3}{c}{Panel (III.b): $p=3$} && \multicolumn{3}{c}{Panel (IV.b): $p=5$} \\
& \multicolumn{3}{c}{\underline{\textbf{intercept}}} && \multicolumn{3}{c}{\underline{\textbf{intercept}}} && \multicolumn{3}{c}{\underline{\textbf{intercept}}} && \multicolumn{3}{c}{\underline{\textbf{intercept}}}\\
          & \multicolumn{3}{c}{Sample Size} && \multicolumn{3}{c}{Sample Size} && \multicolumn{3}{c}{Sample Size} && \multicolumn{3}{c}{Sample Size}\\
          \cmidrule(lr){2-4}
          \cmidrule(lr){6-8}
          \cmidrule(lr){10-12}
          \cmidrule(lr){14-16}
          & 250 & 500 & 1000 && 250 & 500 & 1000 && 250 & 500 & 1000 && 250 & 500 & 1000\\
\midrule
mean               &    -0.0077  & -0.0004  & -0.0019  &&  0.0044  & -0.0037  &  0.0028 &&  -0.0049 &   0.0010  &  0.0001  &&  0.0023 &  -0.0024  & -0.0044 \\
standard deviation &     0.2396  &  0.2383  &  0.1967  &&  0.2483  &  0.2179  &  0.1642 &&   0.1609 &   0.1544  &  0.1256  &&  0.1540 &   0.1460  &  0.1089 \\
kurtosis           &     3.0068  &  3.2621  &  3.0542  &&  3.3724  &  3.4301  &  3.0503 &&   2.9366 &   3.1829  &  3.0051  &&  3.0813 &   3.3344  &  3.0628 \\
skewness           &    -0.0370  &  0.0671  & -0.0195  &&  0.0905  & -0.0007  &  0.0072 &&  -0.0084 &   0.0119  &  0.0058  && -0.0173 &  -0.0068  & -0.0001 \\
\bottomrule
& \multicolumn{3}{c}{Panel (I.c): $p=1$} && \multicolumn{3}{c}{Panel (II.c): $p=2$} && \multicolumn{3}{c}{Panel (III.c): $p=3$} && \multicolumn{3}{c}{Panel (IV.c): $p=5$} \\
& \multicolumn{3}{c}{\underline{\textbf{slope}}} && \multicolumn{3}{c}{\underline{\textbf{slope}}} && \multicolumn{3}{c}{\underline{\textbf{slope}}} && \multicolumn{3}{c}{\underline{\textbf{slope}}}\\
          & \multicolumn{3}{c}{Sample Size} && \multicolumn{3}{c}{Sample Size} && \multicolumn{3}{c}{Sample Size} && \multicolumn{3}{c}{Sample Size}\\
          \cmidrule(lr){2-4}
          \cmidrule(lr){6-8}
          \cmidrule(lr){10-12}
          \cmidrule(lr){14-16}
          & 250 & 500 & 1000 && 250 & 500 & 1000 && 250 & 500 & 1000 && 250 & 500 & 1000\\
\midrule
mean               &     0.0137 &  -0.0013  &  0.0021 &&  -0.0237 &  -0.0061 &  -0.0101  &&  0.0168  &  0.0003 &   0.0034 &&  -0.0020  &  0.0000  &  0.0126\\
standard deviation &     0.4567 &   0.4239  &  0.3499 &&   0.9791 &   0.6639 &   0.4672  &&  0.4841  &  0.3962 &   0.3019 &&   0.7661  &  0.6474  &  0.5455\\
kurtosis           &     3.7132 &   3.3192  &  3.0741 &&   4.0835 &   4.1379 &   3.8797  &&  3.2178  &  3.3337 &   3.3105 &&   1.8627  &  2.0687  &  1.9109\\
skewness           &-   -0.0132 &  -0.0308  & -0.0121 &&   0.8756 &   0.5659 &   0.3234  &&  0.2993  &  0.2660 &   0.2157 &&   0.0338  &  0.0724  &  0.0331\\
\bottomrule
\end{tabular}
\end{threeparttable}
}
\end{table}       

\begin{table}[H]
\caption{Test Points (Goodness-of-Fit -- Intercept)}
\label{T:results2}
\centering{}
\resizebox{\textwidth}{!}{
\begin{threeparttable}
The table shows the average bias and mean squared error for the intercept estimation at different test points (MSE). The first column shows the points considered. For the case where $p>1$, all the values of $\bs{Z}$ are equal to the number indicated in the first column. Panel (I) considers the case of Model I ($p=1$). Panel (II) considers the case of Model II ($p=2$).  Panel (III) considers the case of Model III ($p=3$).  Panel (IV) considers the case of Model IV ($p=5$).   
\begin{tabular}{@{}lccccccccccccccc@{}}
\toprule
& \multicolumn{7}{c}{\underline{Panel I: $p=1$}} && \multicolumn{7}{c}{\underline{Panel II: $p=2$}}\\
&\multicolumn{3}{c}{\underline{(a): Bias}} && \multicolumn{3}{c}{\underline{(b): MSE}} && \multicolumn{3}{c}{\underline{(a): Bias}} && \multicolumn{3}{c}{\underline{(b): MSE}}\\
& \multicolumn{3}{c}{Sample Size} && \multicolumn{3}{c}{Sample Size} && \multicolumn{3}{c}{Sample Size} && \multicolumn{3}{c}{Sample Size}\\
          \cmidrule(lr){2-4}
          \cmidrule(lr){6-8}
          \cmidrule(lr){10-12}
          \cmidrule(lr){14-16}
Point     & 250 & 500 & 1000 && 250 & 500 & 1000 &&  250 & 500 & 1000 && 250 & 500 & 1000\\
\midrule
$0$   &    0.0561  &   0.0062 &   -0.0111  & &   0.0927  &   0.0702  &   0.0563  & &   0.0028 &    0.0016 &    0.0051  & &   0.0766 &    0.0644  &   0.0321\\
$0.1$ &    0.0024  &   0.0102 &   -0.0312  & &   0.0563  &   0.0510  &   0.0424  & &   0.0037 &   -0.0014 &    0.0049  & &   0.0531 &    0.0462  &   0.0245\\
$0.2$ &    0.0081  &   0.0139 &    0.0044  & &   0.0533  &   0.0566  &   0.0407  & &  -0.0080 &   -0.0053 &    0.0269  & &   0.0646 &    0.0416  &   0.0291\\
$0.3$ &   -0.0196  &  -0.0109 &    0.0289  & &   0.1139  &   0.0811  &   0.0441  & &   0.0161 &   -0.0263 &    0.0114  & &   0.1096 &    0.0701  &   0.0320\\
$0.4$ &   -0.0004  &  -0.0257 &    0.0134  & &   0.0574  &   0.0503  &   0.0348  & &   0.0035 &    0.0017 &    0.0075  & &   0.1097 &    0.0543  &   0.0243\\
$0.5$ &    0.0001  &  -0.0113 &   -0.0008  & &   0.0487  &   0.0590  &   0.0373  & &   0.0145 &   -0.0040 &    0.0276  & &   0.0722 &    0.0404  &   0.0219\\
$0.6$ &   -0.0090  &   0.0064 &    0.0146  & &   0.0645  &   0.0529  &   0.0425  & &  -0.0115 &   -0.0137 &   -0.0057  & &   0.1027 &    0.0626  &   0.0250\\
$0.7$ &   -0.0143  &  -0.0076 &   -0.0235  & &   0.1413  &   0.0598  &   0.0447  & &  -0.0252 &    0.0074 &    0.0050  & &   0.0917 &    0.0696  &   0.0254\\
$0.8$ &   -0.0086  &  -0.0264 &   -0.0281  & &   0.0644  &   0.0623  &   0.0399  & &   0.0177 &   -0.0182 &    0.0001  & &   0.0578 &    0.0334  &   0.0255\\
$0.9$ &   -0.0313  &   0.0170 &    0.0023  & &   0.0650  &   0.0493  &   0.0505  & &   0.0205 &   -0.0040 &    0.0173  & &   0.0539 &    0.0380  &   0.0232\\
$1$   &   -0.0393  &   0.0144 &   -0.0029  & &   0.0761  &   0.0779  &   0.0569  & &   0.0426 &   -0.0113 &    0.0159  & &   0.0600 &    0.0567  &   0.0311\\
\midrule
& \multicolumn{7}{c}{\underline{Panel III: $p=3$}} && \multicolumn{7}{c}{\underline{Panel IV: $p=5$}}\\
&\multicolumn{3}{c}{\underline{(a): Bias}} && \multicolumn{3}{c}{\underline{(b): MSE}} && \multicolumn{3}{c}{\underline{(a): Bias}} && \multicolumn{3}{c}{\underline{(b): MSE}}\\
& \multicolumn{3}{c}{Sample Size} && \multicolumn{3}{c}{Sample Size} && \multicolumn{3}{c}{Sample Size} && \multicolumn{3}{c}{Sample Size}\\
          \cmidrule(lr){2-4}
          \cmidrule(lr){6-8}
          \cmidrule(lr){10-12}
          \cmidrule(lr){14-16}
Point     & 250 & 500 & 1000 && 250 & 500 & 1000 &&  250 & 500 & 1000 && 250 & 500 & 1000\\
\midrule
$0$   &    -0.0195  &  -0.0147  &  -0.0041 & &    0.0282  &   0.0269 &    0.0160  & &  -0.0064 &    0.0031  &  -0.0153 &&     0.0344  &   0.0238  &   0.0123 \\
$0.1$ &    -0.0159  &  -0.0130  &   0.0034 & &    0.0249  &   0.0216 &    0.0157  & &  -0.0096 &    0.0060  &  -0.0123 &&     0.0323  &   0.0245  &   0.0128 \\
$0.2$ &    -0.0035  &  -0.0044  &  -0.0055 & &    0.0222  &   0.0223 &    0.0127  & &  -0.0041 &    0.0064  &  -0.0038 &&     0.0264  &   0.0201  &   0.0095 \\
$0.3$ &     0.0066  &   0.0028  &  -0.0015 & &    0.0209  &   0.0251 &    0.0138  & &  -0.0060 &    0.0046  &   0.0024 &&     0.0198  &   0.0164  &   0.0087 \\
$0.4$ &-    0.0110  &   0.0109  &   0.0026 & &    0.0189  &   0.0155 &    0.0119  & &  -0.0018 &    0.0073  &  -0.0055 &&     0.0185  &   0.0123  &   0.0069 \\
$0.5$ &     0.0172  &   0.0062  &  -0.0023 & &    0.0194  &   0.0123 &    0.0083  & &   0.0146 &   -0.0033  &  -0.0092 &&     0.0211  &   0.0095  &   0.0048 \\
$0.6$ &    -0.0011  &  -0.0174  &   0.0009 & &    0.0254  &   0.0180 &    0.0140  & &   0.0215 &    0.0076  &   0.0033 &&     0.0189  &   0.0141  &   0.0087 \\
$0.7$ &    -0.0233  &  -0.0113  &  -0.0103 & &    0.0265  &   0.0171 &    0.0156  & &   0.0164 &    0.0051  &   0.0152 &&     0.0192  &   0.0136  &   0.0099 \\
$0.8$ &    -0.0341  &  -0.0205  &  -0.0124 & &    0.0316  &   0.0221 &    0.0163  & &   0.0035 &   -0.0195  &   0.0062 &&     0.0237  &   0.1428  &   0.0104 \\
$0.9$ &    -0.0216  &  -0.0147  &  -0.0046 & &    0.0409  &   0.0254 &    0.0204  & &   0.0115 &   -0.0038  &   0.0027 &&     0.0255  &   0.0206  &   0.0112 \\
$1$   &    -0.0178  &  -0.0125  &  -0.0067 & &    0.0380  &   0.0342 &    0.0199  & &   0.0072 &   -0.0036  &   0.0003 &&     0.0268  &   0.0199  &   0.0110 \\
\bottomrule
\end{tabular}
\end{threeparttable}
}
\end{table}

\begin{table}[H]
\caption{Test Points (Goodness-of-Fit -- Slope)}
\label{T:results3}
\centering{}
\resizebox{\textwidth}{!}{
\begin{threeparttable}
The table shows the average bias and mean squared error for the intercept estimation at different test points (MSE). The first column shows the points considered. For the case where $p>1$, all the values of $\bs{Z}$ are equal to the number indicated in the first column. Panel (I) considers the case of Model I ($p=1$). Panel (II) considers the case of Model II ($p=2$).  Panel (III) considers the case of Model III ($p=3$).  Panel (IV) considers the case of Model IV ($p=5$).
\begin{tabular}{@{}lccccccccccccccc@{}}
\toprule
& \multicolumn{7}{c}{\underline{Panel I: $p=1$}} && \multicolumn{7}{c}{\underline{Panel II: $p=2$}}\\
&\multicolumn{3}{c}{\underline{(a): Bias}} && \multicolumn{3}{c}{\underline{(b): MSE}} && \multicolumn{3}{c}{\underline{(a): Bias}} && \multicolumn{3}{c}{\underline{(b): MSE}}\\
& \multicolumn{3}{c}{Sample Size} && \multicolumn{3}{c}{Sample Size} && \multicolumn{3}{c}{Sample Size} && \multicolumn{3}{c}{Sample Size}\\
          \cmidrule(lr){2-4}
          \cmidrule(lr){6-8}
          \cmidrule(lr){10-12}
          \cmidrule(lr){14-16}
Point     & 250 & 500 & 1000 && 250 & 500 & 1000 &&  250 & 500 & 1000 && 250 & 500 & 1000\\
\midrule
$0$   &   -0.0693  &   0.0344  &   0.0421  &&    0.2735  &   0.2201 &    0.1663 & &   -0.6764 &   -0.4972 &   -0.4115  & &   0.2912 &    0.2247  &   0.1126\\
$0.1$ &   -0.0095  &  -0.0085  &   0.0464  &&    0.1532  &   0.1674 &    0.1096 & &   -0.5533 &   -0.2788 &   -0.1789  & &   0.1990 &    0.1399  &   0.0800\\
$0.2$ &   -0.0310  &  -0.0413  &   0.0198  &&    0.1829  &   0.1606 &    0.1260 & &   -0.7367 &   -0.3826 &   -0.2830  & &   0.2374 &    0.1470  &   0.0787\\
$0.3$ &   -0.4974  &  -0.2378  &  -0.1641  &&    0.5734  &   0.3932 &    0.2077 & &   -0.8960 &   -0.9171 &   -1.0405  & &   0.9729 &    0.7729  &   0.5702\\
$0.4$ &    0.0131  &   0.0427  &  -0.0127  &&    0.1655  &   0.1768 &    0.0957 & &    1.8125 &    0.9784 &    0.5650  & &   0.5971 &    0.2518  &   0.0853\\
$0.5$ &    0.0092  &   0.0211  &   0.0164  &&    0.1530  &   0.1779 &    0.1128 & &    1.3114 &    0.6609 &    0.3321  & &   0.3396 &    0.1632  &   0.0669\\
$0.6$ &    0.0331  &  -0.0032  &  -0.0280  &&    0.1818  &   0.1641 &    0.1041 & &    1.8445 &    0.9141 &    0.5676  & &   0.6678 &    0.2635  &   0.0929\\
$0.7$ &    0.2560  &   0.0292  &  -0.0460  &&    0.8037  &   0.3561 &    0.2422 & &    0.7413 &    0.3897 &   -0.0711  & &   0.7615 &    0.8367  &   0.5368\\
$0.8$ &   -0.0146  &   0.0259  &   0.0563  &&    0.1765  &   0.1954 &    0.1028 & &   -0.8750 &   -0.4580 &   -0.2882  & &   0.2634 &    0.1475  &   0.0786\\
$0.9$ &    0.0513  &  -0.0293  &  -0.0013  &&    0.1968  &   0.1386 &    0.1509 & &   -0.6706 &   -0.3357 &   -0.2046  & &   0.2004 &    0.1341  &   0.0766\\
$1$   &    0.0660  &  -0.0088  &   0.0298  &&    0.2547  &   0.2396 &    0.1601 & &   -0.7933 &   -0.5480 &   -0.4556  & &   0.2332 &    0.1989  &   0.1218\\
\midrule
& \multicolumn{7}{c}{\underline{Panel III: $p=3$}} && \multicolumn{7}{c}{\underline{Panel IV: $p=5$}}\\
&\multicolumn{3}{c}{\underline{(a): Bias}} && \multicolumn{3}{c}{\underline{(b): MSE}} && \multicolumn{3}{c}{\underline{(a): Bias}} && \multicolumn{3}{c}{\underline{(b): MSE}}\\
& \multicolumn{3}{c}{Sample Size} && \multicolumn{3}{c}{Sample Size} && \multicolumn{3}{c}{Sample Size} && \multicolumn{3}{c}{Sample Size}\\
          \cmidrule(lr){2-4}
          \cmidrule(lr){6-8}
          \cmidrule(lr){10-12}
          \cmidrule(lr){14-16}
Point     & 250 & 500 & 1000 && 250 & 500 & 1000 &&  250 & 500 & 1000 && 250 & 500 & 1000\\
\midrule
$0$   &     0.2818  &   0.1940 &    0.1568  & &   0.0847  &   0.0828 &    0.0457  & &   0.5893 &    0.4986  &   0.5042 & &    0.1206 &    0.0853  &   0.0447 \\
$0.1$ &     0.2228  &   0.1240 &    0.0670  & &   0.0819  &   0.0647 &    0.0477  & &   0.5565 &    0.4389  &   0.4202 & &    0.1170 &    0.0895  &   0.0423 \\
$0.2$ &     0.1771  &   0.0910 &    0.0604  & &   0.0727  &   0.0639 &    0.0361  & &   0.5083 &    0.3906  &   0.3521 & &    0.1004 &    0.0705  &   0.0331 \\
$0.3$ &     0.1643  &   0.1040 &    0.0662  & &   0.0647  &   0.0769 &    0.0428  & &   0.5059 &    0.4038  &   0.3476 & &    0.0744 &    0.0585  &   0.0321 \\
$0.4$ &-    0.2247  &   0.1486 &    0.1013  & &   0.0607  &   0.0468 &    0.0331  & &   0.5655 &    0.4792  &   0.4418 & &    0.0767 &    0.0535  &   0.0317 \\
$0.5$ &    -0.0243  &  -0.0329 &    0.0215  & &   0.0703  &   0.0607 &    0.0491  & &   0.0192 &    0.0441  &   0.0392 & &    0.0905 &    0.0815  &   0.0586 \\
$0.6$ &     1.0165  &   0.8154 &    0.5452  & &   0.1167  &   0.0768 &    0.0478  & &   0.9135 &    0.8141  &   0.6863 & &    0.0959 &    0.0683  &   0.0428 \\
$0.7$ &     0.7940  &   0.5273 &    0.3219  & &   0.1142  &   0.0590 &    0.0499  & &   0.8217 &    0.6647  &   0.5278 & &    0.0884 &    0.0575  &   0.0414 \\
$0.8$ &     0.7404  &   0.4528 &    0.2789  & &   0.1047  &   0.0748 &    0.0483  & &   0.8040 &    0.6659  &   0.5314 & &    0.1012 &    0.2341  &   0.0425 \\
$0.9$ &     0.7893  &   0.5353 &    0.3340  & &   0.1477  &   0.0838 &    0.0621  & &   0.8085 &    0.6992  &   0.6035 & &    0.1139 &    0.0817  &   0.0469 \\
$1$   &     0.8816  &   0.7217 &    0.5808  & &   0.1392  &   0.1114 &    0.0648  & &   0.8349 &    0.7665  &   0.7086 & &    0.1165 &    0.0792  &   0.0465 \\
\bottomrule
\end{tabular}
\end{threeparttable}
}
\end{table}

\begin{table}[H]
\caption{Test Points (Coverage -- Intercept)}
\label{T:results4}
\centering{}
\resizebox{\textwidth}{!}{
\begin{threeparttable}
The table shows the average coverage for the intercept estimation at different test points (MSE). The first column shows the points considered. For the case where $p>1$, all the values of $\bs{Z}$ are equal to the number indicated in the first column. Panel (I) considers the case of Model I ($p=1$). Panel (II) considers the case of Model II ($p=2$).  Panel (III) considers the case of Model III ($p=3$).  Panel (IV) considers the case of Model IV ($p=5$).
\begin{tabular}{@{}lccccccccccccccc@{}}
\toprule
& \multicolumn{7}{c}{\underline{Panel I: $p=1$}} && \multicolumn{7}{c}{\underline{Panel II: $p=2$}}\\
&\multicolumn{3}{c}{\underline{(a): 90\%}} && \multicolumn{3}{c}{\underline{(b): 95\%}} && \multicolumn{3}{c}{\underline{(a): 90\%}} && \multicolumn{3}{c}{\underline{(b): 95\%}}\\
& \multicolumn{3}{c}{Sample Size} && \multicolumn{3}{c}{Sample Size} && \multicolumn{3}{c}{Sample Size} && \multicolumn{3}{c}{Sample Size}\\
          \cmidrule(lr){2-4}
          \cmidrule(lr){6-8}
          \cmidrule(lr){10-12}
          \cmidrule(lr){14-16}
Point     & 250 & 500 & 1000 && 250 & 500 & 1000 &&  250 & 500 & 1000 && 250 & 500 & 1000\\
\midrule
$0$   &    0.9000  &   0.8900  &   0.8750  & &   0.9400 &    0.9650  &   0.9300 & &    0.9050  &   0.9100  &   0.9000  & &   0.9350 &    0.9500  &   0.9650\\
$0.1$ &    0.9000  &   0.9150  &   0.8800  & &   0.9500 &    0.9500  &   0.9300 & &    0.8750  &   0.9200  &   0.8900  & &   0.9450 &    0.9550  &   0.9250\\
$0.2$ &    0.9150  &   0.8750  &   0.9050  & &   0.9650 &    0.9500  &   0.9500 & &    0.9100  &   0.9050  &   0.9000  & &   0.9550 &    0.9450  &   0.9500\\
$0.3$ &    0.8950  &   0.9000  &   0.9000  & &   0.9450 &    0.9350  &   0.9550 & &    0.8950  &   0.8900  &   0.8950  & &   0.9450 &    0.9450  &   0.9600\\
$0.4$ &    0.9250  &   0.8850  &   0.9000  & &   0.9700 &    0.9800  &   0.9500 & &    0.9350  &   0.8950  &   0.9300  & &   0.9450 &    0.9550  &   0.9650\\
$0.5$ &    0.9150  &   0.9000  &   0.8950  & &   0.9600 &    0.9600  &   0.9400 & &    0.9050  &   0.9050  &   0.8850  & &   0.9500 &    0.9500  &   0.9450\\
$0.6$ &    0.8850  &   0.8900  &   0.8850  & &   0.9650 &    0.9500  &   0.9700 & &    0.9000  &   0.9100  &   0.8850  & &   0.9550 &    0.9500  &   0.9450\\
$0.7$ &    0.9350  &   0.9100  &   0.9200  & &   0.9600 &    0.9350  &   0.9650 & &    0.8950  &   0.8700  &   0.9100  & &   0.9350 &    0.9400  &   0.9550\\
$0.8$ &    0.9050  &   0.9100  &   0.8950  & &   0.9350 &    0.9450  &   0.9450 & &    0.8850  &   0.8850  &   0.9100  & &   0.9500 &    0.9500  &   0.9550\\
$0.9$ &    0.8950  &   0.8900  &   0.8850  & &   0.9350 &    0.9500  &   0.9500 & &    0.8900  &   0.9050  &   0.8950  & &   0.9500 &    0.9600  &   0.9600\\
$1$   &    0.8950  &   0.9200  &   0.9000  & &   0.9400 &    0.9550  &   0.9550 & &    0.9000  &   0.9100  &   0.8900  & &   0.9350 &    0.9550  &   0.9550\\
\midrule
& \multicolumn{7}{c}{\underline{Panel III: $p=3$}} && \multicolumn{7}{c}{\underline{Panel IV: $p=5$}}\\
&\multicolumn{3}{c}{\underline{(a): 90\%}} && \multicolumn{3}{c}{\underline{(b): 95\%}} && \multicolumn{3}{c}{\underline{(a): 90\%}} && \multicolumn{3}{c}{\underline{(b): 95\%}}\\
& \multicolumn{3}{c}{Sample Size} && \multicolumn{3}{c}{Sample Size} && \multicolumn{3}{c}{Sample Size} && \multicolumn{3}{c}{Sample Size}\\
          \cmidrule(lr){2-4}
          \cmidrule(lr){6-8}
          \cmidrule(lr){10-12}
          \cmidrule(lr){14-16}
Point     & 250 & 500 & 1000 && 250 & 500 & 1000 &&  250 & 500 & 1000 && 250 & 500 & 1000\\
\midrule
$0$   &    0.9100 &    0.8800  &   0.9000 & &    0.9300  &   0.9250  &   0.9500  & &   0.9000  &   0.8750  &   0.9150 & &    0.9300 &    0.9400  &   0.9400\\
$0.1$ &    0.9150 &    0.8950  &   0.8900 & &    0.9400  &   0.9450  &   0.9450  & &   0.8900  &   0.9000  &   0.9100 & &    0.9250 &    0.9500  &   0.9500\\
$0.2$ &    0.9050 &    0.8950  &   0.9000 & &    0.9600  &   0.9450  &   0.9600  & &   0.8850  &   0.8950  &   0.9050 & &    0.9450 &    0.9650  &   0.9550\\
$0.3$ &    0.8950 &    0.9200  &   0.8950 & &    0.9500  &   0.9500  &   0.9300  & &   0.8950  &   0.8850  &   0.9050 & &    0.9350 &    0.9550  &   0.9600\\
$0.4$ &    0.9050 &    0.9050  &   0.9100 & &    0.9300  &   0.9550  &   0.9450  & &   0.8850  &   0.8900  &   0.9150 & &    0.9550 &    0.9450  &   0.9450\\
$0.5$ &    0.9150 &    0.8900  &   0.9050 & &    0.9650  &   0.9550  &   0.9600  & &   0.9300  &   0.9100  &   0.8800 & &    0.9700 &    0.9400  &   0.9500\\
$0.6$ &    0.9100 &    0.9100  &   0.8900 & &    0.9550  &   0.9650  &   0.9700  & &   0.8800  &   0.9050  &   0.9050 & &    0.9600 &    0.9450  &   0.9500\\
$0.7$ &    0.9100 &    0.9050  &   0.9050 & &    0.9600  &   0.9600  &   0.9450  & &   0.9100  &   0.9150  &   0.9000 & &    0.9450 &    0.9600  &   0.9450\\
$0.8$ &    0.8850 &    0.9100  &   0.9050 & &    0.9450  &   0.9550  &   0.9550  & &   0.9000  &   0.9950  &   0.8900 & &    0.9350 &    0.9950  &   0.9400\\
$0.9$ &    0.9350 &    0.9150  &   0.9200 & &    0.9750  &   0.9350  &   0.9600  & &   0.9050  &   0.9250  &   0.9000 & &    0.9400 &    0.9550  &   0.9600\\
$1$   &    0.8750 &    0.8950  &   0.9050 & &    0.9500  &   0.9350  &   0.9650  & &   0.8850  &   0.9050  &   0.9050 & &    0.9300 &    0.9650  &   0.9600\\
\bottomrule
\end{tabular}
\end{threeparttable}
}
\end{table} 

\begin{table}[H]
\caption{Test Points (Coverage -- Slope)}
\label{T:results5}
\centering{}
\resizebox{\textwidth}{!}{
\begin{threeparttable}
The table shows the average coverage for the slope estimation at different test points (MSE). The first column shows the points considered. For the case where $p>1$, all the values of $\bs{Z}$ are equal to the number indicated in the first column. Panel (I) considers the case of Model I ($p=1$). Panel (II) considers the case of Model II ($p=2$).  Panel (III) considers the case of Model III ($p=3$).  Panel (IV) considers the case of Model IV ($p=5$).
\begin{tabular}{@{}lccccccccccccccc@{}}
\toprule
& \multicolumn{7}{c}{\underline{Panel I: $p=1$}} && \multicolumn{7}{c}{\underline{Panel II: $p=2$}}\\
&\multicolumn{3}{c}{\underline{(a): 90\%}} && \multicolumn{3}{c}{\underline{(b): 95\%}} && \multicolumn{3}{c}{\underline{(a): 90\%}} && \multicolumn{3}{c}{\underline{(b): 95\%}}\\
& \multicolumn{3}{c}{Sample Size} && \multicolumn{3}{c}{Sample Size} && \multicolumn{3}{c}{Sample Size} && \multicolumn{3}{c}{Sample Size}\\
          \cmidrule(lr){2-4}
          \cmidrule(lr){6-8}
          \cmidrule(lr){10-12}
          \cmidrule(lr){14-16}
Point     & 250 & 500 & 1000 && 250 & 500 & 1000 &&  250 & 500 & 1000 && 250 & 500 & 1000\\
\midrule
$0$   &    0.9100  &   0.9150  &   0.8650  & &   0.9500  &   0.9450  &   0.9450 & &    0.6350 &    0.7250  &   0.6600 & &    0.7200 &    0.8150 &    0.7800\\
$0.1$ &    0.9000  &   0.9000  &   0.8950  & &   0.9450  &   0.9400  &   0.9250 & &    0.6350 &    0.8200  &   0.8500 & &    0.7350 &    0.8750 &    0.9100\\
$0.2$ &    0.9250  &   0.8950  &   0.9050  & &   0.9550  &   0.9450  &   0.9600 & &    0.5650 &    0.7500  &   0.7700 & &    0.6800 &    0.8350 &    0.8450\\
$0.3$ &    0.8450  &   0.8700  &   0.8800  & &   0.9150  &   0.9250  &   0.9450 & &    0.7700 &    0.7050  &   0.6150 & &    0.8350 &    0.8050 &    0.7400\\
$0.4$ &    0.8900  &   0.9100  &   0.9000  & &   0.9500  &   0.9550  &   0.9650 & &    0.2350 &    0.3900  &   0.4100 & &    0.3450 &    0.5550 &    0.5300\\
$0.5$ &    0.9200  &   0.8900  &   0.8850  & &   0.9650  &   0.9400  &   0.9400 & &    0.3200 &    0.5300  &   0.6800 & &    0.4100 &    0.6700 &    0.7700\\
$0.6$ &    0.8900  &   0.8750  &   0.9150  & &   0.9550  &   0.9500  &   0.9800 & &    0.2850 &    0.4550  &   0.3900 & &    0.4050 &    0.5750 &    0.5150\\
$0.7$ &    0.8650  &   0.8900  &   0.8800  & &   0.9400  &   0.9650  &   0.9600 & &    0.7450 &    0.8550  &   0.8950 & &    0.8450 &    0.9250 &    0.9500\\
$0.8$ &    0.8950  &   0.9000  &   0.8900  & &   0.9300  &   0.9500  &   0.9650 & &    0.5150 &    0.6950  &   0.7100 & &    0.5950 &    0.7750 &    0.7950\\
$0.9$ &    0.8950  &   0.9050  &   0.9100  & &   0.9400  &   0.9500  &   0.9500 & &    0.5950 &    0.7350  &   0.8300 & &    0.7200 &    0.8450 &    0.9100\\
$1$   &    0.8950  &   0.9050  &   0.8850  & &   0.9550  &   0.9550  &   0.9450 & &    0.5400 &    0.6650  &   0.6500 & &    0.6450 &    0.7550 &    0.7600\\
\midrule
& \multicolumn{7}{c}{\underline{Panel III: $p=3$}} && \multicolumn{7}{c}{\underline{Panel IV: $p=5$}}\\
&\multicolumn{3}{c}{\underline{(a): 90\%}} && \multicolumn{3}{c}{\underline{(b): 95\%}} && \multicolumn{3}{c}{\underline{(a): 90\%}} && \multicolumn{3}{c}{\underline{(b): 95\%}}\\
& \multicolumn{3}{c}{Sample Size} && \multicolumn{3}{c}{Sample Size} && \multicolumn{3}{c}{Sample Size} && \multicolumn{3}{c}{Sample Size}\\
          \cmidrule(lr){2-4}
          \cmidrule(lr){6-8}
          \cmidrule(lr){10-12}
          \cmidrule(lr){14-16}
Point     & 250 & 500 & 1000 && 250 & 500 & 1000 &&  250 & 500 & 1000 && 250 & 500 & 1000\\
\midrule
$0$   &    0.7800  &   0.7950  &   0.8100  & &   0.8500  &   0.9200 &    0.8750  & &   0.4450 &    0.4400 &    0.2400  & &   0.5900  &   0.5750  &   0.3400\\
$0.1$ &    0.8050  &   0.8550  &   0.8900  & &   0.8850  &   0.9250 &    0.9400  & &   0.4850 &    0.5700 &    0.3700  & &   0.6300  &   0.6850  &   0.4750\\
$0.2$ &    0.8250  &   0.8850  &   0.8800  & &   0.9150  &   0.9250 &    0.9550  & &   0.5200 &    0.5700 &    0.3600  & &   0.6200  &   0.7000  &   0.5500\\
$0.3$ &    0.8500  &   0.9200  &   0.8800  & &   0.9200  &   0.9500 &    0.9400  & &   0.4000 &    0.4800 &    0.3450  & &   0.5500  &   0.6100  &   0.4700\\
$0.4$ &    0.7700  &   0.8250  &   0.8500  & &   0.8650  &   0.8900 &    0.9000  & &   0.3850 &    0.3500 &    0.1750  & &   0.5050  &   0.4700  &   0.2500\\
$0.5$ &    0.8900  &   0.8900  &   0.8850  & &   0.9400  &   0.9450 &    0.9400  & &   0.9000 &    0.8900 &    0.9000  & &   0.9550  &   0.9450  &   0.9600\\
$0.6$ &    0.1000  &   0.0900  &   0.1500  & &   0.1650  &   0.1650 &    0.2650  & &   0.0850 &    0.0550 &    0.0550  & &   0.1350  &   0.1450  &   0.1050\\
$0.7$ &    0.2450  &   0.2900  &   0.5900  & &   0.3400  &   0.4250 &    0.6950  & &   0.1250 &    0.1500 &    0.1750  & &   0.2050  &   0.2150  &   0.2750\\
$0.8$ &    0.2350  &   0.4950  &   0.6100  & &   0.3600  &   0.5950 &    0.7450  & &   0.2000 &    0.7600 &    0.2150  & &   0.2900  &   0.8900  &   0.3000\\
$0.9$ &    0.3200  &   0.4150  &   0.6250  & &   0.4550  &   0.5350 &    0.7400  & &   0.2400 &    0.2150 &    0.1150  & &   0.3550  &   0.2950  &   0.2000\\
$1$   &    0.2050  &   0.3000  &   0.2500  & &   0.3000  &   0.3950 &    0.3900  & &   0.2250 &    0.1150 &    0.0600  & &   0.3000  &   0.2100  &   0.0900\\
\bottomrule
\end{tabular}
\end{threeparttable}
}
\end{table} 
\begin{figure}[t!]
\centering
\begin{subfigure}[t]{0.35\textwidth}
\includegraphics[scale=0.4]{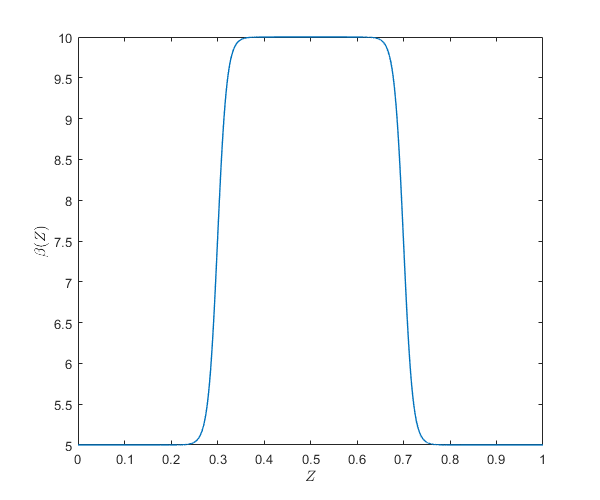}
\caption{Model I}
\end{subfigure}
~
\begin{subfigure}[t]{0.35\textwidth}
\centering
\includegraphics[scale=0.45]{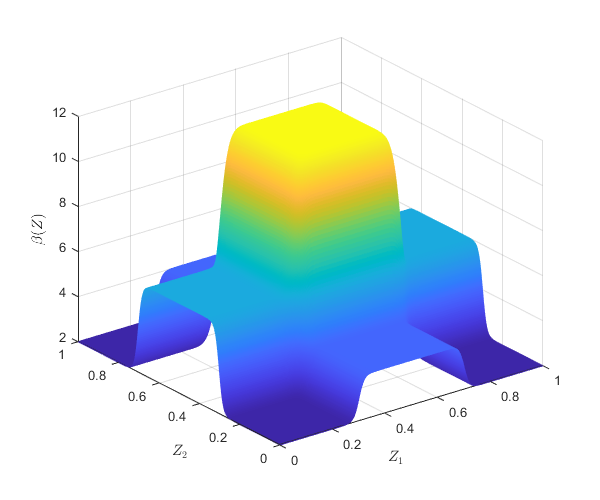}
\caption{Model II}
\end{subfigure}
\\
\begin{subfigure}[t]{0.35\textwidth}
\centering
\includegraphics[scale=0.45]{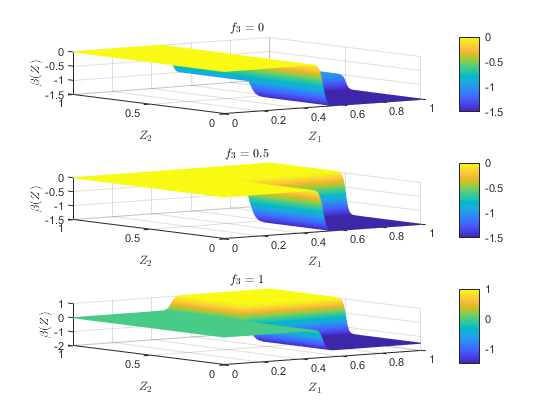}
\caption{Model III}
\end{subfigure}
\caption{Heterogeneity pattern in $\beta_1(\bs{Z}_i)$ for the simulated models.}
\label{fig:Simul1}
\end{figure}

\begin{figure}[t!]
    \centering
    \begin{subfigure}[t]{0.5\textwidth}
        \centering
        \includegraphics[scale=0.4]{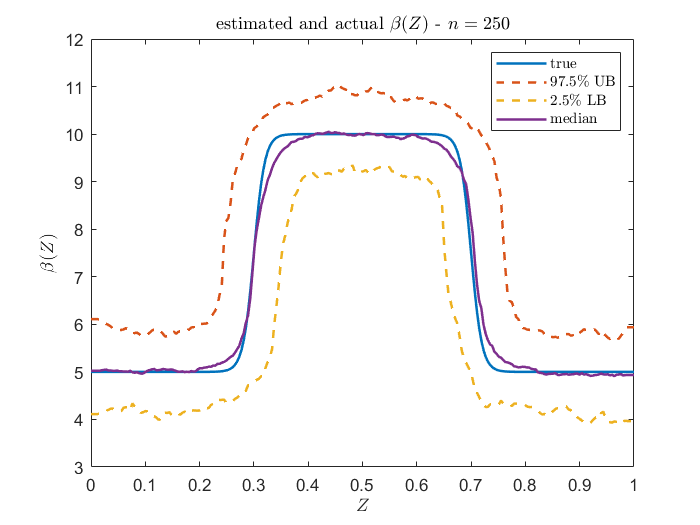}
        \caption{$n=250$}
    \end{subfigure}%
    ~ 
    \begin{subfigure}[t]{0.5\textwidth}
        \centering
        \includegraphics[scale=0.4]{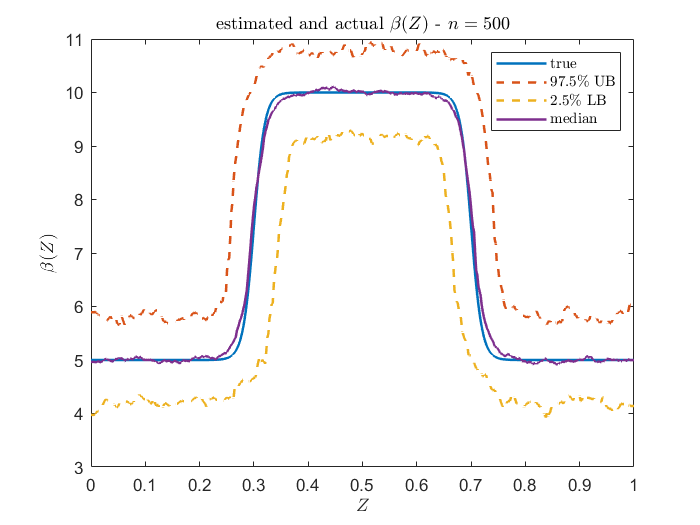}
        \caption{$n=500$}
    \end{subfigure}
\\
\begin{subfigure}[t]{0.5\textwidth}
        \centering
        \includegraphics[scale=0.4]{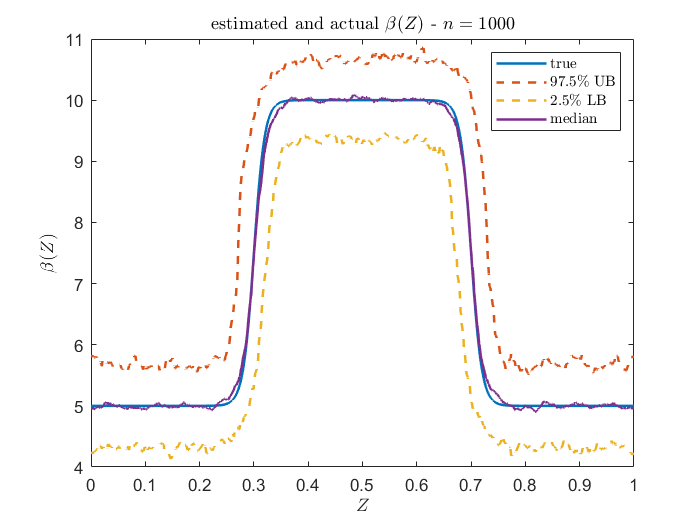}
        \caption{$n=1000$}
    \end{subfigure}%
    \caption{Actual and fitted $\beta(Z)$}
    \label{F:simul2}
\end{figure}

\begin{figure}[t!]
    \centering
    \begin{subfigure}[t]{0.5\textwidth}
        \centering
        \includegraphics[scale=0.4]{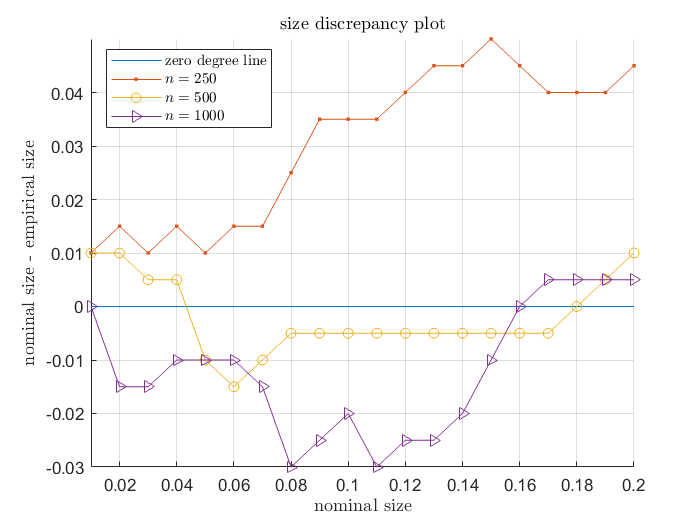}
        \caption{$p=1$}
    \end{subfigure}%
    ~ 
    \begin{subfigure}[t]{0.5\textwidth}
        \centering
        \includegraphics[scale=0.4]{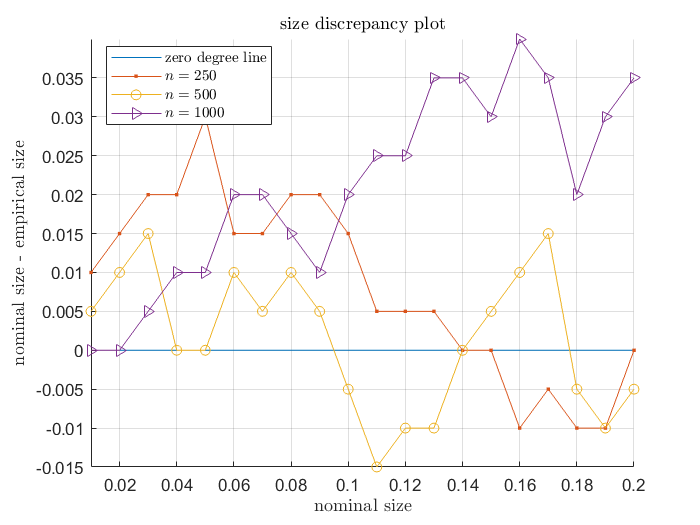}
        \caption{$p=2$}
    \end{subfigure}
\\
\begin{subfigure}[t]{0.5\textwidth}
        \centering
        \includegraphics[scale=0.4]{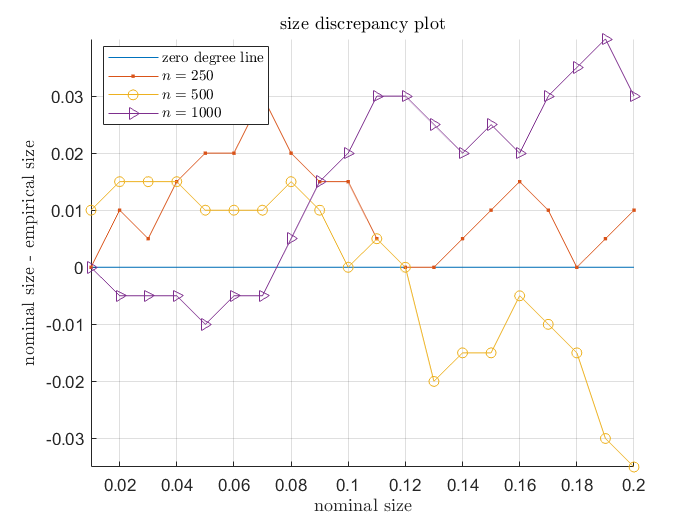}
        \caption{$p=3$}
    \end{subfigure}%
    ~ 
    \begin{subfigure}[t]{0.5\textwidth}
        \centering
        \includegraphics[scale=0.4]{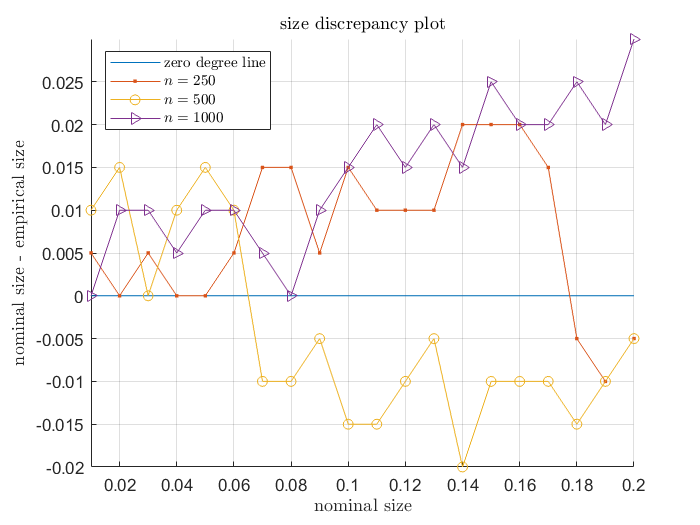}
        \caption{$p=5$}
    \end{subfigure}
    \caption{Size discrepancy plots}
    \label{F:size}
\end{figure}

\begin{figure}[t!]
    \centering
    \begin{subfigure}[t]{0.5\textwidth}
        \centering
        \includegraphics[scale=0.4]{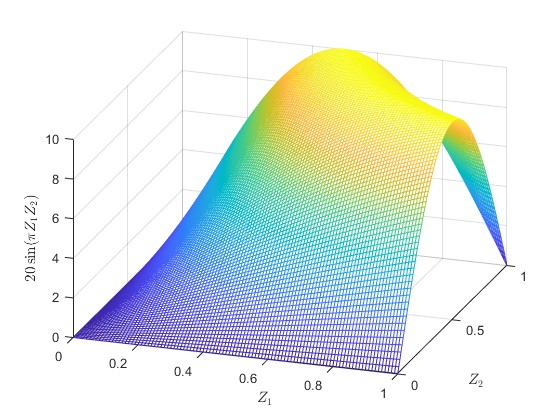}
        \caption{Interaction}
    \end{subfigure}%
    ~ 
    \begin{subfigure}[t]{0.5\textwidth}
        \centering
        \includegraphics[scale=0.4]{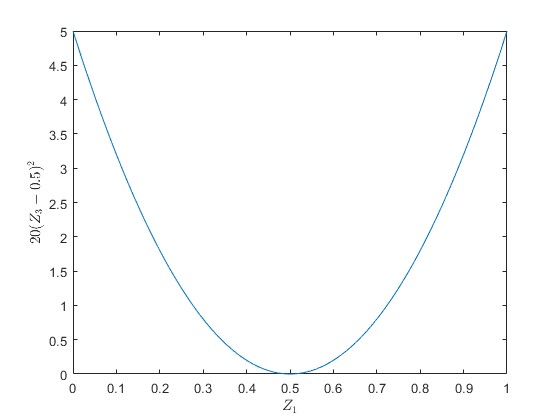}
        \caption{Quadratic}
    \end{subfigure}
\\
\begin{subfigure}[t]{0.5\textwidth}
        \centering
        \includegraphics[scale=0.4]{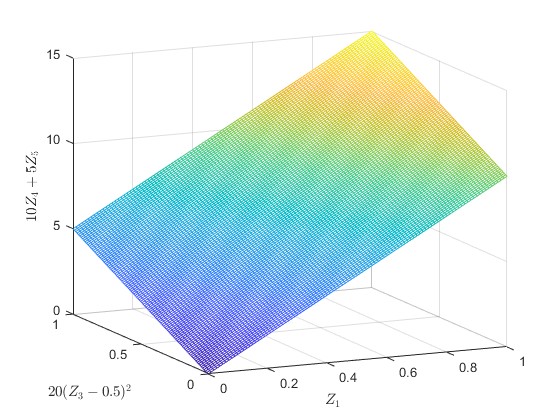}
        \caption{Linear}
    \end{subfigure}%
    \caption{Function componets}
    \label{F:emp1}
\end{figure}

\begin{figure}
    \centering
    \includegraphics[width=0.5\linewidth]{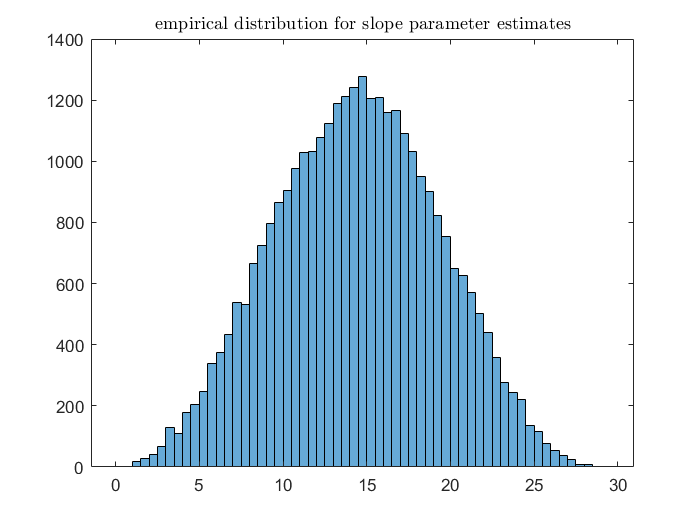}
    \caption{$\gamma$ heterogeneity}
    \label{F:empdist}
\end{figure}

\begin{figure}[t!]
    \centering
    \begin{subfigure}[t]{0.5\textwidth}
        \centering
        \includegraphics[scale=0.4]{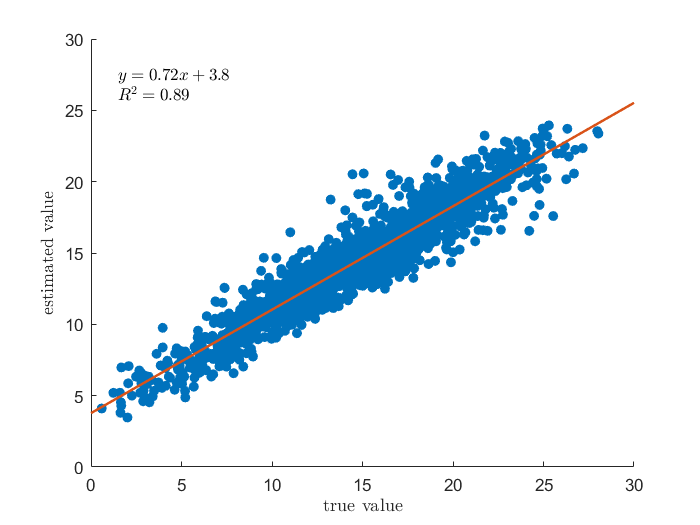}
        \caption{$n=2000$}
    \end{subfigure}%
    ~ 
    \begin{subfigure}[t]{0.5\textwidth}
        \centering
        \includegraphics[scale=0.4]{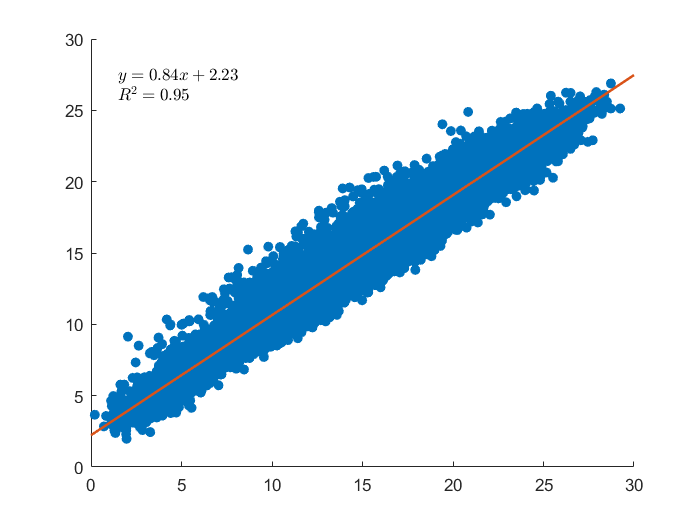}
        \caption{$n=32000$}
    \end{subfigure}
    \caption{Scatter plot of the true slope parameter versus the estimated one}
    \label{F:scatter}
\end{figure}

\begin{figure}
    \centering
    \includegraphics[width=0.5\linewidth]{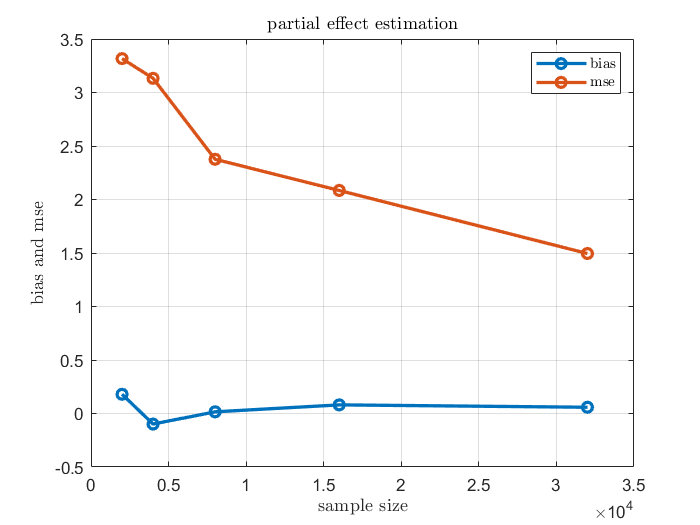}
    \caption{Bias and MSE for the estimation of the slope parameter}
    \label{F:biasmse}
\end{figure}

\begin{figure}
    \centering
    \includegraphics[width=\linewidth]{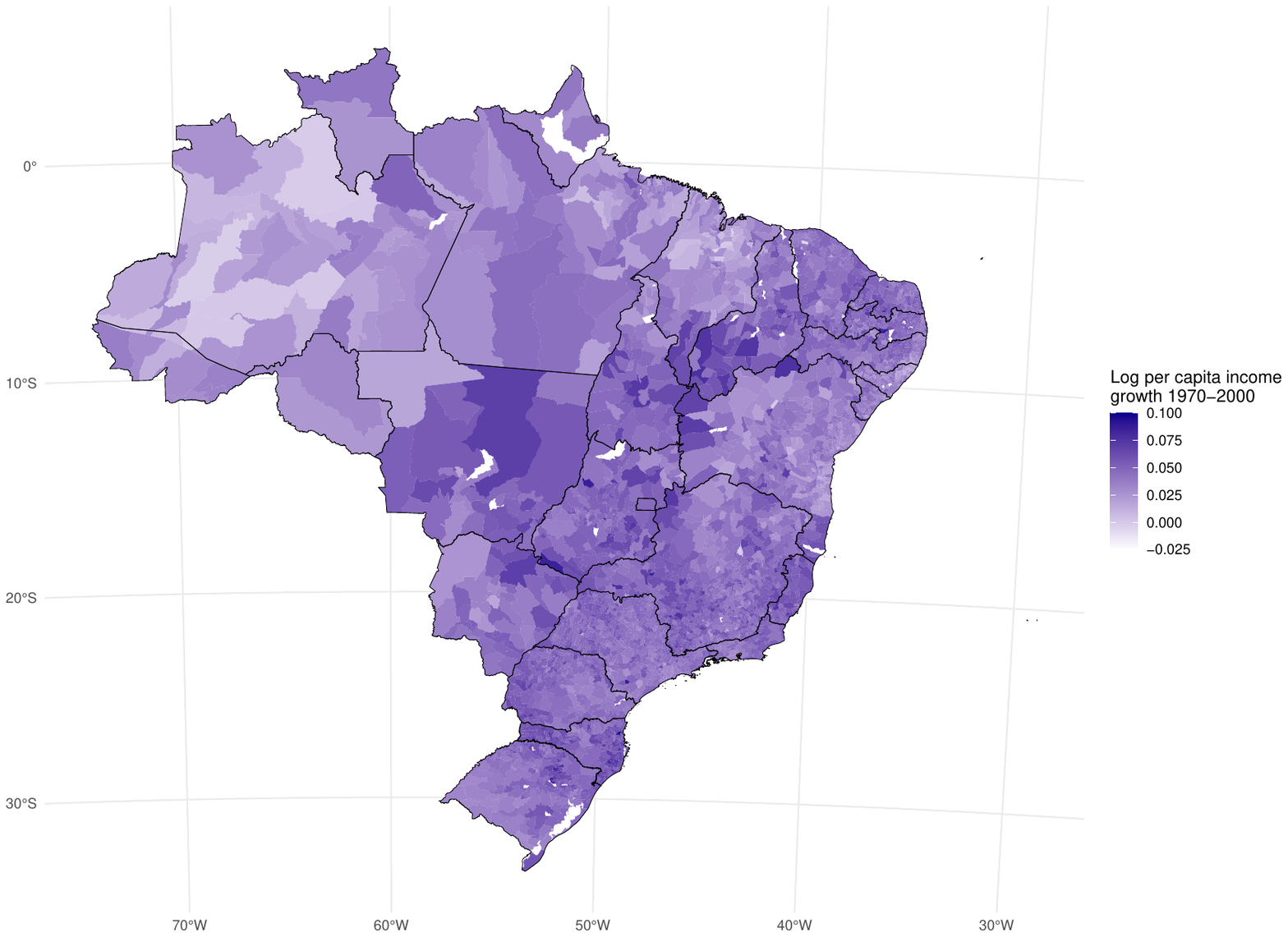}
    \caption{Growth heterogeneity}
    \label{fig:geo1}
\end{figure}

\begin{figure}
    \centering
    \includegraphics[width=\linewidth]{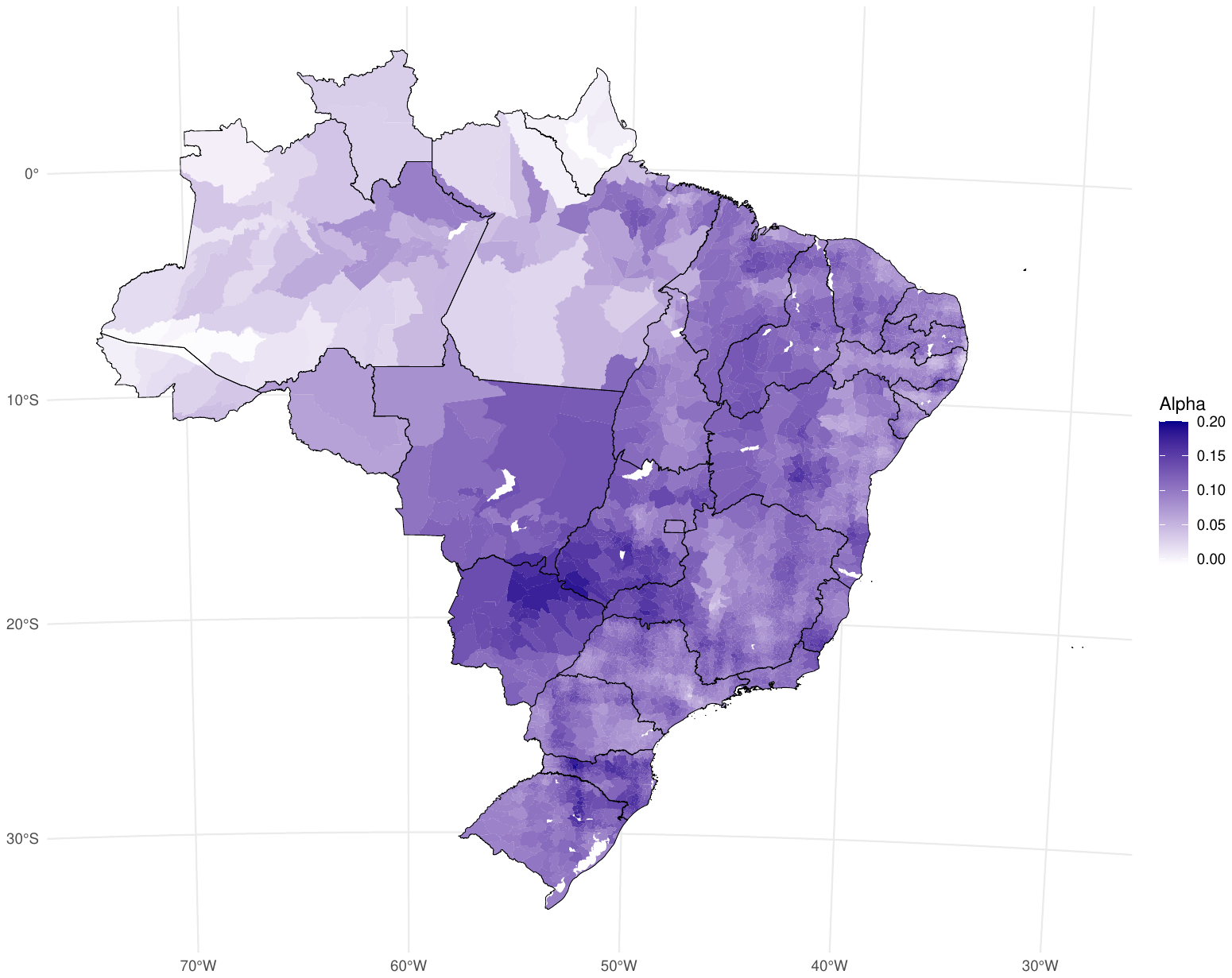}
    \caption{$\alpha$ heterogeneity}
    \label{fig:geo2}
\end{figure}

\begin{figure}
    \centering
    \includegraphics[width=\linewidth]{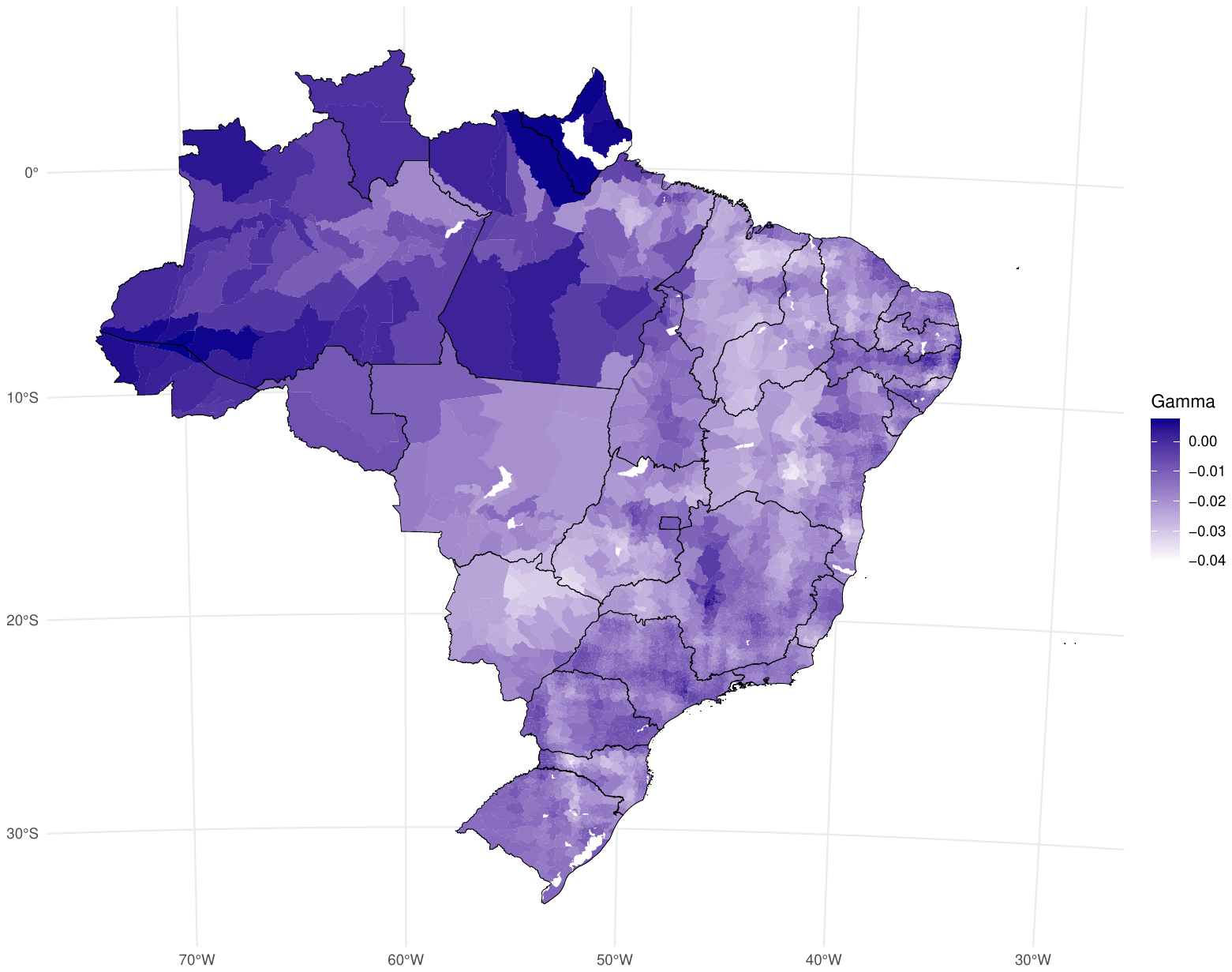}
    \caption{$\gamma$ heterogeneity}
    \label{fig:geo3}
\end{figure}

\newpage

\bibliographystyle{apalike} 
\bibliography{references} 

\begin{thebibliography}{}

\bibitem[Areosa et~al., 2011]{wAmMmM2011}
Areosa, W., McAleer, M., and Medeiros, M. (2011).
\newblock Moment-based estimation of smooth transition regression models with endogenous variables.
\newblock {\em Journal of Econometrics}, 165:100--111.

\bibitem[Athey and Imbens, 2016]{sAgI2016}
Athey, S. and Imbens, G. (2016).
\newblock Recursive partitioning for heterogeneous causal effects.
\newblock {\em Proceedings of the National Academy of Sciences}, 113:7353--7360.

\bibitem[Athey and Imbens, 2019]{sAgI2019}
Athey, S. and Imbens, G. (2019).
\newblock Machine learning methods that economists should know about.
\newblock {\em Annual Review of Economics}, 11:685--725.

\bibitem[Athey et~al., 2019]{sAjTsW2019}
Athey, S., Tibishirani, J., and Wager, S. (2019).
\newblock Generalized random forests.
\newblock {\em Annals of Statistics}, 47:1148--1178.

\bibitem[Barro and {Sala-i-Martin}, 1992]{rjBxS1992}
Barro, R. and {Sala-i-Martin}, X. (1992).
\newblock Convergence.
\newblock {\em Journal of Political Economy}, 100:223--251.

\bibitem[Bonhomme and Manresa, 2015]{sBeM2015}
Bonhomme, S. and Manresa, E. (2015).
\newblock Grouped patterns of heterogeneity in panel data.
\newblock {\em Econometrica}, 83:1147--1184.

\bibitem[Breiman, 2001]{breiman2001random}
Breiman, L. (2001).
\newblock Random forests.
\newblock {\em Machine Learning}, 45:5--32.

\bibitem[Cai et~al., 2000]{zCjFqY2000}
Cai, Z., Fan, J., and Yao, Q. (2000).
\newblock Functional-coefficient regression models for nonlinear time series.
\newblock {\em Journal of the American Statistical Association}, 95:941--956.

\bibitem[Cattaneo et~al., 2024]{mdCrkCmhFyF2024}
Cattaneo, M., Crump, R., Farrell, M., and Feng, Y. (2024).
\newblock On binscatter.
\newblock {\em American Economic Review}, 114:1488--1514.

\bibitem[Chan and Tong, 1986]{ksChT1986a}
Chan, K.~S. and Tong, H. (1986).
\newblock On estimating thresholds in autoregressive models.
\newblock {\em Journal of Time Series Analysis}, 7:179--190.

\bibitem[Chen and Tsay, 1993]{rCrsT1993b}
Chen, R. and Tsay, R.~S. (1993).
\newblock Functional coefficient autoregressive models.
\newblock {\em Journal of the American Statistical Association}, 88:298--308.

\bibitem[Chernozhukov et~al., 2018]{vCiFyL2018}
Chernozhukov, V., Fernández-Val, I., and Luo, Y. (2018).
\newblock The sorted effects method: Discovering heterogeneous effects beyond their averages.
\newblock {\em Econometrica}, 86:1911--1938.

\bibitem[Cribari-Neto et~al., 2000]{brazil2000}
Cribari-Neto, F., Garcia, N.~L., and Vasconcellos, K. L.~P. (2000).
\newblock A note on inverse moments of binomial variates.
\newblock {\em Brazilian Review of Econometrics}, 20(2).

\bibitem[Dagenais, 1969]{mgD1969}
Dagenais, M.~G. (1969).
\newblock A threshold regression model.
\newblock {\em Econometrica}, 37:193--203.

\bibitem[Fan et~al., 2001]{jFcZjZ2001}
Fan, J., Zhang, C., and Zhang, J. (2001).
\newblock Generalized likelihood ratio statistics and wilks phenomenon.
\newblock {\em Annals of Statistics}, 29:153--193.

\bibitem[Fan and Zhang, 1999]{jFwZ1999}
Fan, J. and Zhang, W. (1999).
\newblock Statistical estimation in varying coefficient models.
\newblock {\em Annals of Statistics}, 27:1491--1518.

\bibitem[Feller, 1957]{feller1957}
Feller, W. (1957).
\newblock {\em An Introduction to Probability Theory and Its Applications}.
\newblock Number v. 1-2 in An Introduction to Probability Theory and Its Applications. Wiley.

\bibitem[Friedberg et~al., 2021]{rFjTsAsW2021}
Friedberg, R., Tibshirani, J., Athey, S., and Wager, S. (2021).
\newblock Local linear forests.
\newblock {\em Journal of Computational and Graphical Statistics}, 30:503--517.

\bibitem[Goldfeld and Quandt, 1972]{smGreQ1972}
Goldfeld, S.~M. and Quandt, R. (1972).
\newblock {\em Nonlinear Methods in Econometrics}.
\newblock North Holland, Amsterdam.

\bibitem[Hansen, 2000]{beH2000}
Hansen, B. (2000).
\newblock Sample splitting and threshold estimation.
\newblock {\em Econometrica}, 575:575--603.

\bibitem[Hastie and Tibishirani, 1993]{tHrT1993}
Hastie, T. and Tibishirani, R. (1993).
\newblock Varying-coefficient models.
\newblock {\em Journal of the Royal Statistical Society. Series B (Methodological)}, 55:757--796.

\bibitem[Kiefer, 1978]{nmK1978}
Kiefer, N.~M. (1978).
\newblock Discrete parameter variation: Efficient estimation of a switching regression model.
\newblock {\em Econometrica}, 46:427--434.

\bibitem[Masini et~al., 2023]{rMeMmM2023}
Masini, R., Mendes, E., and Medeiros, M. (2023).
\newblock Machine learning advances for time series forecasting.
\newblock {\em Journal of Economic Surveys}, 37:76--111.

\bibitem[Medeiros and Veiga, 2000]{mcMaV2000}
Medeiros, M.~C. and Veiga, A. (2000).
\newblock A hybrid linear-neural model for time series forecasting.
\newblock {\em {IEEE} Transactions on Neural Networks}, 11:1402--1412.

\bibitem[Medeiros and Veiga, 2005]{mcMaV2005}
Medeiros, M.~C. and Veiga, A. (2005).
\newblock A flexible coefficient smooth transition time series model.
\newblock {\em {IEEE} Transactions on Neural Networks}, 16:97--113.

\bibitem[Peng et~al., 2022]{Peng2022}
Peng, W., Coleman, T., and Mentch, L. (2022).
\newblock {Rates of convergence for random forests via generalized U-statistics}.
\newblock {\em Electronic Journal of Statistics}, 16(1):232 -- 292.

\bibitem[Quandt, 1972]{rQ1972}
Quandt, R. (1972).
\newblock A new approach to estimating switching regression.
\newblock {\em Journal of the American Statistical Association}, 67:306--310.

\bibitem[Suarez-Fariñas et~al., 2004]{sMcPmcM2004}
Suarez-Fariñas, Pedreira, C., and Medeiros, M.~C. (2004).
\newblock Local-global neural networks: A new approach for nonlinear time series modelling.
\newblock {\em Journal of the American Statistical Association}, 99:1092--1107.

\bibitem[Teräsvirta, 1994]{tT1994a}
Teräsvirta, T. (1994).
\newblock Specification, estimation, and evaluation of smooth transition autoregressive models.
\newblock {\em Journal of the American Statistical Association}, 89:208--218.

\bibitem[Tong and Lim, 1980]{hTksL1980}
Tong, H. and Lim, K. (1980).
\newblock Threshold autoregression, limit cycles and cyclical data (with discussion).
\newblock {\em Journal of the Royal Statistical Society, Series B}, 42:245--292.

\bibitem[Tsay, 1989]{rsT1989}
Tsay, R. (1989).
\newblock Testing and modeling threshold autoregressive processes.
\newblock {\em Journal of the American Statistical Association}, 84:431--452.

\bibitem[Vaart, 1998]{Vaart_1998}
Vaart, A. W. v.~d. (1998).
\newblock {\em Asymptotic Statistics}.
\newblock Cambridge Series in Statistical and Probabilistic Mathematics. Cambridge University Press.

\bibitem[Wager and Athey, 2018]{athey_wager2018}
Wager, S. and Athey, S. (2018).
\newblock Estimation and inference of heterogeneous treatment effects using random forests.
\newblock {\em Journal of the American Statistical Association}, 113:1228--1242.

\bibitem[Wager and Walther, 2016]{WW2016}
Wager, S. and Walther, G. (2016).
\newblock Adaptive concentration of regression trees, with application to random forests.

\end{thebibliography}

\newpage

\appendix
\renewcommand{\theequation}{S.\arabic{equation}}

\section{Proofs}


\subsection{Proof of Theorem \ref{thm:unbiasedness}}

Decompose
\[
\widehat{\bs\beta}(\bs{z}) - \widetilde{\bs\beta}(\bs{z}) =\widehat{\bs{\Omega}}(\bs{z})^ {-1}\widehat{\bs{\gamma}}(\bs{z}) - \widetilde{\bs\beta}(\bs{z})= \widehat{\bs{\Omega}}(\bs{z})^ {-1}\left[\sum_{i\in R(\bs{z})}\bs{X}_i\bs{X}_i^{\T}\big(\bs\beta(\bs{Z}_i) - \widetilde{\bs\beta}(\bs{z})\big) + \sum_{i\in R(\bs{z})} \bs{X}_i\epsilon_i\right].
\]
Denote by $\mathcal{G}$ the sigma-algebra generated by $X_1,\dots,X_n$ and $1\{Z_1\in R(\bs{z},\omega)\},\dots,1\{Z_n\in R(\bs{z},\omega)\}$. Then for $i\in[n]$
\[
\E\left[\frac{1}{|\mathcal{A}(\bs{z})|}\sum_{i\in R(\bs{z})} \bs{X}_i\epsilon_i|\mathcal{G}\right] = \frac{1}{|A(\bs{z})|}\sum_{i\in R(\bs{z})}\bs{X}_i\E[\epsilon_i|\bs{X}_i,1\{\bs{Z}_i\in R(\bs{z},\omega)] =0,
\]
where the last equality holds because $R(\bs{z},\omega)$ is independent of $\epsilon_i$ due to the sample split. Furthermore, the second term is also zero since $\E[\bs\beta(\bs{Z}_i) - \widetilde{\bs\beta}(\bs{z})|\mathcal{G}] = \E[\bs\beta(\bs{Z}) - \widetilde{\bs\beta}(\bs{z})|\bs{Z}\in R(\bs{z},\omega) ]=0$ by the definition of $\widetilde{\bs\beta}$. Hence $\E[\widehat{\bs\beta}(\bs{z})|\mathcal{G}] = \widetilde{\bs\beta}(\bs{z})$ and, therefore $\E[\widehat{\bs\beta}(\bs{z})] = \widetilde{\bs\beta}(\bs{z})$ for all $\bs{z}\in[0,1]^d$.

\subsection{Proof of Theorem \ref{thm:beta_converge_rate}}

Recall that based on a subsample $\mathcal{S}\subseteq[n]$, we have the partition $\mathcal{A}\cup \mathcal{B} =\mathcal{S}$ and we define $\mathcal{A}(\bs{z},\omega)\subseteq \mathcal{A}$ as the set of indices $i\in\mathcal{A}$ such that $\bs{Z}_i\in R(\bs{z},\omega)$. By Assumption \ref{ass:main}(b) we have $Y_i = \bs{X}_i^{\T}\bs\beta(\bs{Z}_i) + \epsilon_i$ where $\E[\epsilon_i|\bs{X}_i,\bs{Z}_i]=0$, and write
\begin{align*}
    \widehat{\bs\beta}(\bs{z},w) -\bs\beta(\bs{z}) &= \left[\sum_{i\in\cA(\bs{z},\omega)}\bs{X}_i\bs{X}_i^{\T}\right]^{-1}\left\{\sum_{i\in\cA(\bs{z},\omega)}\bs{X}_i\bs{X}_i^{\T}\bs\beta(\bs{Z}_i) + \sum_{i\in\cA(\bs{z},\omega)}\bs{X}_i\epsilon_i\right\} -\bs\beta(\bs{z})\\
    &=\widehat{\bs{\Omega}}(\bs{z})^{-1}\left\{\frac{1}{|\cA(\bs{z},\omega)|}\sum_{i\in\cA(\bs{z},\omega)}\bs{X}_i\bs{X}_i^{\T}\big[\bs\beta(\bs{Z}_i)-\bs\beta(\bs{z})\big] + \frac{1}{|\cA(\bs{z},\omega)|}\sum_{i\in\cA(\bs{z},\omega)}\bs{X}_i\epsilon_i\right\}.
\end{align*}
Let $\Delta(\bs{z},\omega):=\widehat{\bs{\Omega}}(\bs{z},\omega) - \bs{\Omega}(\bs{z})$, then
\begin{align*}
    \widehat{\bs{\Omega}}(\bs{z},\omega)^{-1}
    &= \left[\bs{\Omega}(\bs{z}) + \Delta(\bs{z})\right]^{-1} = \bs{\Omega}(\bs{z})^ {-1} - \bs{\Omega}(\bs{z})^ {-1} \Delta(\bs{z},\omega)\widehat{\bs{\Omega}}(\bs{z},\omega),
\end{align*}
Therefore, we have the following decomposition
\begin{align}\label{eq:tree_decomposition}
    \widehat{\bs\beta}(\bs{z},w) -\bs\beta(\bs{z}) &=\bs{\Omega}(\bs{z})^{-1}\Big\{I_{d_X} - \Delta(\bs{z},\omega) \big[\Delta(\bs{z},\omega)+\bs{\Omega}(\bs{z})\big]\Big\}\\
    &\quad \left\{\frac{1}{|\cA(\bs{z},\omega)|}\sum_{i\in\cA(\bs{z},\omega)}\bs{X}_i\bs{X}_i^{\T}\big[\bs\beta(\bs{Z}_i)-\bs\beta(\bs{z})\big] + \frac{1}{|\cA(\bs{z},\omega)|}\sum_{i\in\cA(\bs{z},\omega)}\bs{X}_i\epsilon_i\right\}\nonumber.
\end{align}
Apply Lemma \ref{lem:tree_covergence_rate} with $W_i=X_{i,j}X_{i,k}$ for $j,k\in[d_X]$ to conclude that
\begin{equation}\label{eq:XX_convergence_rate}
    \Delta(\bs{z},\omega) \lesssim_\P  k^{-\frac{1}{2}}\left(\frac{s}{k}\right)^{\epsilon/2} + \left(\frac{k}{s}\right)^{(1-\delta)\frac{K(\alpha)\pi}{d_Z}}.
\end{equation}
Similarly, Lemma \ref{lem:tree_covergence_rate} with $W_i=X_{i,j}\epsilon_i$ for $j\in[d_X]$ yields
\[
\frac{1}{|\cA(\bs{z},\omega)|}\sum_{i\in\cA(\bs{z},\omega)}\bs{X}_i\epsilon_i\lesssim_\P k^{-\frac{1}{2}}\left(\frac{s}{k}\right)^{\epsilon/2} + \left(\frac{k}{s}\right)^{(1-\delta)\frac{K(\alpha)\pi}{d_Z}}.
\]

Also, for $j\in[d_X]$, we have
\[
\frac{1}{|\cA(\bs{z},\omega)|}\sum_{i\in\cA(\bs{z},\omega)}\left[\bs{X}_i\bs{X}_i^{\T}\big[\bs\beta(\bs{Z}_i)-\bs\beta(\bs{z})\big]\right]_j = \sum_{\ell=1}^{d_X}\frac{1}{|\cA(\bs{z},\omega)|}\sum_{i\in\cA(\bs{z},\omega)}X_{i,j}X_{i\ell}\big[\beta_\ell(\bs{Z}_i)-\beta_\ell(\bs{z})\big].
\]
By Cauchy-Schwartz inequality, \eqref{eq:generic_tree_unbiasedness} and Markov's inequality we have
\begin{align*}
\frac{1}{|\cA(\bs{z},\omega)|}\sum_{i\in\cA(\bs{z},\omega)}X_{i,j}X_{i\ell}\big[\beta_\ell(\bs{Z}_i)-\beta_\ell(\bs{z})\big]&\leq \left(\frac{1}{|\cA(\bs{z},\omega)|}\sum_{i\in\cA(\bs{z},\omega)}(X_{i,j}X_{i\ell})^2\frac{1}{|\cA(\bs{z},\omega)|}\sum_{i\in\cA(\bs{z},\omega)}\big[\beta_\ell(\bs{Z}_i)-\beta_\ell(\bs{z})\big]^2\right)^{1/2}\\
&\leq \left(\frac{1}{|\cA(\bs{z},\omega)|}\sum_{i\in\cA(\bs{z},\omega)}(X_{i,j}X_{i\ell})^2\right)^ {1/2}\max_{u\in R(\bs{z},\omega)}\|\bs\beta(u)-\bs\beta(\bs{z})\|\\
&\lesssim_\P \left[\E(X_j^2X_k^2|\bs{Z}\in R(\bs{z},\omega)\right]^ {1/2}\max_{u\in R(\bs{z},\omega)}\|\bs\beta(u)-\bs\beta(\bs{z})\|\\
&\lesssim \text{diam}(R(\bs{z},\omega)).
\end{align*}
Finally, the union bound followed by Lemma \ref{lem:upper_bound_diameter}(b) yields
\[
\frac{1}{|\cA(\bs{z},\omega)|}\sum_{i\in\cA(\bs{z},\omega)}\bs{X}_i\bs{X}_i^{\T}\big[\bs\beta(\bs{Z}_i)-\bs\beta(\bs{z})\big]\lesssim_\P\text{diam}(R(\bs{z},\omega)))\lesssim_\P \left(\frac{k}{s}\right)^ {(1-\delta)\frac{K(\alpha)\pi}{d_Z}}.
\]

\subsection{Proof of Theorem \ref{thm:asym_normality_RF}}

Taking the average of all trees using decomposition \eqref{eq:tree_decomposition}, we are left with
\begin{align}\label{eq:forest_decomposition}
    \overline{\bs\beta}(\bs{z},w) -\bs\beta(\bs{z}) &=\bs{\Omega}(\bs{z})^{-1}\frac{1}{B}\sum_{b=1}^B\left\{\Big[\bs{I}_{d_X} - \Delta(\bs{z},\omega) \big(\Delta(\bs{z},\omega)+\bs{\Omega}(\bs{z})\big)\Big]\right.\\
    &\quad \left.\left[\frac{1}{|\cA(\bs{z},\omega)|}\sum_{i\in\cA(\bs{z},\omega)}\bs{X}_i\bs{X}_i^{\T}\big[\bs\beta(\bs{Z}_i)-\bs\beta(\bs{z})\big] + \frac{1}{|\cA(\bs{z},\omega)|}\sum_{i\in\cA(\bs{z},\omega)}\bs{X}_i\epsilon_i\right]\right\}\nonumber.
\end{align}
From the expression above, we note that the term in the form $\frac{1}{B}\sum_{b=1}^B \Delta(\bs{z},\omega)\Delta(\bs{z},\omega)$ will not vanish in probability unless each tree is consistent. In other words, the tree consistency is necessary for the random forest consistency in our setup.

From \eqref{eq:XX_convergence_rate} we have
\begin{align*}
    \overline{\bs\beta}(\bs{z},w) -\bs\beta(\bs{z}) &=\bs{\Omega}(\bs{z})^{-1}\frac{1}{B}\sum_{b=1}^B\left\{\Big[\bs{I}_{d_X} +o_\P(1)\Big]\right.\\
    &\quad \left.\left[\frac{1}{|\cA(\bs{z},\omega)|}\sum_{i\in\cA(\bs{z},\omega)}\bs{X}_i\bs{X}_i^{\T}\big[\bs\beta(\bs{Z}_i)-\bs\beta(\bs{z})\big] + \frac{1}{|\cA(\bs{z},\omega)|}\sum_{i\in\cA(\bs{z},\omega)}\bs{X}_i\epsilon_i\right]\right\}\nonumber.
\end{align*}
Consider the term $\mathcal{T}_1:=\frac{1}{B}\sum_{b=1}^B \frac{1}{|\cA(\bs{z},\omega)|}\sum_{i\in\cA(\bs{z},\omega)}\bs{X}_i\epsilon_i$. By Lemma \ref{lem:tree_covergence_rate}, we have that
\[
\E[\mathcal{T}_1]= \sum_{i=1}^s\E[S_i\bs{X}_i\epsilon_i]=\E[X\epsilon|\bs{Z}\in R(\bs{z},\omega)] = \E\big[X\E[\epsilon|\bs{X},\bs{Z}\in R(\bs{z},\omega)]\big]=0.
\]

Below, we show that, as $n\to\infty$,
\begin{itemize}
    \item[(i)] $\Lambda(\bs{z})^{-1/2}\frac{1}{B}\sum_{b=1}^B
\frac{1}{|\cA(\bs{z},\omega_b)|}\sum_{i\in\cA(\bs{z},\omega_b)}\bs{X}_i\epsilon_i\cd \mathsf{N}(0,I)$ for some covariance matrix $\Lambda(\bs{z})$;
    \item[(ii)] $\Sigma(\bs{z})^{-1/2}\bs{\Omega}(\bs{z})^{-1}\frac{1}{|\cA(\bs{z},\omega)|}\sum_{i\in\cA(\bs{z},\omega)}\bs{X}_i\bs{X}_i^{\T}\big[\bs\beta(\bs{Z}_i)-\bs\beta(\bs{z})\big]\cp 0$ where $\Sigma(\bs{x}):=\bs{\Omega}^{-1}(\bs{z}) \Lambda(\bs{z})\bs{\Omega}^{-1}$, 
\end{itemize}
the result then follows by Slutsky's lemma.

\subsubsection*{Proof of $(i)$}

Recall that, for a sequence of random vectors $(Z_n)$,  $Z_n\cd \mathsf{N}(0,I)$ is equivalent to $a^{\T}Z_n\cd \mathsf{N}(0,1)$ for all $\|a\|=1$. Set $W_i=a^{\T} \bs{X}_i\epsilon_i$ in Lemma \ref{lem:generic_RF_assymptotic_normality} to conclude for all $\|a\|=1$, provided that $k\asymp s^{\eta}$ for $\eta\in (1- K(\alpha,\pi), 1)\subseteq (0,1)$, $s\to\infty$ and $s(\log n)^ {d_Z} = o(n)$, we have
\[
\overline{T}_B(\bs{z}):=\frac{1}{\sqrt{a^{\T}\Lambda(\bs{z}) a}}\frac{1}{B}\sum_{b=1}^ B T_W(\bs{z},\omega_b) \cd \mathsf{N}(0,1),
\]
where let $\mathcal{H}_1(v):=\E[T_W(\bs{z},\omega;Z_1,\dots,Z_s)|Z_1=v]$
\[
    \Lambda(\bs{z}):=\frac{s^2}{n^ 2}\sum_{i=1}^n \V[\mathcal{H}_1(\bs{Z}_i)] =\frac{s^2}{n}\V[X_1\epsilon_1\E[S_1(\bs{z},\omega)|Z_1]],
\]
and $\frac{s^\frac{1-\eta}{K(\alpha)}}{n
(\log s)^{d_Z}}\lesssim \min_{a:\|a\|=1}\lambda^2(\bs{z}) a^{\T}\Lambda(\bs{z}) a$.

\subsubsection*{Proof of $(ii)$}

For $j\in[d_X]$ we wirte
\[
\frac{1}{|\cA(\bs{z},\omega)|}\sum_{i\in\cA(\bs{z},\omega)}[\bs{X}_i\bs{X}_i^{\T}\big[\bs\beta(\bs{Z}_i)-\bs\beta(\bs{z})\big]]_j = \sum_{\ell=1}^{d_X}\frac{1}{|\cA(\bs{z},\omega)|}\sum_{i\in\cA(\bs{z},\omega)}X_{i,j}X_{i\ell}\big[\beta_\ell(\bs{Z}_i)-\beta_\ell(\bs{z})\big]
\]
By Cauchy-Schwartz inequality and Lemma \ref{lem:tree_covergence_rate}
\begin{align*}
\frac{1}{|\cA(\bs{z},\omega)|}\sum_{i\in\cA(\bs{z},\omega)}X_{i,j}X_{i\ell}\big[\beta_\ell(\bs{Z}_i)-\beta_\ell(\bs{z})\big]&\leq \left(\frac{1}{|\cA(\bs{z},\omega)|}\sum_{i\in\cA(\bs{z},\omega)}(X_{i,j}X_{i\ell})^2\frac{1}{|\cA(\bs{z},\omega)|}\sum_{i\in\cA(\bs{z},\omega)}\big[\beta_\ell(\bs{Z}_i)-\beta_\ell(\bs{z})\big]^2\right)^{1/2}\\
&\leq \left(\frac{1}{|\cA(\bs{z},\omega)|}\sum_{i\in\cA(\bs{z},\omega)}(X_{i,j}X_{i\ell})^2\right)^ {1/2}\max_{u\in R(\bs{z},\omega)}\|\bs\beta(u)-\bs\beta(\bs{z})\|\\
&\leq \left[\E(X_j^2X_k^2|\bs{Z}=\bs{z})\right]^ {1/2}\max_{u\in R(\bs{z},\omega)}\|\bs\beta(u)-\bs\beta(\bs{z})\| + o_\P(1).
\end{align*}
Thus
\[
\left\|\frac{1}{|\cA(\bs{z},\omega)|}\sum_{i\in\cA(\bs{z},\omega)}\bs{X}_i\bs{X}_i^{\T}\big[\bs\beta(\bs{Z}_i)-\bs\beta(\bs{z})\big]\right\|\lesssim_\P \text{diam}(R(\bs{z},\omega)).
\]
By the Lipschitz condition on $\bs{z}\mapsto \bs\beta(\bs{z})$ (Assumption \ref{ass:main}(b)), we have
\begin{align*}
    \|\Sigma(\bs{z})^{-1/2}\bs{\Omega}(\bs{z})^{-1}\frac{1}{|\cA(\bs{z},\omega)|}\sum_{i\in\cA(\bs{z},\omega)}\bs{X}_i\bs{X}_i^{\T}\big[\bs\beta(\bs{Z}_i)-\bs\beta(\bs{z})\big]\|&\leq C_{p_X}\frac{\|\bs{\Omega}(\bs{z})^ {-1}\|^2}{\|\Lambda(\bs{z})\|^ {1/2}}\max_{u\in R(\bs{z},\omega)}\|\bs\beta(u)-\bs\beta(\bs{z})\|\\
    &\leq C_{p_X}C_{\bs\beta}\frac{\|\bs{\Omega}(\bs{z})^ {-1}\|^2}{\|\Lambda(\bs{z})\|^ {1/2}}\text{diam}(R(\bs{z},\omega)),
\end{align*}
where $C_{p_X}$ is a constant only depending on $p_X$ and $C_{\bs\beta}$ is the Lipschitz constant.

Since we assume that $\|\bs{\Omega}(\bs{z})^{-1}\|_2\lesssim 1$ by Assumption \ref{ass:main}(d), it suffices for $(ii)$  that $\text{diam}(R(\bs{z},\omega))=o_\P(\|\Lambda(\bs{z})\|^{1/2})$. For that, we have from Lemma \ref{lem:lower_bound} and the proof of part $(i)$,
\[
\frac{\text{diam}(R(\bs{z},\omega))}{\|\Lambda(\bs{z})\|^{1/2}}\lesssim_\P s^{-\frac{ (1-\eta)\pi K(\alpha)}{2d_Z}} \left(\frac{s^{\frac{1-\eta}{K(\alpha)}}}{n(\log s)^{d_Z}}\right)^{-1/2} \lesssim n^{-\frac{1}{2}\big[\bs\beta(1-\eta)\big(\frac{\pi K(\alpha)}{d_Z}+  \frac{1}{K(\alpha)}\big) -1\big]}(\log n)^{{d_Z}/2}.
\]
By assumption $\bs\beta(1-\eta)\big(\frac{\pi K(\alpha)}{d_Z}+ \frac{1}{K(\alpha)}\big)>1$ hence the right-hand side converges to $0$ is probability which proves $(ii)$.

\subsection{Proof of Theorem \ref{thm:LR_Test}}

Recall the random forest decomposition.
\begin{align*}
    \overline{\bs\beta}(\bs{z},w) -\bs\beta(\bs{z}) &=\bs{\Omega}(\bs{z})^{-1}\frac{1}{B}\sum_{b=1}^B\left\{\Big[I_{d_X} - \Delta(\bs{z},\omega) \big(\Delta(\bs{z},\omega)+\bs{\Omega}(\bs{z})\big)\Big]\right.\\
    &\quad \left.\left[\frac{1}{|\cA(\bs{z},\omega)|}\sum_{i\in\cA(\bs{z},\omega)}\bs{X}_i\bs{X}_i^{\T}\big[\bs\beta(\bs{Z}_i)-\bs\beta(\bs{z})\big] + \frac{1}{|\cA(\bs{z},\omega)|}\sum_{i\in\cA(\bs{z},\omega)}\bs{X}_i\epsilon_i\right]\right\}\nonumber.
\end{align*}
Under $\mathcal{H}_0$, the bias term (the first term in the square brackets) vanishes, and we are left with
\begin{align*}
    \overline{\bs\beta}(\bs{z},w) -\beta_0 &=\bs{\Omega}(\bs{z})^{-1}\frac{1}{B}\sum_{b=1}^B\left\{\Big[I_{d_X} - \Delta(\bs{z},\omega) \big(\Delta(\bs{z},\omega)+\bs{\Omega}(\bs{z})\big)\Big] \frac{1}{|\cA(\bs{z},\omega)|}\sum_{i\in\cA(\bs{z},\omega)}\bs{X}_i\epsilon_i\right\}\nonumber\\
    &=\bs{\Omega}(\bs{z})^{-1}\left[\frac{1}{B}\sum_{b=1}^B\frac{1}{|\cA(\bs{z},\omega)|}\sum_{i\in\cA(\bs{z},\omega)}\bs{X}_i\epsilon_i + \zeta_1(\bs{z})\right]\\
    &=\bs{\Omega}(\bs{z})^{-1}\left[\frac{s}{n}\sum_{i=1}^n \mathcal{H}_1(\bs{z},W_i) + \zeta_1(\bs{z}) + \zeta_2(\bs{z})\right],
\end{align*}
where $\zeta(\bs{z}) = \zeta_1(\bs{z}) + \zeta_2(\bs{z})$ with
\begin{align*}
    \zeta_1(\bs{z}) &:= \frac{1}{B}\sum_{b=1}^B\left\{- \Delta(\bs{z},\omega) \big(\Delta(\bs{z},\omega)+\bs{\Omega}(\bs{z})\big) \frac{1}{|\cA(\bs{z},\omega)|}\sum_{i\in\cA(\bs{z},\omega)}\bs{X}_i\epsilon_i\right\}\\
    \zeta_2(\bs{z}) &:= \frac{1}{B}\sum_{b=1}^B\frac{1}{|\cA(\bs{z},\omega_b)|}\sum_{i\in\cA(\bs{z},\omega_b)}\bs{X}_i\epsilon_i - \frac{s}{n}\sum_{i=1}^n \mathcal{H}_1(\bs{z},W_i)\\
    &=\sum_{i=1}^n\big[\omega(\bs{z},\bs{Z}_i)- s\theta(\bs{z},\bs{Z}_i)/n\big] \epsilon_i \bs{X}_i
\end{align*}
where $\omega(\bs{z},\bs{Z}_i):=\frac{1}{B}\sum_{b=1}^B \frac{1\{\bs{Z}_i\in \cA(\bs{z},\omega_b)\}}{|\cA(\bs{z},\omega_b)|}$. Note that $\E[\zeta_2(Z_j)|Z_1,\dots, Z_n,\omega]=0$ for $j\in[n]$ because $\omega(Z_j,\bs{Z}_i)- s\theta(Z_j,\bs{Z}_i)/n$ is $Z_1,\dots, Z_n,\omega$ measurable and $\E[\epsilon_i|\bs{Z}_i,\bs{X}_i]=0$. Also 
\[
\V[\zeta_2(\bs{Z}_i|\bs{Z})= \sum_{i=1}^n \V[\epsilon_i\bs{X}_i]
\]
\begin{align*}
    \sum_{j=1}^n \zeta_2(Z_j) = \sum_{i=1}^n \epsilon_i \bs{X}_i\sum_{j=1}^n\big[\omega(Z_j,\bs{Z}_i)- s\theta(Z_j,\bs{Z}_i)/n\big] = \sum_{i=1}^n \epsilon_i \bs{X}_iq_i(\bs{Z}^{(n)})
\end{align*}
and
\[
\V[\sum_{j=1}^n \zeta_2(Z_j)|\bs{Z}^{(n)}]=\sigma^ 2\sum_{j=1}^n\E[\bs{\Omega}(Z_j)]q_j^2
\]

\begin{align*}
\E[\zeta_2(Z_\ell)\zeta_2(Z_k)^{\T}] &=\sum_{i=1}^n\sum_{j=1}^n\E\left\{\epsilon_i\epsilon_j\big[\omega(Z_\ell,\bs{Z}_i)- s\theta(Z_\ell,\bs{Z}_i)/n\big]\big[\omega(Z_k,\bs{Z}_i)- s\theta(Z_k,\bs{Z}_i)/n\big]  \bs{X}_iX_j^{\T}\right\}\\
&=\sigma^ 2\sum_{i=1}^n\big[\omega(Z_\ell,\bs{Z}_i)- s\theta(Z_\ell,\bs{Z}_i)/n\big]\big[\omega(Z_k,\bs{Z}_i)- s\theta(Z_k,\bs{Z}_i)/n\big]  \bs{X}_iX_j^{\T}\\
\end{align*}
\[
\V[\sum_{i=1}^n \zeta_2(\bs{Z}_i)] \sum_{i=1}\sum_{j=1}^n\big[\omega(\bs{z},\bs{Z}_i)- s\theta(\bs{z},\bs{Z}_i)/n\big] \epsilon_i \bs{X}_i
\]
write 

Recall that 
\[
\V[\zeta_2(\bs{z})] = \V[T] +\V[T_0]
\]

Then
\begin{align*}
    \Delta:=\RSS_0 - \RSS= \underbrace{2\sum_{i=1}^n\epsilon_i\bs{X}_i^{\T}(\overline{\bs\beta}(\bs{Z}_i)-\widetilde{\bs\beta})}_{=:\Delta_1} + \underbrace{\sum_{i=1}^n(\overline{\bs\beta}(\bs{Z}_i)-\widetilde{\bs\beta})^{\T}\bs{X}_i\bs{X}_i^{\T}(\overline{\bs\beta}(\bs{Z}_i)-\widetilde{\bs\beta})}_{=:\Delta_2} 
\end{align*}

For $\bs{Z}_i,Z_j$ and $\omega$ define
\[
S(\bs{Z}_i,Z_j,\omega):=\begin{cases}
    \left|\{k\in \mathcal{A}:Z_k\in R(Z_j,\omega)\}\right|^{-1}& \text{if $\bs{Z}_i\in R(Z_j,\omega)$} \\
    0 &\text{otherwise}.
    \end{cases}
\]
and
\[
\theta(\bs{Z}_i,Z_j) := \E[S(\bs{Z}_i,Z_j,\omega)|\bs{Z}_i,Z_j].
\]
Note that $\theta$ is symmetric in its two arguments because $S$ is symmetric with respect to its first two arguments. Also, for $i=j$ we have that $S(\bs{Z}_i,\bs{Z}_i,\omega)= (1+B)^{-1}$ where $B\sim\text{Binom}(s-1,p_i)$ conditional on $\omega$ and $\bs{Z}_i$ where $p_i\asymp R(\bs{Z}_i,\omega)$. Then from \eqref{eq:leaves_prob_bounds} on a event with probability approach 1
\[
\E[S(\bs{Z}_i,\bs{Z}_i,\omega)|\omega, \bs{Z}_i] \asymp \frac{1}{(s-1)p_i} \asymp \frac{1}{s(k/s)^{1+o(1)}}
\]
therefore
\[
\theta(\bs{Z}_i,\bs{Z}_i) \asymp_\P \frac{1}{s(k/s)^{1+o(1)}}.
\]
For $i\neq j$ we have $S(\bs{Z}_i,Z_j,\omega)\leq S(\bs{Z}_i,\bs{Z}_i,\omega)$ then $\theta(\bs{Z}_i,Z_j)\lesssim_\P \frac{1}{s(k/s)^{1+o(1)}}$.

Also, $\mathcal{H}_1(\bs{Z}_i,Z_j)=\E[\bs{X}_i\epsilon_i S(\bs{Z}_i,Z_j,\omega)|\bs{Z}_i,\bs{X}_i,Y_i] = \epsilon_i \bs{X}_i \theta(\bs{Z}_i,Z_j)$ thus
\begin{align*}
    \overline{\bs\beta}(\bs{Z}_i)-\widetilde{\bs\beta} &=\overline{\bs\beta}(\bs{Z}_i)- \bs\beta(\bs{Z}_i) +\bs\beta(\bs{Z}_i) - \widetilde{\bs\beta} \\
    &=\bs{\Omega}(\bs{Z}_i)^{-1}\frac{s}{n}\sum_{j=1}^n \mathcal{H}_1(\bs{Z}_i,Z_j) +\zeta(\bs{Z}_i) 
 - \left(\sum_{j=1}^n X_jX_j^{\T}\right)^{-1}\sum_{j=1}^n \epsilon_jX_j  - \left(\sum_{j=1}^n X_jX_j^{\T}\right)^{-1}\sum_{j=1}^n (\bs\beta(Z_j) - \bs\beta(\bs{Z}_i))^{\T}X_j\\
    &=  \bs{\Omega}(\bs{z})^{-1}\frac{s}{n}\sum_{j=1}^n \epsilon_jX_j \theta(\bs{z},Z_j)  -\Big(\E[\bs{X}\bs{X}^{\T}]\Big)^{-1} \frac{1}{n}\sum_{j=1}^n \epsilon_jX_j +\zeta(\bs{z}) +O_\P(n^{-1})
\end{align*}
Under $\mathcal{H}_0$, we have $\bs{\Omega}:=\E[\bs{X}\bs{X}^{\T}]=\bs{\Omega}(\bs{z})$ for $\bs{z}\in[0,1]^d$, then
\begin{equation}\label{eq:beta_decomp_LRT}
    \overline{\bs\beta}(\bs{Z}_i)-\widetilde{\bs\beta} =  \bs{\Omega}^{-1}\frac{1}{n}\sum_{j=1}^n \epsilon_jX_j \big[s\theta(\bs{Z}_i,Z_j) - 1\big]  +\zeta(\bs{Z}_i) +O_\P(n^{-1});\quad i\in[d]
\end{equation}
Now for $\Delta_1$ we have
\begin{align*}
    \Delta_1&:=2\sum_{i=1}^n\epsilon_i\bs{X}_i^{\T}(\overline{\bs\beta}(\bs{Z}_i)-\widetilde{\bs\beta})\\
    &=\frac{2}{n}\sum_{i=1}^n\sum_{j=1}^n\epsilon_i\epsilon_j\big[s\theta(\bs{Z}_i,Z_j) - 1\big]\bs{X}_i^{\T}\bs{\Omega}^{-1}X_j   +2\sum_{i=1}^n\epsilon_i\bs{X}_i^{\T} \zeta(\bs{Z}_i) + O_\P(n^ {-1})\sum_{i=1}^n \epsilon_i\bs{X}_i\\
    &=\frac{2}{n}\sum_{i=1}^n\epsilon_i^ 2\big[s\theta(\bs{Z}_i,\bs{Z}_i) - 1\big]\bs{X}_i^{\T}\bs{\Omega}^{-1}\bs{X}_i +\frac{2}{n}\sum_{i\neq j}^n\epsilon_i\epsilon_j\big[s\theta(\bs{Z}_i,Z_j) - 1\big]\bs{X}_i^{\T}\bs{\Omega}^{-1}X_j \\
    &\qquad +2\sum_{i=1}^n\epsilon_i\bs{X}_i^{\T} \zeta(\bs{Z}_i) +O_\P(n^{-1/2})\\
    &=:\Delta_{11} + \Delta_{12} + \mathcal{E}_1
\end{align*}.

Note that $\E[\Delta_{11}] = 2\sigma^2 d_X \big(s\E[\theta(Z_1,Z_1)] - 1\big)$ and $\V[\Delta_{11}]=O(\V[\theta(\bs{Z},\bs{Z})]s^2/n) =O(\E[\theta^2(\bs{Z},\bs{Z})]s^2/n) \asymp [(k/s)^ {2(1+o(1))} n]^{-1}$. Hence
\begin{align*}
    \Delta_{11} =  2\sigma^2 d_X \big(s\E[\theta(Z_1,Z_1)] - 1\big)+ O_\P\left(\frac{(s/k)^ {1+o(1)}}{\sqrt{n}}\right),
\end{align*}
and therefore,
\begin{align*}
    \Delta_1&= 2\sigma^2 d_X \big(s\E[\theta(Z_1,Z_1)] - 1\big)+\frac{2}{n}\sum_{i\neq j}^n\epsilon_i\epsilon_j\big[s\theta(\bs{Z}_i,Z_j) - 1\big]\bs{X}_i^{\T}\bs{\Omega}^{-1}X_j \\
    &\qquad + O_\P\left(\frac{(s/k)^ {1+o(1)}}{\sqrt{n}}\right) +\mathcal{E}_1.
\end{align*}.

For $\Delta_2$, let $T_{ij}:=s\theta(\bs{Z}_i,Z_j) - 1$ for $i,j\in[d]$, use \eqref{eq:beta_decomp_LRT} and collect terms we are left with
\begin{align*}  
\Delta_2&:=\sum_{i=1}^n(\overline{\bs\beta}(\bs{Z}_i)-\widetilde{\bs\beta})^{\T}\bs{X}_i\bs{X}_i^{\T}(\overline{\bs\beta}(\bs{Z}_i)-\widetilde{\bs\beta})\\
 &=\frac{1}{n^2}\sum_{i=1}^n\sum_{j=1}^n\sum_{k=1}^n \epsilon_j\epsilon_k T_{ij}T_{ik}X_j^{\T}\bs{\Omega}^{-1}\bs{X}_i\bs{X}_i^{\T}\bs{\Omega}^{-1} X_k + \mathcal{E}_2. 
\end{align*}
Decompose the first term in the last expression as $\Delta_{21} +  \Delta_{22}$ where
\begin{align*}
    \Delta_{21} &:= \frac{1}{n^2}\sum_{i=1}^n\sum_{j=1}^n \epsilon_j^2 T_{ij}^2X_j^{\T}\bs{\Omega}^{-1}\bs{X}_i\bs{X}_i^{\T}\bs{\Omega}^{-1} X_j \\
    \Delta_{22} &:=\frac{1}{n^2}\sum_{i=1}^n\sum_{j\neq k}^n\epsilon_j\epsilon_k T_{ij}T_{ik}X_j^{\T}\bs{\Omega}^{-1}\bs{X}_i\bs{X}_i^{\T}\bs{\Omega}^{-1} X_k.
\end{align*}
Also
\[
\E[\Delta_{21}] = \frac{\sigma^2}{n^ 2}\sum_{i=1}^n\sum_{j=1}^n\E[T_{ij}^2\mathsf{tr}\big(\bs{\Omega}(Z_j)\bs{\Omega}(\bs{Z}_i)^{-1}\big)].
\]
 Under $\mathcal{H}_0$, we have $\bs{\Omega}(\bs{Z}_i)=\bs{\Omega}$, hence
\[
\E[\Delta_{21}] = \frac{\sigma^2d_X}{n^2}\sum_{i,j}\E[T_{ij}^2] =\frac{\sigma^2d_X}{n^2}(n\E[T_{11}^2] + n(n-1)\E[T_{12}^2]).
\]
Decompose further $\Delta_{22} = \Delta_{221} +\Delta_{222} $ where
\begin{align*}
    \Delta_{221} &:=\frac{1}{n^2}\sum_{i\notin\{j,k\}}^n\sum_{j\neq k}^n\epsilon_j\epsilon_k T_{ij}T_{ik}X_j^{\T}\bs{\Omega}(\bs{Z}_i)^{-1}\bs{X}_i\bs{X}_i^{\T}\bs{\Omega}(\bs{Z}_i)^{-1} X_k\\
    \Delta_{222} &:= \frac{2}{n^2}\sum_{j\neq k}^n\epsilon_j\epsilon_k T_{jj}T_{jk}X_j^{\T}\bs{\Omega}(Z_j)^{-1}\bs{X}_i\bs{X}_i^{\T}\bs{\Omega}(Z_j)^{-1} X_k
\end{align*}
Define
\begin{align*}
H_{jk}&:=\E[T_{ij}T_{ik}\bs{\Omega}(\bs{Z}_i)^{-1}\bs{X}_i\bs{X}_i^{\T}\bs{\Omega}(\bs{Z}_i)^{-1}|Z_j,Z_k]\\
&=\E\Big[T_{ij}T_{ik}\bs{\Omega}(\bs{Z}_i)^{-1}\E\big[\bs{X}_i\bs{X}_i^{\T}|\bs{Z}_i,Z_j,Z_k\big]\bs{\Omega}(\bs{Z}_i)^{-1}|Z_j,Z_k\Big]\\
&=\E\Big[T_{ij}T_{ik}\bs{\Omega}(\bs{Z}_i)^{-1}|Z_j,Z_k\Big]
\end{align*}
where $\bs{Z}'$ is an independent copy of $Z_1,\dots Z_n$.
Then
\begin{align*}
    \Delta_{221} = \frac{2(n-2)}{n^2}\sum_{1\leq j < k\leq n}\epsilon_j\epsilon_k X_j^{\T}H_{jk}X_k + O_\P()
\end{align*}

Furthermore, under the $\mathcal{H}_0$,
\[
H_{jk} = \bs{\Omega}^{-1}\E[T_{ij}T_{ik}|Z_j,Z_k] = \bs{\Omega}^{-1}\tau(Z_j,Z_k),
\]
where $\tau(\bs{z},\bs{z}') := \E[(s\theta(\bs{Z},z) -1)(s\theta(\bs{Z},z') -1)]$.

Putting everything together, under $\mathcal{H}_0$,
\begin{align*}
    \Delta &= 2\sigma^2 d_X \E[T_{11}] + \frac{\sigma^2d_X}{n^2}\sum_{i,j}\E[T_{ij}^2]\\
    &\qquad+\frac{2}{n}\sum_{i\neq j}^n\epsilon_i\epsilon_jT_{ij}\bs{X}_i^{\T}\bs{\Omega}^{-1}X_j +\frac{2(n-2)}{n^2}\sum_{1\leq j < k\leq n}\epsilon_j\epsilon_k X_j^{\T}H_{jk}X_k + O_\P()\\
    &=2\sigma^2 d_X \E[T_{11}] + \frac{\sigma^2d_X}{n^2}\sum_{i,j}\E[T_{ij}^2]\\
    &\qquad + \frac{2}{n}\sum_{i\neq j}^n\epsilon_i\epsilon_j\big[T_{ij} + \tfrac{n-2}{n}\tau(\bs{Z}_i,Z_j)\big]\bs{X}_i^{\T}\bs{\Omega}^{-1}X_j\\
    &\qquad + o_\P(1).
\end{align*}
Let $\eta(\bs{z},\bs{z}'):= s\theta(\bs{z},\bs{z}')-1 + \tfrac{n-2}{n}\tau(\bs{z},\bs{z}')$ for $\bs{z},\bs{z}'\in[0,1]^{d_Z}$ and 
\begin{align*}
\E\big[\epsilon_1^2\epsilon_2^2\eta(Z_1,Z_2)^2(X_1^{\T}\bs{\Omega}^{-1}X_2)^2\big] &= \sigma^4\E\big[\eta(Z_1,Z_2)^2(X_1^{\T}\bs{\Omega}^{-1}X_2)^2\big]\\
&= \sigma^4\E\Big[\eta(Z_1,Z_2)^2 X_1^{\T}\bs{\Omega}^{-1}\E\big[X_2X_2^{\T}|Z_1,Z_2, X_1\big]\bs{\Omega}^ {-1}X_1\Big]\\
&= \sigma^4\E\Big[\eta(Z_1,Z_2)^2 X_1^{\T}\bs{\Omega}^{-1}X_1\Big]\\
&= \sigma^4 d_X\E\Big[\eta(Z_1,Z_2)^2 \Big].
\end{align*}
Then $\nu^{-1}(\Delta - \mu)\cd \mathsf{N}(0,1)$ where
\[
\mu := 2\sigma^2 d_X (s\E[\theta(\bs{Z},\bs{Z})]-1)+\frac{\sigma^2d_X}{n^2} (n\E[T_{11}^2] + n(n-1)/2\E[T_{12}^2])
\]
and
\[
\nu^2 := \frac{4n(n-1)}{n^2}\sigma^4 d_X\E\Big[\eta(Z_1,Z_2)^2 \Big].
\]
Finally, note that $\tau(\bs{Z},\bs{Z}')\lesssim (s/k)^{2(1+o(1))}$.

\section{Auxiliary Lemmas and Proofs}

\begin{lemma}[Upper bounds on the leaf diameter]\label{lem:upper_bound_diameter} 
For $\bs{z}\in[0,1]^d$, let $R(\bs{z},\omega)$ denote the unique leaf of a tree grown by Algorithm \ref{alg:cap} that contains $z$ and $M_j(\bs{z})$ is the total
number of splits along $Z_j$ to form $R(\bs{z},\omega)$ for $j\in[d_Z]$. Then
\begin{itemize}
    \item[(a)] For $\delta\in(0,1)$
    \begin{equation*}
    \P\left( M_{j}(\bs{z})\leq  \tfrac{(1-\delta)\pi}{d_Z} \frac{\log(s/(2k-1))}{\log(1/(1-\alpha))}\right)\leq \left(\frac{s}{2k-1}\right)^{-\frac{\delta^2\pi^2}{2d_Z^2\log(1/(1-\alpha))}}.
    \end{equation*}
    \item[(b)] For $\delta\in(0,1)$,
    \begin{equation}\label{eq:diameter_bound_pointwise}
        \P\left(\textnormal{diam}(R(\bs{z},\omega)) \gtrsim \left(\frac{k}{s}\right)^ {(1-\delta)\frac{K(\alpha)\pi}{d_Z}}\right)\lesssim  \left(\frac{k}{s}\right)^{\frac{\delta^2\pi^2}{2(1+o(1))^2d_Z^ 2\log (1/\alpha)}} + \frac{1}{s},
    \end{equation}
    where $K(\alpha):= \frac{\log ((1-\alpha)^{-1})}{\log(1/\alpha)}$ for $\alpha\in (0,0.5)$ and $\pi\in (0,1]$;
    \item[(c)] For $t>0$ and $1\leq  p\leq q<\infty$
    \begin{equation}\label{eq:diameter_bound_Lp}
        \P\left(\Big(\E_{\bs{Z}'} \Big[\textnormal{diam}(R(\bs{Z}',\omega))^p\Big]\Big)^{1/p}\gtrsim t\left(\tfrac{k}{s}\right)^{(1-\delta^*)\frac{K(\alpha)\pi}{d_Z}}\right)\lesssim \frac{1}{t^q} + \frac{1}{s},
    \end{equation}
    for some $\delta^*\in (0,1)$ defined by \eqref{def:delta_star}, where $\E_{\bs{Z}'}(\cdot)$ denotes the expectation with respect $\bs{Z}'$ which equals in distribution $\bs{Z}$ but in independent of everything else;
    
    \item[(d)] $\P\left(\sup_{bs{z}\in[0,1]^d} \textnormal{diam}(R(\bs{z},\omega))\geq 1\right)\to 1$ for $d_Z\geq 2$  provided that $k=o(s)$ as $s\to\infty$.
\end{itemize}
\end{lemma}
\begin{remark}
    In the proof of part $(c)$ in Lemma \ref{lem:upper_bound_diameter}, we show that with probability approaching 1, there exists a leaf $R_\ell$ that is never split along the variable $Z_j$, which in turn means that the diameter of this leaf is at least 1 with probability approaching 1. Hence, a uniform convergence for a tree grown by Algorithm \ref{alg:cap} is not possible with high probability as long as it depends on the diameter to shrink to zero as $s\to\infty$.
\end{remark}
\begin{proof}[of Lemma \ref{lem:upper_bound_diameter}]
Let  $M(\bs{z}):=\sum_{j\in[d_Z]} M_j(\bs{z})$ denote the total number of (random) splits to form the leaf $R(\bs{z},\omega)$. Since the variable to split is chosen at random, we have 
\[
M_j(\bs{z})|M(\bs{z})\sim\text{Binomial}(M(\bs{z}),q_j);\qquad j\in[d_Z],\; bs{z}\in[0,1]^{d_Z},
\]
where $q_j$ is the probability of a split in $Z_j$ which is lower bounded by $\pi/d_Z>0$ by assumption.  

Let $N^\mathcal{B}(\bs{z})$ be the number of observations of the sample $\mathcal{B}$ on the $R(\bs{z},\omega)$. Then, by the minimum leaf size condition on the tree construction, we have 
\[
k\leq N^{\mathcal{B}}(\bs{z}) \leq 2k-1;\qquad  bs{z}\in[0,1]^{d_Z}.
\]
Also, $N^\mathcal{B}(\bs{z}) =s\prod_{\ell=1}^{M(\bs{z})}\widetilde{\alpha}_\ell$ for some sequence $\{\widetilde{\alpha}_\ell\}$ such that $\alpha\leq \widetilde{\alpha}_\ell\leq (1-\alpha)$ for $\ell\in[M(\bs{z})]$ by the $\alpha$-regularity condition. Thus 
\[
k\leq s\alpha^{M(\bs{z})}\leq s(1-\alpha)^{M(\bs{z})}\leq {2k-1};\qquad bs{z}\in[0,1]^{d_Z}.
\]
Hence, from the last expressions, we obtain a lower bound on the number of splits of any leaf given by
\begin{equation}\label{eq:num_splits_bounds}
    \underline{M}:=\frac{\log(s/(2k-1))}{\log(1/(1-\alpha))}\leq M(\bs{z});\qquad ;\qquad bs{z}\in[0,1]^{d_Z}.
\end{equation}

Recall the Chernoff's inequality. Let  $S_m = W_1 + \dots +W_m$ where $W_1,\dots, W_m$ are independent and  $W_j\in[0,1]$ for $j\in[m]$, then for every $t>0$,
\begin{equation}\label{eq:chernoff_inequality}
    \P(S \leq  \E[S] -t)\leq \exp(-t^2/2m).
\end{equation}
Use \eqref{eq:num_splits_bounds} twice combined with \eqref{eq:chernoff_inequality} to obtain for $\delta\in(0,1)$,
\begin{align*}
    \P\left( M_{j,\ell}\leq  \tfrac{(1-\delta)\pi}{d_Z} \underline{M}\right)&\leq \E\left[\P\left( M_{j,\ell}\leq  \tfrac{(1-\delta)\pi}{d_Z} M_\ell \Big| M_\ell\right)\right]\\
    &\leq\E\left[\exp\left(-\frac{\delta^2\pi^2 M_\ell}{2d_Z^2} \right)\right]\\
    &\leq \exp\left(-\frac{\delta^2\pi^2 \underline{M}}{2d_Z^2} \right)\\ 
    &= \left(\frac{s}{2k-1}\right)^{-\frac{\delta^2\pi^2}{2d_Z^2\log(1/(1-\alpha))}},
\end{align*}
which concludes the proof of part $(a)$.

For $(b)$, we use \eqref{eq:diameter_Lebesgue_measure_relation} to write that on $\mathcal{E}_j$ and for $x>0$,
   \begin{align*}
        \P\left(\text{diam}_j(R_\ell)\geq  x|\mathcal{E}_j\right)\leq \P\left(M_{j,\ell}\leq  \frac{-\log x}{(1+o(1))\log(1/(1-\alpha+o(1)))}\right).
    \end{align*}
Set $\bs{X}=\left(\tfrac{2k-1}{s}\right)^ {(1-\delta)\pi\frac{K(\alpha)}{d_Z}}$ and use the last inequality to write
\begin{align*}
    \P\left(\text{diam}_j(R_\ell) \geq \left(\tfrac{2k-1}{s}\right)^ {(1-\delta)\pi\frac{K(\alpha)}{d_Z}}\Big|\mathcal{E}_j\right)&\leq\P\left(M_{j,\ell} \leq \frac{(1-\delta)\pi\underline{M}}{\I_1d_Z}\Big|\mathcal{E}_j\right)\\
    &=\E\left[\P\left(M_{j,\ell} \leq \frac{(1-\delta)\pi\underline{M}}{\I_1d_Z}|\mathcal{E}_j, M_\ell\right)\right]\\
    &=\E\left[\P\left(M_{j,\ell} \leq \frac{(1-\delta)\pi\underline{M}}{\I_1d_Z}| M_\ell\right)\right]\\
    &\leq \left(\frac{s}{2k-1}\right)^{-\frac{\delta^2\pi^2}{2\I_1(\alpha)^2d_Z^ 2\log (1/\alpha)}},
\end{align*}
and from the union bound
\begin{align*}
    \P\left(\text{diam}_jR((\bs{z},\omega)) \geq \left(\tfrac{s}{2k-1}\right)^ {(1-\delta)\frac{K(\alpha,\pi)}{d_Z}}\right)\leq  \left(\frac{s}{2k-1}\right)^{-\frac{\delta^2\pi^2}{2\I_1(\alpha)^2d_Z^ 2\log (1/\alpha)}} + \frac{1}{s}.
\end{align*}
Therefore, $(b)$ follows from the union bound and the last expression because
\begin{align*}
    \P\left(\text{diam}(R(\bs{z},\omega)) \geq \sqrt{d_Z}\left(\tfrac{s}{2k-1}\right)^{(1-\delta)\frac{K(\alpha,\pi)}{d_Z}}\right) &\leq \P\left(\max_{j\in[d_Z]}\text{diam}_j(R(\bs{z},\omega)) \geq \sqrt{d_Z}\left(\tfrac{s}{2k-1}\right)^{(1-\delta)\frac{K(\alpha,\pi)}{d_Z}}\right)\\
    &\leq d_Z\max_{j\in[d_Z]}\P\left(\sup_{bs{z}\in[0,1]^d}\text{diam}_j(R(\bs{z},\omega)) \geq \sqrt{d_Z}\left(\tfrac{s}{2k-1}\right)^{(1-\delta)\frac{K(\alpha,\pi)}{d_Z}}\right).
\end{align*}

For $(c)$, first note that, for $q\geq 1$ and $r>0$,
\begin{align*}
    \E\big[\text{diam}(R(\bs{z},\omega))^q\big] &= \int_0^1\P\big[\text{diam}(R(\bs{z},\omega))\geq t^{1/q} \big]dt\\
    &= \int_0^{r^q}\P\big[\text{diam}(R(\bs{z},\omega))\geq t^{1/q} \big]dt + \int_{r^q}^1\P\big[\text{diam}(R(\bs{z},\omega))\geq t^{1/q} \big]dt\\
    &\leq r^q +  \P\big[\text{diam}(R(\bs{z},\omega))\geq r \big].
\end{align*}
Set $r=\sqrt{d_Z}\left(\tfrac{2k-1}{s}\right)^{(1-\delta)\frac{K(\alpha)\pi}{d_Z}}$ and use \eqref{eq:diameter_bound_pointwise} to write
    \begin{align*}
    \E\big[\text{diam}(R(\bs{z},\omega))^q|\mathcal{E}\big] &\lesssim \left(\frac{k}{s}\right)^{q(1-\delta)\frac{K(\alpha)\pi}{d_Z}} +  \left(\frac{k}{s}\right)^{\frac{\delta^2\pi^2}{2\I_1(\alpha)^2d_Z^ 2\log (1/\alpha)}}.
    \end{align*}
Note that $\delta\mapsto q(1-\delta)\frac{K(\alpha)\pi}{d_Z}$ is decreasing and vanishes as $\delta\uparrow 1$ and $\delta\mapsto \frac{\delta^2\pi^2}{2\I_1(\alpha)^2d_Z^ 2\log (1/\alpha)}$ is non-negative increasing, so there is a $\delta^*\in(0,1)$ function of $q,\alpha,\pi, d_Z$ such that that those two functions agree,i.e.,
\begin{equation}\label{def:delta_star}
    q(1-\delta^*)\frac{K(\alpha))}{d_Z}=\frac{{\delta^ *}^2\pi^2}{2\I_1(\alpha)^2d_Z^ 2\log (1/\alpha)}.
\end{equation}
Therefore
\begin{align*}
    \E\big[\text{diam}(R(\bs{z},\omega))^q|\mathcal{E}\big] &\lesssim \left(\frac{k}{s}\right)^{q(1-\delta^*)\frac{K(\alpha)\pi}{d_Z}}. 
\end{align*}

Set temporarily $D_p:=\Big(\E_{\bs{Z}'} \Big[\textnormal{diam}(\bs{Z}',\omega)^p\Big]\Big)^{1/p}$. For $1\leq p\leq q<\infty$ and $x>0$, we have, by Markov's inequality followed by Jensen's inequality and Fubini theorem,
        \begin{align*}
            \P(D_p\geq x|\mathcal{E}) = \P(D_p^q\geq \bs{X}^q|\mathcal{E})&\leq \frac{\E\big[(D^p_p)^{q/p}|\mathcal{E}\big]}{\bs{X}^q}\\
            &\leq \frac{\E\Big[\E_{\bs{Z}'} \Big[\text{diam}(R(\bs{Z}',\omega))^q\Big]|\mathcal{E}\Big]}{\bs{X}^q} = \frac{\E_{\bs{Z}'}\Big[\E \Big[\text{diam}(R(\bs{Z}',\omega))^q|\mathcal{E}\Big]\Big]}{\bs{X}^ q}.
        \end{align*}
          Set $x  = t\sqrt{d_Z}\left(\tfrac{2k-1}{s}\right)^{(1-\delta^*)\frac{K(\alpha)\pi}{d_Z}}$ for $t >0$ to conclude that for  $1\leq p\leq q<\infty$,
    \[
        \P\left(D_p\gtrsim t\left(\tfrac{k}{s}\right)^{(1-\delta^*)\frac{K(\alpha)\pi}{d_Z}}\right)\leq \P(D_p\geq x|\mathcal{E}) + \P(\mathcal{E}^c)\lesssim \frac{1}{t^q} + \frac{d_Z}{s}.
        \]
which completes the proof of $(c)$.

Finally for $(d)$, let $M_j(\bs{z},\omega)$ denote the number of splits in $R(\bs{z},\omega) $ along the variable $j$. Now, if $\inf_{bs{z}\in[0,1]^d} M_j(\bs{z},\omega)=\min_{\ell\in[L]} M_{j,\ell}=0$ for some $j\in[d_Z]$ then $\sup_{bs{z}\in[0,1]^d} \text{diam}(R(\bs{z},\omega))\geq 1$ thus
\begin{align*}
    \P\left(\sup_{bs{z}\in[0,1]^d} \text{diam}(R(\bs{z},\omega))\geq 1\right)\geq \P\left(\inf_{bs{z}\in[0,1]^d} M_{j,\ell}=0\right)
\end{align*}
Also, note that $M_{j,\ell}$ and $M_{j,\ell'}$ are independent conditonally on $(M_\ell,M_{\ell'})$ for $\ell\neq \ell'$. 
\begin{align*}
    \P\left(\min_{\ell\in[L]} M_{j,\ell}=0\right) &= 1 - \prod_{\ell=1}^L \big[1-\P(M_{j,\ell}=0)\big]\\
    &= 1 - \prod_{\ell=1}^L \big[1-(1-q_j)^{M_\ell}\big]\\
    &\geq  1 - \big[1-(1-\pi/d_Z)^{\overline{M}}\big]^L\\
    &= 1 - \left[1 - (s/k)^{\frac{\log(1-\pi/d_z)}{\log 1/\alpha}}\right]^L
\end{align*}
Since $L\leq s/2k$ let $\bs\beta := \frac{\log 1/\alpha}{\log(1/(1-\pi/d_z))} $ and note that since $1-\pi/d_Z\geq 1 - 1/dz>\alpha$ for $d_Z\geq 2$ we have that $\bs\beta>1$. Set $v=(s/k)^{1/\bs\beta}\to\infty$ then
\[
\liminf_{s\to\infty} \P\left(\min_{\ell\in[L]} M_{j,\ell}=0\right)
\geq 1 - \lim_{s\to\infty}(1-(s/k)^{-1/{\bs\beta}})^{s/(2k)} = 1 - \lim_{v\to\infty}(1-1/v)^{v^{\bs\beta}/2}
\]
Since $1-x\leq \exp(-x)$ we have
\[
\liminf_{s\to\infty} \P\left(\min_{\ell\in[L]} M_{j,\ell}=0\right)
\geq  1 - \lim_{v\to\infty}\exp(-v^{(\bs\beta-1)/2}) = 1,
\]
which completes the proof of part $(d)$.
\end{proof}

\begin{lemma}[Lower bounds on the number of observations in the leaves]\label{lem:lower_bound} 
Let $L$ denote the number of leaves in a tree  constructed according to Algorithm \ref{alg:cap} and
$N:=(N_1,\dots, N_L)$ where $N_\ell$ is the number of observation of the sample $\mathcal{A}$  the $\ell$-th leave for $\ell\in [L]$. Then for $k\in\N$

\begin{itemize}
    \item[(a)] $\liminf_{s\to\infty}\P(N_\ell =0)\geq  e^{-2\overline{f}(2k-1)}$ for $\ell\in[L]$;
    \item[(b)] $\lim_{s\to\infty}\P(\min_{\ell\in[L]} N_\ell =0) =1$;
\end{itemize}
Also,
\begin{itemize}
    \item[(c)]  If $k\gtrsim s^\epsilon$ for some  $\epsilon\in(0,1)$ then $k(k/s)^\epsilon\lesssim_\P N_\ell\lesssim_\P k(k/s)^{-\epsilon}$ for $\ell\in[L]$;
    \item[(d)] If $k\gtrsim s^{1/2+\epsilon}$ for some  $\epsilon\in(0,1/2)$ then $\min_{\ell\in[L]} N_\ell\gtrsim_\P k(k/s)^\epsilon$.
\end{itemize}

Let  $K(\alpha):= \frac{\log ((1-\alpha)^{-1})}{\log(1/\alpha)}$ for $\alpha\in (0,0.5)$ and $\pi\in (0,1]$. For $k\in[s]$ and $\delta\in(0,1)$
and
\begin{align}
    \P\left(|\mathcal{A}(\bs{z},\omega)| \leq (1-\delta)\underline{f} s\left(\frac{s}{k}\right)^{-\frac{\I_2(\alpha)}{K(\alpha)}}\right)&\leq \frac{1}{\delta^2\underline{f} s(s/k)^{-\frac{\I_2(\alpha)}{K(\alpha)}}}\label{eq:obs_lower_bound};\\
    \P\left(|\mathcal{A}(\bs{z},\omega)|\geq (1+\delta) s \overline{f}\left(\frac{s}{2k-1}\right)^{-\I_1(\alpha)K(\alpha)}\right)&\leq \frac{1+\delta}{\delta^2s\underline{f} \left(\frac{s}{k}\right)^{-\frac{\I_2(\alpha)}{K(\alpha)}}}\label{eq:obs_upper_bound},
\end{align}
where $\underline{f},\overline{f}$ are the lower and upper bound on the density of $\bs{Z}$ respectively, $\I_1(\alpha):=(1+o(1))\frac{\log(1/(1-\alpha +o(1)))}{\log(1/(1-\alpha))}\to 1$, $\I_2(\alpha):=(1+o(1))\frac{\log(1/(\alpha +o(1)))}{\log(1/\alpha)}\to 1$.

In particular, if $s\to\infty$ and $k\asymp s^{\eta}$ for $\eta\in (1- K(\alpha,\pi), 1)\subseteq (0,1)$ then the right-hand side of  \eqref{eq:obs_lower_bound}-\eqref{eq:obs_upper_bound} vanishes and we conclude
\[
\textnormal{diam}(R(\bs{z},\omega)) \lesssim_\P s^{-\frac{ (1-\eta)\pi K(\alpha)}{2d_Z}}\quad\text{and}\quad  s^{1-\frac{1-\eta}{K(\alpha)}}\lesssim_\P|\mathcal{A}(\bs{z},\omega)|\lesssim_
\P s^{1-(1-\eta)K(\alpha)}.
\]
\end{lemma}
\begin{remark}
    From parts (a) and (b) in Lemma \ref{lem:lower_bound}, we conclude that every leaf from a tree grown by Algorithm \ref{alg:cap} with a minimum leaf parameter $k\in\N$ is empty with probability at least $e^{-2(2k-1)}$ as $s\to\infty$. Therefore, to ensure that the leaves have a minimum number of observations with high probability, we need to have $k\to\infty$ as $s\to\infty$. In fact, by ignoring the $o(1)$ terms in \eqref{eq:diameter_Lebesgue_measure_relation}, we have that $N_\ell\asymp k$ as $k\to\infty$ and $s\to\infty$.
\end{remark}
\begin{proof}[of Lemma \ref{lem:lower_bound}]

Let $L$ denote the number of leaves in a tree and
$N^\mathcal{A}:=(N^\mathcal{A}_1,\dots, N^\mathcal{A}_L)$ where $N^\mathcal{A}_\ell$ is the number of observation of the sample $\mathcal{A}$ on the $\ell$-th leave for $\ell\in [L]$. Also let $p=(p_1,\dots,p_L)$ where $p_\ell:=\P(\bs{Z}\in R_\ell)$. Then
\[
N^\mathcal{A}\sim\text{Multinomial}_L(s,p);\qquad N^\mathcal{A}_\ell \sim\text{Binomial}(s,p_\ell)\quad \ell\in[L].
\]
Let $F_{n',p'}$ denote the cdf of a Binomial distribution with $n'$ trials and probability of success $p'$. Then 
\[
\P(N_\ell^ {\mathcal{A}}\leq x) = F_{s,p_\ell}(\bs{x});\quad  x\in\R,\; \ell\in[L].
\]
Also, since the density of $\bs{Z}$ is bounded away from zero and infinity, we have that  
\[
\underline{f}\prod_{j=1}^{d_Z} \text{diam}_j (R_\ell)=\underline{f}\mu(R_\ell)\leq p_\ell\leq \overline{f} \mu(R_\ell)= \overline{f}\prod_{j=1}^{d_Z} \text{diam}_j (R_\ell);\qquad \ell\in[L].
\]

Let $M_j(\bs{z},\omega)$ denote the number of splits in $Z_j$ forming $R(\bs{z},\omega)$ for $j\in[d_Z]$. By $\alpha$-regularity we have, from Lemma 12 and 13 in \cite{WW2016}, that $\P(\mathcal{E}_j)\geq 1-1/s$ for $j\in[d_Z]$ where
\begin{equation}\label{eq:diameter_Lebesgue_measure_relation}
    \mathcal{E}_j:=\big\{(\alpha+o(1))^{(1+o(1))M_j(\bs{z},\omega)}\leq \text{diam}_j(R(\bs{z},\omega))\leq (1-\alpha+o(1))^{(1+o(1))M_j(\bs{z},\omega)}:bs{z}\in[0,1]^d\big\}.
\end{equation}
Then, on $\mathcal{E}:=\bigcap_j \mathcal{E}_j$,
\[
\alpha^{(1+\zeta_1) M}\leq \mu(R_\ell)\leq (1-\alpha)^{(1+\zeta_2) M}
\]
where $M:=M(\bs{z},\omega):=\sum_{j=1}^{d_Z} M_j(\bs{z},\omega)$ is the total number of splits forming $R(\bs{z},\omega)$ and
\[
\zeta_1 := (1+o(1))\frac{\log(\alpha + o(1))}{\log(\alpha)}-1 = o(1);\qquad \zeta_2 := (1+o(1))\frac{\log(1-\alpha + o(1))}{\log(1-\alpha)}-1 = o(1).
\]
Also, by $\alpha$-regularity
\[
k\leq s\alpha^M\leq s(1-\alpha)^M\leq 2k-1.
\]
Then, on $\mathcal{E}$ which occurs with probability at least $1-d_Z/s$,
\begin{equation}\label{eq:leaves_prob_bounds}
    \underline{f}(k/s)^{1+\zeta_1}\leq p_\ell\leq \overline{f} ((2k-1)/s)^{1+\zeta_2};\qquad \ell\in[L].
\end{equation}

Recall that the Binomial distribution is decreasing in $p$, in the sense that $F_{n,p}\leq F_{n,p'}$ pointwise for $p'\leq p$. Then, using the inequality above
\[
F_{s,\frac{\overline{f}(2k-1)((2k-1)/s)^{o(1)}}{s}}  \leq F_{s,p_\ell}\leq F_{s,\frac{\underline{f}k(k/s)^ {o(1)}}{s}}.
\]
Let $G_{\bs\Lambda}$ denote the Poisson cdf with mean $\lambda$. Note that $0\leq k/s\leq 1$, then if $\underline{\bs\Lambda}:=\liminf_{s\to\infty } (k/s)^{o(1)} >0$ we have, by the pointwise limit $F_{n',p'}\to G_{\lambda'}$ as $n'p'\to\lambda'$,
\[
\limsup_{s\to\infty}F_{s,\frac{\underline{f}k(k/s)^ {o(1)}}{s}}  \leq G_{\underline{\bs\Lambda} \underline{f} k}.
\]
Othewise if $\underline{\bs\Lambda} = 0$ we get a trivial bound $\limsup_{s\to\infty}F_{s,\frac{\underline{f}k(k/s)^ {o(1)}}{s}}\leq 1$. Similarly, we have $0\leq (2k -1)/s\leq 2$ and if $\overline{\bs\Lambda}:=\limsup_{s\to\infty } ((2k-1)/s)^{o(1)} >0$ we have 
\[
\liminf_{s\to\infty}F_{s,\frac{\overline{f}(2k-1)((2k-1)/s)^{o(1)}}{s}}    \geq G_{\overline{\bs\Lambda}\overline{f} (2k-1)}.
\]
Othewise when  $\overline{\bs\Lambda} = 0$ we get a \emph{non} trivial bound $\liminf_{s\to\infty}F_{s,\frac{\overline{f}(2k-1)((2k-1)/s)^{o(1)}}{s}}\geq  1$. Therefore, defining pointwise $G_0(m):=\lim_{\lambda\downarrow 0} G_\lambda(m) = \1\{m\geq 0\} $, we have shown that for $k\in\N$ and $\ell\in[L]$:
\[
G_{\overline{f}\overline{\bs\Lambda}(2k-1)}(\cdot)\leq \liminf_{s\to\infty}\P(N_\ell\leq \cdot) \leq \limsup_{s\to\infty}\P(N_\ell\leq \cdot)\leq G_{\underline{f}\underline{\bs\Lambda}k}(\cdot).
\]
We then obtain the result $(a)$ by evaluating the last inequality at $0$ and using the fact that $\overline{\bs\Lambda}\leq 2$ to conclude 
\begin{equation}\label{eq:prob_empty_leaf}
    \liminf_{s\to\infty}\P(N_\ell=0)\geq  G_{\overline{f}\overline{\bs\Lambda}(2k-1)}(0)\geq  G_{2\overline{f}(2k-1)}(0)= \frac{1}{e^ {2\overline{f}(2k-1)}}.
\end{equation}

For $(b)$, suppose that there are only two leaves $L=2$ indexed by $\ell,\ell'$, then
\begin{align*}
    \P(N_\ell\geq k,N_{\ell'} \geq k)&=\P(N_\ell\geq k|N_{\ell'} \geq k)\P(N_{\ell'} \geq k)\\
    &=\P(N_\ell\geq k|N_{\ell} \leq s-k)\P(N_{\ell'} \geq k)\\
    &=\P(k\leq N_\ell\leq s-k)\P(N_{\ell} \leq s-k)\P(N_{\ell'} \geq k)\\
    &\leq \P( N_\ell\geq k)\P(N_{\ell'} \geq k),
\end{align*}
where we use the fact that $N_\ell + N_{\ell'}=s$. Applying induction, it is easy to verify that 
\begin{equation}\label{eq:multinomial_inequality}
    \P(N_1\geq k,\dots, N_{L} \geq k)\leq\prod_{l=1}^L\P(N_\ell\geq k);\quad k\in\N_0.
\end{equation}
Let $N_{\min}:=\min_{\ell\in[L]} N_\ell$. Use \eqref{eq:multinomial_inequality} with $k=1$, the fact that $1-x\leq e^ {-x}$ for $x\in\R$ and part $(a)$  to write for $k\in\N$
\begin{align*}
    \P(N_{\min} = 0) &= 1 - \P(N_1\geq  1,\dots, N_{L} \geq 1)\\
    &\geq 1 - \prod_{l=1}^L\P(N_\ell\geq 1)\\
    &= 1 - \prod_{l=1}^L\big[1- \P(N_\ell=0)\big]\\
    &\geq 1 - \exp\big[-\sum_{l=1}^L \P(N_\ell=0)\big]\\
    &\geq  1 - \exp\big[-L(e^{-2\overline{f}(2k-1)} + o(1) )\big]\to 1,
\end{align*}
because $L\asymp s/k\to\infty$ for each $k\in\N$ as $s\to\infty$, which concludes the proof of part $(b)$. Note that if $k\geq \log s/4$ then $\P(N_{\min}=0) \not\to 1$.

For $(c)$, we have by Markov's inequality $N_\ell \lesssim_\P \E[N_\ell] = sp_\ell$ hence $N_\ell\lesssim_\P sp_\ell$. For the other direction, we apply tail bounds for the binomial distribution. Refer to \cite[pg 151]{feller1957}. Let $S_n\sim\text{Binom}(m,p)$, then
\begin{align}
    \P(S\leq r)&\leq \frac{(m-r)p}{(mp-r)^2};\qquad  r\leq mp \label{eq:binomial_lower_bound}\\
    \P(S\geq  r)&\leq \frac{r(1-p)}{(mp-r)^2};\qquad r\geq mp\label{eq:binomial_upper_bound}.
\end{align}

Then using  \eqref{eq:binomial_lower_bound}, we write for some $C>1$
\[
\P(N_\ell\leq sp_\ell/C)\leq \frac{(s-sp_\ell/C)p_\ell}{(1-1/C)^ 2(sp_\ell)^2}\leq \frac{sp_\ell}{(1-1/C)^ 2(sp_\ell)^2}=\frac{1}{(1-1/C)^ 2sp_\ell}.
\]
Also, from \eqref{eq:leaves_prob_bounds} we have for $\epsilon>0$
\begin{equation}\label{eq:sp_bounds}
    k\left(\frac{k}{s}\right)^{\epsilon}\lesssim k\left(\frac{k}{s}\right)^{\zeta_1}\lesssim \underline{f}s\left(\frac{k}{s}\right)^{1+\zeta_1}\leq sp_\ell\leq \overline{f} s\left(\frac{2k-1}{s}\right)^{1+\zeta_2}\lesssim   k\left(\frac{2k-1}{s}\right)^{\zeta_2}\lesssim  k\left(\frac{k}{s}\right)^{-\epsilon}
\end{equation}
because $\zeta_1=o(1)$ and $\zeta_2=o(1)$. By assumption $k\gtrsim s^\epsilon$  then the left-hand side 
is at least of order $k^\epsilon\to\infty$ as $s\to\infty$ which ensures that $sp_\ell\to\infty$. Then $N_\ell\gtrsim_\P sp_\ell$ and $N\asymp_\P sp_\ell$. Therefore, from \eqref{eq:leaves_prob_bounds} we have for $\epsilon>0$
\[
k\left(\frac{k}{s}\right)^{\epsilon}\lesssim_\P N_\ell\lesssim_\P  k\left(\frac{k}{s}\right)^{-\epsilon},
\]
which concludes the proof of part $(c)$.

For $(d)$, by the union bound, followed by \eqref{eq:binomial_lower_bound}, we have for $m<sp_{\min}$ where $p_{\min}:=\min_{\ell\in[L]}p_\ell$,
\begin{equation*}
    \P(N_{\min}\leq m) =  \P\left(\bigcup_{\ell=1}^L \{N_\ell\leq  m\}\right)\leq \sum_{\ell=1}^ L\P(X_\ell\leq  m)\leq \sum_{\ell=1}^L\frac{(s-m)p_\ell}{(sp_\ell - m)^2}
    \leq \frac{(s-m)}{(sp_{\min} - m)^2}.
\end{equation*}
Set $m=sp_{\min}/C$ for a constant  $C>1$ and use \eqref{eq:sp_bounds} to write
\begin{equation*}
    \P(N_{\min}\leq sp_{\min}/C) \leq \frac{s-sp_{\min}/C}{(1-1/C)^2sp_{\min}}\leq\frac{s}{(1-1/C)^2(sp_{\min})^2}\lesssim \frac{1}{(1-1/C)^2K^{2\epsilon}},
\end{equation*}
where we use the assumption that $k\gtrsim s^{1/2+\epsilon}$. Then $N_{\min}\gtrsim_\P sp_{\min}\gtrsim_\P k(k/s)^\epsilon$ which demostrate part $(d)$ and concludes the proof of the lemma.

\end{proof}

\begin{lemma}[Tree Variance-type bound]\label{lem:tree_variance_rate} Let $\{W_i:i\in[n]\}$ be $n$ independent copies of the random variable $W$ and define the tree $T_W(\bs{z},\omega):=\frac{1}{|\cA(\bs{z},\omega)|}\sum_{i\in\cA (\bs{z},\omega)} W_i$ and the map $V(\bs{z}):= T_W(\bs{z},\omega) - \E[W|\bs{Z} = z]$ for $\bs{z}\in[0,1]^{d_Z}$. If $\bs{z}\mapsto \E[|W|^q|\bs{Z}=\bs{z}]\leq M_q<\infty$ for some $q\geq 2$ and constant $M_q\geq 0$ and $k\gtrsim s^\epsilon$ for some  $\epsilon\in(0,1)$  then for $t>0$, as $s\to\infty$,
\[
\P\Big(|V(\bs{z})|\gtrsim tk^{-\frac{1}{2}}(s/k)^{\epsilon/2}\Big) \lesssim  t^{-q} + s^{-1},\qquad \bs{z}\in[0,1]^d.  
\] 
Also, let $\Delta_p:=\Big(\E_{\bs{Z}'} \Big[V(\bs{Z}')^p\Big]\Big)^{1/p}$ for $p\in[2,\infty)$ where $\bs{Z}'$ and independent copy of $\bs{Z}$ and $\Delta_\infty:=\sup_{bs{z}\in[0,1]^{d_Z}} V(\bs{z})$, then
\begin{align*}
    \P\Big(\Delta_p\gtrsim tk^{-\frac{1}{2}}(s/k)^ {\epsilon/2}\Big) &\lesssim  t^{-q} + s^{-1};\qquad 2\leq  p\leq q<\infty,\\
    \P\Big(\Delta_\infty\gtrsim tk^{-\frac{1}{2}}(s/k)^ {1/q + \epsilon/2}\Big) &\lesssim  t^{-q} + k^{-1}.
\end{align*}
\end{lemma}
\begin{proof}[Lemma \ref{lem:tree_variance_rate}]
It is convenient re-write $T_W(\bs{z},\omega)$ for a subsample of size $|\mathcal{S}|=:s\leq n$ as
\[
T_W(\bs{z},\omega) = \sum_{i=1}^s S_iW_{i}; \qquad S_i :=S_i(\bs{z},\omega):=\begin{cases}
    \left|\{j\in \mathcal{A}:Z_j\in R(\bs{z},\omega)\}\right|^{-1}& \text{if $\bs{Z}_i\in R(\bs{z},\omega)$} \\
    0 &\text{otherwise}.
\end{cases}
\]
First we note that $T_W(\bs{z},\omega)$ is unbiased for $\E[W|\bs{Z}\in R(\bs{z},\omega)]$. In fact, $T_W(\bs{z},\omega)$ is even conditional (on $\S:=(S_1,\dots, S_d)$) unbiased since
\begin{equation*}\label{eq:generic_tree_unbiasedness}
    \E[T_W(\bs{z},\omega)|\S] = \left[\sum_{i=1}^s S_i\E[W_i|S_i]\right] = \E[W|\bs{Z}\in R(\bs{z},\omega)]\sum_{i=1}^s S_i = \E[W|\bs{Z}\in R(\bs{z},\omega)],
\end{equation*}
where we use the fact that, due to the sample split, $S_i\E[W_i|S_i]=|\mathcal{A}(\bs{z},\omega)|^{-1}\E[W|\bs{Z}\in R(\bs{z},\omega)]$ if $\bs{Z}_i\in R(\bs{z},\omega)$ and $0$ otherwise, and $\sum_{i=1}^s S_i = 1$. Then $\E[T_W(\bs{z},\omega)] = \E_{\S}\big[\E[T_W(\bs{z},\omega)|\S]\big] = \E[W|\bs{Z}\in R(\bs{z},\omega)]$. 

For convenience define $U_i:= W_i-\E[W|\bs{Z}\in R(\bs{z},\omega)]$ for $i\in[s]$ and note that the sequence $\{S_iU_i:i\in[s]\}$ is zero mean and independent conditional on $\S$. Then, by Marcinkiewicz–Zygmund inequality followed by Jensen's inequality, we have, for $q\geq 2$,
    \begin{align*}
        \E\left[\left|\sqrt{|\mathcal{A}(\bs{z},\omega)|}\sum_{i=1}^s S_iU_i\right|^q\Big|\S\right] &\leq C_q \E\left[\left|\sum_{i=1}^s |\mathcal{A}(\bs{z},\omega)|S_i^2U_i^ 2\right|^{q/2}\Big|\S\right]\\
        &= C_q \E\left[\left|\sum_{i=1}^s S_iU_i^ 2\right|^{q/2}\Big|\S\right]\\
        &\leq C_q\left(\sum_{i=1}^s S_i^{q/2}\E[|U_i|^q|S_i]\right)\\
        &\leq C_q\left(\sum_{i=1}^s S_i \E[|U_i|^q|S_i]\right)=  C_q\E[|U|^q|\bs{Z}\in R(\bs{z},\omega)],
    \end{align*}
where $C_q$ is constant only depending on $q$. We also use $|\mathcal{A}(\bs{z},\omega)|^ {k-1}S_i^k = S_i$ and $S_i^k\leq S_i$ for $k\geq 1$. Hence
 \begin{align*}
        \E\left[\left|\sum_{i=1}^s S_iU_i\right|^q\Big|\S\right]&\lesssim \frac{1}{|\mathcal{A}(\bs{z},\omega)|^{q/2}\lor 1 }\E\left[\left|\sqrt{|\mathcal{A}(\bs{z},\omega)|}\sum_{i=1}^s S_iU_i\right|^q\Big|\S\right]\\
        &\lesssim\frac{\E[|U|^q|\bs{Z}\in R(\bs{z},\omega)]}{|\mathcal{A}(\bs{z},\omega)|^{q/2}\lor 1}.
    \end{align*}

From \cite{brazil2000} we have that
    \[
    \E[(1+B)^{-\alpha}]\lesssim (np)^{-\alpha};\qquad B\sim\text{Binomial}(n,p),\quad \alpha\in \R.
    \]
    Recall that $|\mathcal{A}(\bs{z},\omega)|\sim \text{Binomial}(s,p(\bs{z},\omega))$ conditional on $R(\bs{z},\omega)$ and $p(\bs{z},\omega)$ is bounded by below on $\mathcal{E}$ as per \eqref{eq:diameter_Lebesgue_measure_relation}. Then conditional on $R(\bs{z},\omega)$ and $\mathcal{E}$, for $\alpha\in R$, we have
    \[
    \E\left[\frac{1}{|\mathcal{A}(\bs{z},\omega)|^{q/2}\lor 1 }\right]\leq 2\E\left[\frac{1}{(1+ |\mathcal{A}(\bs{z},\omega)|)^{q/2}}\right]\lesssim \frac{1}{(sp(\bs{z},\omega))^{q/2}}\leq \frac{1}{(s\underline{f}(k/s)^{1+\zeta_1})^{q/2}}
    \]
    Therefore
    \begin{align*}
       \E\left[\left|\sum_{i=1}^s S_iU_i\right|^q\Big|\mathcal{E}\right] &=\E\left\{\E\left[\left|\sum_{i=1}^s S_iU_i\right|^q\Big |\S,\mathcal{E}\right]\Big|\mathcal{E} \right\}\\
       &\lesssim\E\left\{ \frac{\E[|U|^q|\bs{Z}\in R(\bs{z},\omega)]}{|\mathcal{A}(\bs{z},\omega)|^{q/2}\lor 1}\Big|\mathcal{E}\right\}\\
       &= \E\left\{\E[|U|^q|\bs{Z}\in R(\bs{z},\omega)]\E\left[ \frac{1}{|\mathcal{A}(\bs{z},\omega)|^{q/2}\lor 1}\Big|\mathcal{E},R(\bs{z},\omega)\right]\Big|\mathcal{E}\right\}\\
       &\lesssim\frac{\E[|U|^q|\bs{Z}\in R(\bs{z},\omega)|\mathcal{E}]}{(s\underline{f}(k/s)^{1+\zeta_1})^{q/2}} \\
       &\leq M_q\left(\frac{1}{(s\underline{f}(k/s)^{1+\zeta_1})^{q/2}}\right).
    \end{align*}

    Finally, by Markov inequality and the last bound, we have
    \[
        \P\left(V(\bs{z})\gtrsim \frac{t}{\sqrt{s(k/s)^{1+\zeta_1}}}\right)\leq \P(V(\bs{z})\geq x|\mathcal{E}) + \P(\mathcal{E}^c)\lesssim \frac{1}{t^q} + \frac{d_Z}{s},
        \]
    which demonstrates the first result.

    For the second result in the case $2\leq q<\infty$, we have, for $p\leq q$ and $x>0$, by Markov's inequality followed by Jensen's inequality and Fubini theorem.
        \[
        \P(\Delta_p\geq x|\mathcal{E}) = \P(\Delta_p^q\geq \bs{X}^q|\mathcal{E})\leq \frac{\E\big[(\Delta^p_p)^{q/p}|\mathcal{E}\big]}{\bs{X}^q}\leq \frac{\E\Big[\E_{\bs{Z}'} \Big[V(\bs{Z}')^q\Big]|\mathcal{E}\Big]}{\bs{X}^q} = \frac{\E_{\bs{Z}'}\Big[\E \Big[V(\bs{Z}')^q|\mathcal{E}\Big]\Big]}{\bs{X}^ q}.
        \]
    Set $x = t(s\underline{f}(k/s)^{1+\zeta_1})^{-1/2}$ for $t >0$ to conclude that for  $2\leq p\leq q<\infty$,
    \[
        \P\left(\Delta_p\gtrsim \frac{\delta}{\sqrt{s(k/s)^{1+\zeta_1}}}\right)\leq \P(\Delta_p\geq x|\mathcal{E}) + \P(\mathcal{E}^c)\lesssim \frac{1}{t^q} + \frac{d_Z}{s}.
        \]
For the case $p=\infty$, we have that the number of leaves is bounded  by $s/k$ and then by the union bound for $x = t(s/k)^{1/q}k^{-\frac{1}{2}}(s/k)^ { \epsilon/2}$
\begin{align*}
    \P\left(\sup_{bs{z}\in[0,1]^d} |V(\bs{z})|\geq x\right) &\leq \frac{s}{k}\sup_{bs{z}\in[0,1]^d}\P\left( V(\bs{z})\geq x\right)\lesssim \frac{s}{k}\left( \frac{1}{(s/k)t^q} + \frac{d_Z}{s}\right) =   \frac{1}{t^q} + \frac{d_Z}{k}
\end{align*}
\end{proof}

\begin{lemma}[Tree Bias-type bound]\label{lem:tree_bias_rate} Consider the same setup of Lemma \ref{lem:tree_variance_rate} and define $ B(\bs{z}):=\E[W|\bs{Z}\in R(\bs{z},\omega)] - \E[W|\bs{Z} = z]$ for $\bs{z}\in[0,1]^{d_Z}$. If $\bs{z}\mapsto \E[W|\bs{Z}=\bs{z}]$ is Lipschitz then for $\delta\in(0,1)$,
\begin{equation}\label{eq:pointwise_bias}
        \P\left(|B(\bs{z})| \gtrsim \left(\frac{k}{s}\right)^ {(1-\delta)\frac{K(\alpha)\pi}{d_Z}}\right)\lesssim  \left(\frac{k}{s}\right)^{\frac{\delta^2\pi^2}{2(1+o(1))^2d_Z^ 2\log (1/\alpha)}} + \frac{1}{s}.
\end{equation}
Also, let $\Pi_p:=\Big(\E_{\bs{Z}'} \Big[B(\bs{Z}')^p\Big]\Big)^{1/p}$ for $p\in[2,\infty)$ where $\bs{Z}'$ and independent copy of $\bs{Z}$ and  then
\begin{equation}\label{eq:bias_bound_Lp}
        \P\left(\Pi_p\gtrsim t\left(\tfrac{k}{s}\right)^{(1-\delta^*)\frac{K(\alpha)\pi}{d_Z}}\right)\lesssim \frac{1}{t^q} + \frac{1}{s};\qquad 1\leq  p\leq q<\infty.
    \end{equation}
\end{lemma}
\begin{proof}[Lemma \ref{lem:tree_bias_rate}]
From  the Lipschitz condition and Lemma \ref{lem:lower_bound}, we have as $s\to\infty$
\[
|B(\bs{z})|\lesssim  \sup_{\bs{z},\bs{z}'\in R(\bs{z},\omega)}\|\bs{z} - \bs{z}'\| =:\text{diam}(R(\bs{z},\omega))
\]
Use \eqref{eq:diameter_bound_pointwise} to upper bound the leaf diameter for each $\bs{z}\in[0,1]^d$ and obtain \eqref{eq:pointwise_bias}. Similarly, use \eqref{eq:diameter_bound_Lp} to upper bound the leaf diameter for each $\bs{z}\in[0,1]^d$ and obtain\eqref{eq:bias_bound_Lp}.
\end{proof}

\begin{lemma}[Tree Rate of Convergence]\label{lem:tree_covergence_rate} 
Let $\{W_i:i\in[n]\}$ be $n$ independent copies of the random variable $W$ and define the tree $T_W(\bs{z},\omega):=\frac{1}{|\cA(\bs{z},\omega)|}\sum_{i\in\cA (\bs{z},\omega)} W_i$. If $\bs{z}\mapsto T_0(\bs{z}):=\E[W|\bs{Z}=\bs{z}]$ is Lipschitz, $\bs{z}\mapsto \E[|W|^q|\bs{Z}=\bs{z}]\leq M_q<\infty$ for some $q\geq 2$ and constant $M_q\geq 0$ and $k\gtrsim s^\epsilon$ for some  $\epsilon\in(0,1)$  then for $\delta\in (0,1)$, as $s\to\infty$,
\[
|T_W(\bs{z},\omega) - T_0(\bs{z})|\lesssim_\P k^{-\frac{1}{2}}\left(\frac{s}{k}\right)^{\epsilon/2} + \left(\frac{k}{s}\right)^{(1-\delta)\frac{K(\alpha)\pi}{d_Z}};\qquad \bs{z}\in[0,1]^d.  
\] 
Also, for $2\leq p\leq q$,
\begin{align*}
    \left[\int_{[0,1]^d}|T_W(\bs{z},\omega) - T_0(\bs{z})|^pf(\bs{z}) dz \right]^{1/p} &\lesssim_\P k^{-\frac{1}{2}}\left(\frac{s}{k}\right)^{\epsilon/2} + \left(\frac{k}{s}\right)^{(1-\delta^*)\frac{K(\alpha)\pi}{d_Z}},
\end{align*}
where $\delta^*\in (0,1)$ is defined in Lemma \ref{lem:tree_bias_rate}.
\begin{proof}[Lemma \ref{lem:tree_covergence_rate}]
    By the triangle inequality combined with Lemma \ref{lem:tree_variance_rate} and \ref{lem:tree_bias_rate} we obtain the first result since
    \begin{align*}
        |T(\bs{z},\omega) - \E[W|\bs{Z}=\bs{z}]| &\leq|T(\bs{z},\omega) - \E[W|\bs{Z}\in R(\bs{z},\omega)]| + |\E[W|\bs{Z}\in R(\bs{z},\omega) - \E[W|\bs{Z}=\bs{z}]|\\
        &\lesssim_\P  k^{-\frac{1}{2}}\left(\frac{s}{k}\right)^{\epsilon/2} + \left(\frac{k}{s}\right)^{(1-\delta)\frac{K(\alpha)\pi}{d_Z}}.
    \end{align*}
Similarly, the second result follows from the triangle inequality, Lemma \ref{lem:tree_variance_rate} and \ref{lem:tree_bias_rate} with $p=q$.
\end{proof}

\end{lemma}

\begin{lemma}\label{lem:generic_RF_assymptotic_normality} Let $\{V_i:= (W_i,\bs{Z}_i):i\in[n]\}$ be $n$ independent copies of the random vector $V:=(W,\bs{Z})$ where $W$ is a random variable . If $\bs{z}\mapsto \E[W|\bs{Z}=\bs{z}]$ and $\bs{z}\mapsto \E[W^2|\bs{Z}=\bs{z}]$ are Lipschitz, $\V[W|\bs{Z}=\bs{z}]\geq 0$ and $k\asymp s^{\eta}$ for $\eta\in (1- K(\alpha,\pi), 1)\subseteq (0,1)$ then, if $s\to\infty$ and $s(\log n)^ {d_Z} = o(n)$
\begin{equation}\label{eq:generic_RF_assymptotic_normality}
    \frac{\overline{T}(\bs{z}) - \E[\overline{T}(\bs{z})]}{\lambda(\bs{z})}\cd \mathsf{N}(0,1)\quad\text{and}\quad  \frac{\V[\overline{T}(\bs{z})]}{\lambda^2(\bs{z})}\to 1,
\end{equation}
where $\overline{T}(\bs{z}) :=\frac{1}{B}\sum_{b=1}^ B T_W(\bs{z},\omega_b)$, $ $ $T_W(\bs{z},\omega_b):=\frac{1}{|\cA(\bs{z},\omega_b)|}\sum_{i\in\cA (\bs{z},\omega_b)} W_i$, $\{\omega_b:b\in[B]\}$ is independent of $\{V_i:i\in[n]\}$ and $\lambda^2(\bs{z})$ is the variance of the H\'ayek Projection of $T_W(\bs{z},\omega_b)$, which can be lower bounded as
\begin{equation}\label{eq:generic_RF_assymptotic_variance_bounds}
    \frac{s^\frac{1-\eta}{K(\alpha)}}{n
(\log s)^{d_Z}}\lesssim \lambda^2(\bs{z}).
\end{equation}
\end{lemma}
\begin{proof}[Lemma \ref{lem:generic_RF_assymptotic_normality}]
We start by writing $\overline{T}(\bs{z})$  as a generalized U-statistics following \cite{Peng2022}. For a fixed $\bs{z}\in[0,1]^d$ define the randomized symmetric (in its $s$ arguments) kernel $h$ by,
\[
h(v_1,\dots, v_s) = \sum_{i=1}^s\chi_i w_i;\qquad \chi_i(\bs{z},\xi):=\begin{cases}
    \left|\{j\in \mathcal{A}:z_j\in R(\bs{z},\omega)\}\right|^{-1}& \text{if $z_i\in R(\bs{z},\xi)$} \\
    0 &\text{otherwise}.
    \end{cases}
\]
where $\{\xi:i\in [(n,s)]\}$ are independent copies of $\xi$, which incorporates the randomness of generating the tree with subsample $\mathcal{B}$ and $\omega$.

Then  generalized U-statistic of order $s\in[n]$ is defined as
\[
\overline{T} = \binom{n}{s}^ {-1} \sum_{(n,s)} h(Z_{i_1},\dots, Z_{i_s},\omega),
\]
which have the H-decomposition expressed as
\begin{equation*}\label{eq:H-decompostion}
    \overline{T}-\E[\overline{T}] =\sum_{j=1}^ s U_j;\qquad U_j:=\binom{s}{j}\binom{n}{j}^ {-1} \sum_{(n,j)} h^{(j)}(\bs{X},V_{i_1},\dots, V_{i_s},\omega); \quad j\in[s],
\end{equation*}
where $U_j$'s are zero-mean and pairwise uncorrelated, and 
\begin{align*}
    h_i(v_1,\dots,v_i) &:= \E[h(v_1,\dots,v_i,V_{i+1},\dots, V_s,\xi)] - \E[h];\qquad i\in[s]\\
    h^{(i)}(v_1,\dots,v_i) &:= h_i(v_1,\dots,v_i) -\sum_{j=1}^ {i-1}\sum_{(s,j)}h^ {(j)}(v_{i_1},\dots,v_{i_j});\qquad i\in[s-1]\\
    h^{(s)}(v_1,\dots,v_i) &:= h(v_1,\dots,v_s,\xi) -\sum_{j=1}^ {s-1}\sum_{(s,j)}h^ {(j)}(v_{i_1},\dots,v_{i_j}).
\end{align*}
    In particular, $h^{(1)}(v) = \mathcal{H}_1(v)=\E[h(V_1,\dotsm V_s,\xi|V_1=v)] - \E[h]$ and
\[
U_1 = \frac{s}{n}\sum_{i=1}^n \mathcal{H}_1(V_i). 
\]
Note that the sequence $\{(s/n)\mathcal{H}_1(V_i):i\in[n]\}$ is centered and $i.i.d$ (for fixed $z$ and $n$ and $s$) hence then Linderberg condition is satisfied. Let $\sigma_1^2  :=\V[\mathcal{H}_1(Z_1)]$ then
\[
J_*:=\frac{\sqrt{n}U_1}{s\sigma_1}\cd \mathsf{N}(0,1)\quad\text{in distribution as}\quad n\to\infty.
\]
Consider the following decomposition.
\[
J:=\frac{\overline{T}-\E[\overline{T}]}{\sqrt{\V[\overline{T}]}} = J_* + (J-J_*).
\]
If the second term vanishes in probability, we have that the term on the left-hand side weakly converges to a standard Gaussian. The second term is zero-mean, and its variance equals $2 - 2\C[J,J_*]$. Then, by Theorem 11.2 in \cite{Vaart_1998}, it is sufficient for convergence in the second mean that 
\[
\rho:=\frac{\V[\overline{T}]}{\V[U_1]}\to 1 \quad\text{as}\quad n\to\infty.
\]
Since $\V[\overline{T}]\geq \V[U_1(\bs{x})]$, it suffices to state condition under which $\limsup_{s\to\infty} \rho \leq 1$. Moreover \cite{Peng2022} show that
\[
\rho\leq 1 + \frac{s}{n}\frac{\V[h]}{s\sigma_1}.
\]

If we apply Lemma \ref{lem:variance_ratio} with $m\asymp s^{1-\frac{1-\eta}{K(\alpha)}}$ then by Lemma \ref{lem:lower_bound} we have that $\P(\mathcal{G})\to 1$ and $m/s\to 0$. Therefore there is a constant $C_{f,d_Z}$ depending only  $f$, the density of $\bs{Z}$ (which is assumed to bounded away from $0$ and $\infty$ by Assumption \ref{ass:main}(a)), and $d_Z$ such that with probability approaching 1
\[
\limsup_{s\to\infty}\frac{C_{f_Z,d_Z}}{(\log s)^{d_Z}}\frac{\V[h]}{s\sigma_1}\lesssim 1.
\]
Hence, if $s\to\infty$ and $s(\log n)^ {d_Z} = o(n)$ as $n\to\infty$
\[
J\cd \mathsf{N}(0,1)\quad\text{and}\quad \sqrt{\rho}J\cd \mathsf{N}(0,1)\quad\text{in distribution},
\]
and the proof of \eqref{eq:generic_RF_assymptotic_normality} is complete.

Also, from Lemma \ref{lem:variance_ratio} we have that  
\[
\sigma_1^2=\V[\E[T_W(x;z)|V_1]\gtrsim \V\big[\E[S_1|V_1]\big]\V[W|\bs{Z}=\bs{z}]\gtrsim \frac{C_{f,d_Z}\V[W|\bs{Z}=\bs{z}]}{sm(\log s)^d}\gtrsim \frac{1}{s(\log s)^d} 
\]
then
\[
\lambda^2:=\frac{s^2}{n}\sigma_1^2 = \left(\frac{s}{n}\right)^ 2 n\sigma_1^2\gtrsim\frac{s}{n} \frac{1}{s^{1-\frac{1-\eta}{K(\alpha)}}(\log s)^d} =\frac{s^\frac{1-\eta}{K(\alpha)}}{n
(\log s)^{d_Z}}.
\]
Moreover, there exist a constant $C$ (depending on $C_{f_Z,d_Z}$ and $\underline{\sigma}^2:=\min_{j\in d_X} \V[W|\bs{Z}]>0$ by Assumption \ref{ass:main}(d)) such that
\[
\sigma_1^2(\bs{z})=\V[\E[T_W(x;\bs{Z})|V_1]\gtrsim \V\big[\E[S_1|V_1]\big]\V[W|\bs{Z}=\bs{z}]\gtrsim \frac{C_{f,d_Z}\V[W|\bs{Z}=\bs{z}]}{sk(\log s)^d}\gtrsim \frac{1}{sk(\log s)^d}
\]
Therefore, there are constants $C_1,C_2$ such that 
\[
C_1 s/(nk(\log s)^d)\leq \sigma_1(\bs{x})\leq C_2(s/n)
\]

Therefore we might take $\Lambda(\bs{z})=\V[\E[T(\bs{z},\omega)|(Y_1,X_1,Z_1) = (\bs{X},y,x)]]$ to conclude that
\[
C_1'  \frac{s^{1-\eta}}{n(\log s)^d}\leq \|\Lambda(\bs{z})\|\leq C_2' (s/n) \to 0
\]
hence $\|\Lambda(\bs{z})\|\cp 0$ as $n\to\infty$ and the proof of $(i)$ is  complete.

Now, recall
\begin{align*}
    \mathcal{H}_1(w,z;z')&:=\E[h(V_1,\dots, V_s,\xi|V_1=v)] - \E[h]\\
    &= \E\left[\sum_{i=1}^s W_i S_i(z',\zeta)\Big| W_1=w,Z_1=\bs{z}\right]- \E[W|\bs{Z}\in R(z',\omega)]\\
    &= w\E[S_1(z',\omega)|Z_1=\bs{z}] - \E[W|\bs{Z}\in R(z',\omega)]\\
    &= w\E[S(z',\omega)|\bs{Z}=\bs{z}] - \E[W|\bs{Z}\in R(z',\omega)].
\end{align*}
Note that when $z', z''\in[0,1]^d$ belong to the same leaf $R_\ell$ then $(w,z)\mapsto \mathcal{H}_1(w,z;z')$ equals $(w,z)\mapsto \mathcal{H}_1(w,z;z'')$ so we define $\kappa_\ell(w,z):=w\E[S^{(\ell)}|\bs{Z}=\bs{z}] - \E[W|\bs{Z}\in R_\ell] $ where $S^{(\ell)}$ is defined as above with $R(z',\omega)$ replaced by $R_\ell$ for $l\in[L]$. Finally, we define
\[
S = \frac{1}{\sqrt{n}}\sum_{i=1}^n Y_i;\qquad Y_i:=(Y_{i,1},\dots,Y_{i,L})^ {\T};\qquad Y_{i,\ell} :=\kappa_\ell(W_i,\bs{Z}_i);\quad \ell\in[L].
\]
So $S$ is a sum of centered iid random vectors with covariance structure $\Gamma:=[\sigma_{\ell,\ell'}]_{\ell,\ell' \in[L]}$
\[\
\gamma_{\ell,\ell'}:=\E[Y_{1,\ell},Y_{1,\ell'}]=\E[\kappa_\ell(W,\bs{Z})\kappa_{\ell'}(W,\bs{Z})] = \C\big[W\E[S^{(\ell)}|\bs{Z}],W\E[S^{(\ell' )}|\bs{Z}]\big].
\]
We need an upper bound for
\[
\E\big[W^2\E[S^{(\ell)}|\bs{Z}]\E[S^{(\ell' )}|\bs{Z}]\big]
\]
Define $S_G:=\sum_{i=1}^n G_i$ where $G_i\sim N(0,\Gamma)$ are iid and $S_G^*$ such that $S_G^*|\text{data}\sim N(0,\widehat{\Gamma})$
\[
\widetilde{\Gamma} = \frac{1}{n}\sum_{i=1}^n (Y_i - \overline{Y})(Y_i - \overline{Y})^{\T};\qquad \overline{Y}:=\frac{1}{n}\sum_{i=1}^nY_i.
\]
By the triangle inequality
\begin{align*}
    \sup_{t\in\R}|\P(\|S\|\leq t) - \P(\|S_G^*\|\leq t|D)|&\leq \sup_{t\in\R}|\P(\|S\|\leq t)
    - \P(\|S_G\|\leq t)|\\
    &\qquad + \sup_{t\in\R}|\P(\|S_G\|\leq t) - \P(\|S_G^*\|\leq t|D)|
\end{align*}

The first term can be bounded by 
\[
\frac{\mu_3(\log L)^2}{\sqrt{n}},\qquad \mu_3:=\E[Y^{\T}\Gamma^ {-1}Y\|Y\|_2]\leq \lambda(\Gamma_{\min})\E[\|Y\|^3]
\]

\end{proof}

\begin{lemma}\label{lem:variance_ratio} Let $V_i = (Y_i,\bs{X}_i,\bs{Z}_i)$ for $i\in[n]$ and define the event $ \mathcal{G}:=\mathcal{G}(m,z)=\{m\leq |A(\omega,z)|\}$. Then for $m\in [s]$ and $ bs{z}\in[0,1]^d$ there is for constant $C_{f,d_Z}$ depending only on $f$ and $d_Z$ such that
\[
    \V\big[\E[S_1|V_1]\big]\gtrsim  \frac{C_{f,d_Z}}{ms(\log s)^{d_Z}} \P(\mathcal{G}) - 1/s^2.
\]
Furthermore, if $\V[W|\bs{Z}=\bs{z}]\leq M<\infty$ and $P(\mathcal{G})\to 1$ and $m/s\to 0$ as $s\to\infty$ then
\[
    \frac{\V[T_W(\bs{z},\omega)]}{s\V\big[\E[Z_1|V_1]}\lesssim_\P \frac{M(\log s)^{d_Z} }{ C_{f,d_Z}}.
\]
If further $\bs{z}\mapsto \E[W|\bs{Z}=\bs{z}]$ and $\bs{z}\mapsto \E[W^2|\bs{Z}=\bs{z}]$ are Lipschitz, and $\V[W|\bs{Z}=\bs{z}]\geq 0$ then
\[
    \frac{\V[T_W(\bs{z},\omega)]}{\V\big[\E[T_W(\bs{z},\omega)|V_1]}\lesssim_\P \frac{M(\log s)^{d_Z} }{ C_{f,d_Z}\V[W|\bs{Z}=\bs{z}]}.
\]
\end{lemma}
\begin{proof}[Lemma \ref{lem:variance_ratio}]
    Recall that $\bs{Z}_i$ is a $m$-potential nearest neighbor (PNN) of $z$ if there is an axis-aligned rectangle containing \emph{only} $z$ and a subset of $Z_1,\dots Z_s$ containing $\bs{Z}_i$ with size between $m$ and $2m-1$ and nothing else. Define
    \[
    P_i:=P_i(m,z):=\1\{\bs{Z}_i\text{ is a $m$-PNN of $z$}\};\quad bs{z}\in[0,1]^d.
    \]
    Note that $P_i=0$ implies $S_i=0$ conditional on $\mathcal{G}$ and then
    \[
    \E[S_1|Z_1,\mathcal{G}] \leq \frac{1}{m}\E[P_1|Z_1 ].
    \]
    Expression (33) in the proof of Lemma 3.2 in \cite{athey_wager2018} gives us.
    \[
    \P\left(\E[P_1|V_1 ]\geq \frac{1}{s^2}\right)\lesssim m\frac{2^ {d+1}(\log s)^d}{(d-1)!s},
    \]
    which implies that
    \[
    \E\big[(\E[S_1|Z_1,\mathcal{G}])^ 2\big]\gtrsim\frac{C_{f,d_Z}}{ms(\log s)^{d_Z}}.
    \]
    Therefore,
    \begin{align*}
        \V\big[\E[S_1|V_1]\big] &= \E\big[(\E[S_1|V_1])^2\big] -  \big(\E[S_1]\big)^2\\
         &=  \E\big[(\E[S_1|V_1])^2\big] - 1/s^2\\
         &=  \E\Big[\big(\E[S_1|Z_1,\mathcal{G}]\P(\mathcal{G}) + \E[S_1|Z_1,\mathcal{G}^c]\P(\mathcal{G}^c)\big)^ 2\Big]- 1/s^2\\
         &\gtrsim  \frac{C_{f,d_Z}}{ms(\log s)^{d_Z}} \P(\mathcal{G}) - 1/s^2.
    \end{align*}

    Now on $\mathcal{G}$, we have
    \[
    m\V[T_W(\bs{z},\omega)]\leq |\mathcal{A}(\omega,z)|\V[T_W(\bs{z},\omega)] =\E\big[|\mathcal{A}(\omega,z)|\V[T_W(\bs{z},\omega)|\S]\big] =\E\big[\V[W|\bs{Z}=\bs{z}]\big] \leq M.
    \]
    Therefore, on $\mathcal{G}$ such that $\P(\mathcal{G})\to 1$ and $m/s\to 0$ we have with probability approaching 1
    \[
    \frac{\V[T_W(\bs{z},\omega)]}{s\V\big[\E[S_1|V_1]}\lesssim
    \frac{M}{\frac{C_{f,d_Z}}{(\log s)^{d_Z}}\P(\mathcal{G})- m/s}\to
    \frac{M(\log s)^{d_Z}}{C_{f,d_Z}}. 
    \]
\end{proof}

\end{document}